\documentclass{article}
\usepackage[utf8]{inputenc}
\usepackage[T2A]{fontenc}
\usepackage[english]{babel}
\inputencoding{utf8}

\usepackage{amsmath,amssymb,amsthm, nicefrac}

\usepackage[
    pdfpagemode=UseNone,
    bookmarks=true,
    bookmarksopen=true,
    pdfpagelayout=SinglePage,
    pdfstartview={XYZ null null 1},
    pdfhighlight=/P,
    colorlinks=true,
    linkcolor=blue,
    citecolor=blue,
    urlcolor=blue
    ]{hyperref}

\usepackage{bookmark}
\bookmarksetup{
  open,
  numbered,
  depth=5, 
}
\usepackage{blkarray}
\usepackage{float}

\advance\textwidth 20mm  \advance\oddsidemargin -10mm
\advance\textheight 20mm \advance\topmargin -15mm

\def\C{\mathbb{C}}

\def\Z{\mathbb{Z}}

\let\ds\displaystyle

\theoremstyle{plain}
\newtheorem{theorem}{Theorem}
\newtheorem{lemma}[theorem]{Lemma}
\newtheorem{proposition}{Proposition}

\newtheorem{definition}{Definition}
\newtheorem{remark}{Remark}
\newtheorem{example}{Example}

\DeclareMathOperator{\const}{const}

\newcommand{\brackets}[1]{\left( #1 \right)}

\DeclareMathOperator{\PI}{P_{1}}
\DeclareMathOperator{\PII}{P_{2}}
\DeclareMathOperator{\PIII}{P_{3}}
\DeclareMathOperator{\PIIIpr}{P_{3}^{\prime}}
\DeclareMathOperator{\PIV}{P_{4}}
\DeclareMathOperator{\PV}{P_{5}}
\DeclareMathOperator{\PVI}{P_{6}}

\newcommand{\PVIn}[1]{\text{P}_6^{ #1 }}
\newcommand{\PVn}[1]{\text{P}_5^{ #1 }}
\newcommand{\PIVn}[1]{\text{P}_4^{ #1 }}

\newcommand{\PIIIprn}[1]{\text{P}_3^{\prime \, #1 }}
\newcommand{\PIIn}[1]{\text{P}_2^{ #1 }}
\newcommand{\PIn}[1]{\text{P}_1^{ #1 }}

\newcommand{\PPVIn}[1]{\phantom{}_{ #1 }\text{P}_6}
\newcommand{\PPVn}[1]{\phantom{}_{ #1 }\text{P}_5}
\newcommand{\PPIVn}[1]{\phantom{}_{ #1 }\text{P}_4}

\newcommand{\PPIIIprn}[1]{\phantom{}_{ #1 }\text{P}_3^{\prime}}
\newcommand{\PPIIn}[1]{\phantom{}_{ #1 }\text{P}_2}

\newcommand{\Painleve}{Painlev{\'e} }
\newcommand{\PPainleve}{Painlev{\'e}}

\usepackage{mathtools}
\mathtoolsset{showonlyrefs,showmanualtags}

\title{Non-abelian Painlevé systems with~generalized~Okamoto~integral}

\date{}

\author{I.A. Bobrova\thanks{National Research University Higher School of Economics, Moscow, Russian Federation.},~ V.V. Sokolov\thanks{L.D.~Landau Institute for Theoretical Physics, Chernogolovka, Russian Federation. 
E-mail: vsokolov@landau.ac.ru}}

\usepackage{tikz}

\begin{document}
\maketitle

\begin{abstract}
We study non-abelian systems of \Painleve type. To derive them, we introduce an auxiliary autonomous system with the frozen independent variable and postulate its integrability in the sense of the existence of a non-abelian first integral that generalizes the Okamoto Hamiltonian.
All non-abelian $\PVI-\PII$-systems with such integrals are found. A coalescence limiting  scheme is constructed for these non-abelian \Painleve systems. This allows us to construct an isomonodromic Lax pair for each of them.

\medskip

\end{abstract}

\section{Introduction}
In Okamoto's paper \cite{okamoto1980polynomial} all \Painleve equations $\PI-\PVI$ have been written as a polynomial Hamiltonian systems. In particular, the sixth \Painleve equation takes the form
\begin{gather} \label{eq:scalP6sys}
\hspace{-6mm}
   \left\{
   \begin{array}{lcl}
        z(z - 1) \, \dfrac{d u}{d z} 
        &=& 2 u^3 v 
        - 2 u^2 v 
        - \kappa_1 u^2 
        + \kappa_2 u
        + z \brackets{
        - 2 u^2 v 
        + 2 u v 
        + \kappa_4 u 
        + (\kappa_1 - \kappa_2 - \kappa_4)
        },
        \\[3mm]
        z(z - 1) \, \dfrac{d v}{d z}
        &=& - 3 u^2 v^2 + 2 u v^2 + 2 \kappa_1 u v - \kappa_2 v + \kappa_3
        + z \brackets{
     2 u v^2 - v^2 - \kappa_4 v
     },
\end{array}
\right.
\hspace{-8mm}
\end{gather}
where  $\kappa_1$, $\kappa_2$, $\kappa_3$, $\kappa_4$ are arbitrary constants and $u(z)$, $v(z)$, $z \in \mathbb{C}$. Eliminating $v$, one can obtain the $\PVI$-equation for $u(z).$
The Hamiltonian $H$ for \eqref{eq:scalP6sys} is given by 
\begin{equation}\label{scalHam}
    \begin{aligned}
    z (z - 1) H
    = u^3 v^2
    - u^2 v^2
    - \kappa_1 u^2 v
    + \kappa_2 u v
    - \kappa_3 u
    + z \left(
    - u^2 v^2
    + u v^2
    + \kappa_4 u v
    + (\kappa_1 - \kappa_2 - \kappa_4) v
    \right).
    \end{aligned}
\end{equation}
Since system \eqref{eq:scalP6sys} is non-autonomous, the function $H$ is not an integral of motion. 

System \eqref{eq:scalP6sys} has the form
\begin{gather} \label{eq:fullansatz}
    \left\{
    \begin{array}{lcl}
         f(z) \, \displaystyle \frac{d u}{d z} 
         &=& P_1 (u, v) + z \, Q_1 (u, v),
         \\[3mm]
         f(z) \, \displaystyle \frac{d v}{d z}
         &=& P_2 (u, v) + z \, Q_2 (u, v)
         ,
    \end{array}
    \right.
\end{gather}
while the Hamiltonian has the following structure:  $f(z) H = H_1 + z \, H_2,$
where $P_i, Q_i, H_i$ are polynomials in $u,v.$

Let us consider the system \begin{gather} \label{eq:noz}
    \left\{
    \begin{array}{lcl}
         \displaystyle \frac{d u}{d t} 
         &=& P_1 (u, v) + z \, Q_1 (u, v),
         \\[3mm]
         \displaystyle \frac{d v}{d t}
         &=& P_2 (u, v) + z \, Q_2 (u, v)
         ,
    \end{array}
    \right.
\end{gather}
where we regard $z$ as a parameter. We call \eqref{eq:noz} {\it the auxiliary autonomous system} for \eqref{eq:fullansatz}. 
 From the fact that \eqref{eq:fullansatz} is a Hamiltonian system with the Hamiltonian $H$ it follows that 
\begin{equation} \label{eq:matintansatz}
    J
    = H_1 (u, v) + z \, H_2 (u, v)
\end{equation}
is an integral of motion for system \eqref{eq:noz} i.e.   $\frac{d J}{d t}=0$.  We call $J$ the {\it Okamoto integral}. 

In Section \ref{section2} we consider systems of the form \eqref{eq:noz}, where $P_i$ and $Q_i$ are non-commutative polynomials 
given by
\begin{align} \label{eq:P1form}
\begin{aligned}
    &\begin{aligned}
    P_1 (u, v)
    = a_1 u^3 v 
    + a_2 u^2 v u 
    + a_3 u v u^2 
    + (2 - a_1 - a_2 - a_3) v u^3
    + c_1 u^2 v 
    \\
    + \, (- 2 - c_1 - c_2) u v u 
    + c_2 v u^2
    - \kappa_1 u^2 
    + \kappa_2 u,
    \end{aligned}
    \\[1mm]
    &\begin{aligned}
    Q_1 (u, v)
    = f_1 u^2 v 
    + (- 2 - f_1 - f_2) u v u 
    + f_2 v u^2
    + h_1 u v 
    + (2 - h_1) v u 
    + \kappa_4 u 
    \\
    + \, (\kappa_1 - \kappa_2 - \kappa_4),
    \end{aligned}
\end{aligned}
\end{align}
\begin{align} \label{eq:Q2form}
\begin{aligned}
    &\begin{aligned}
    P_2 (u, v)
    = b_1 u^2 v^2 
    + b_2 u v u v 
    + b_3 u v^2 u 
    + b_4 v u^2 v
    + b_5 v u v u
    + \brackets{- 3 - \sum b_i} v^2 u^2
    \\
    + \, d_1 u v^2
    + (2 - d_1 - d_2) v u v 
    + d_2 v^2 u 
    + e_1 v u 
    + (2 \kappa_1 - e_1) u v
    - \kappa_2 v 
    + \kappa_3,
    \end{aligned}
    \\[1mm]
    &\begin{aligned}
    Q_2 (u, v)
    = g_1 u v^2
    + (2 - g_1 - g_2) v u v 
    + g_2 v^2 u 
    - v^2 - \kappa_4 v
    \end{aligned}
\end{aligned}
\end{align}
and $\kappa_i$ are arbitrary constants.
In this paper we assume that all coefficients are complex numbers. If $f(z)=z (z-1),$ then the corresponding system \eqref{eq:fullansatz} is a natural non-commutative generalization of the \PPainleve-6 system \eqref{eq:scalP6sys}. To obtain the ansatz \eqref{eq:P1form},  \eqref{eq:Q2form}, we replace each monomial $M$ in \eqref{eq:scalP6sys} by a sum of non-commutative monomials such that this sum coincides with $M$ under the commutative reduction.

The following example of non-abelian Hamiltonian \PPainleve-6 system was found in \cite{Kawakami_2015}:
\begin{gather}
    \label{eq:hamP6case1}
    \tag*{$\PVIn{H}$}
    \left\{
    \begin{array}{lcr}
         z (z - 1) u'
         &=& u^2 v u 
        + u v u^2 
        - 2 u v u
        - \kappa_1 u^2 
        + \kappa_2 u
        \hspace{3.4cm}
        \\[1mm]
        && 
        + \, z \brackets{
        - u^2 v 
        - v u^2
        + u v + v u
        + \kappa_4 u 
         + (\kappa_1 - \kappa_2 - \kappa_4)
         },
         \\[2mm]
         z (z - 1) v'
         &=& - u v u v
        - v u v u
        - v u^2 v
        + 2 v u v
        + \kappa_1 u v
        + \kappa_1 v u
        - \kappa_2 v
        + \kappa_3
        \\[1mm]
        && 
        + \, z \brackets{
        u v^2 
        + v^2 u
        - v^2
        - \kappa_4 v
         }.
    \end{array}
    \right.
\end{gather}

In general, the classification problem we have in mind is to find all sets of coefficients in \eqref{eq:P1form}, \eqref{eq:Q2form} such that the corresponding system \eqref{eq:fullansatz} is integrable.  Consideration of all known non-abelian systems of $\PI-\PVI$ \, type \cite{Balandin_Sokolov_1998,Retakh_Rubtsov_2010, boalch2012simply, Kawakami_2015,Adler_Sokolov_2020_1,Bobrova_Sokolov_2021_1} shows that in all cases the corresponding auxiliary system is integrable in one sense or another.
This observation led us to the idea of formulating  an  integrability criterion for non-abelian systems of $\PI-\PVI$ in terms of the auxiliary system. 

In this paper to find interesting examples of non-abelian systems \eqref{eq:fullansatz} we postulate  for the corresponding system \eqref{eq:noz} the existence of a non-abelian integral of motion of the form \eqref{eq:matintansatz}.
 
In the case of $\PVI$ systems the polynomials $H_i$ have the form
\begin{align}
    \label{eq:matintansatzI2}
    \begin{aligned}
    &\begin{aligned}
    \hspace{-2mm}
    H_1 (u, v)
    = p_1 u^3 v^2
    + p_2 u^2 v u v
    + p_3 u^2 v^2 u
    + p_4 u v u^2 v
    + p_5 u v u v u 
    + p_6 u v^2 u^2 
    + p_7 v u^3 v
    \\
    + p_8 v u^2 v u 
    + p_9 v u v u^2
    + \brackets{1 - \sum p_i} v^2 u^3
    + q_1 u^2 v^2 
    + q_2 u v u v 
    + q_3 u v^2 u
    \\
    + q_4 v u^2 v
    + q_5 v u v u
    + \brackets{- 1 - \sum q_i} v^2 u^2
    + r_1 u^2 v 
    + r_2 u v u 
    \\
    + \brackets{- \kappa_1 - \sum r_i} v u^2
    + s_1 u v 
    + (\kappa_2 - s_1) v u
    - \kappa_3 u,
    \end{aligned}
    \\[2mm]
    &\begin{aligned}
    \hspace{-4mm}
    H_2 (u, v)
    = 
    t_1 u^2 v^2 
    + t_2 u v u v 
    + t_3 u v^2 u
    + t_4 v u^2 v
    + t_5 v u v u 
    + \brackets{- 1 - \sum t_i} v^2 u^2
    + x_1 u v^2 
    \\
    + x_2 v u v
    + \brackets{1 - \sum x_i} v^2 u
    + y_1 u v 
    + (\kappa_4 - y_1) v u
    + (\kappa_1 - \kappa_2 - \kappa_4) v.
    \end{aligned}
    \end{aligned}
\end{align}
In the commutative case these polynomials coincide with 
\begin{equation}\label{scalOkam}
    H_1
    =  u^3 v^2
    - u^2 v^2
    - \kappa_1 u^2 v
    + \kappa_2 u v
    - \kappa_3 u, 
    \qquad 
    H_2
    =
    - u^2 v^2
    + u v^2
    + \kappa_4 u v
    + (\kappa_1 - \kappa_2 - \kappa_4) v,
\end{equation} corresponding to \eqref{eq:scalP6sys}. Notice that we don't assume that 
in the non-abelian case the system \eqref{eq:noz} is Hamiltonian.

As a result, we found 18 non-abelian systems \eqref{eq:fullansatz} of  \PPainleve-6 type (see Appendix \ref{sec:sysintlistP6}). А transformation group acts on the set of these systems. There are 3 orbits of the group action and three non-equivalent systems corresponding to these orbits. Other systems can be derived from these three systems by applying the transformations. 

All these systems turn out to be not Hamiltonian and therefore our approach cannot  reconstruct the system \ref{eq:hamP6case1}.

To justify the integrability of obtained systems,
we find the isomonodromic Lax representations of the form
\begin{gather} \label{eq:zerocurvcond}
    \mathbf{A}_z - \mathbf{B}_{\lambda}
    = [\mathbf{B}, \mathbf{A}]
\end{gather}
for them.
 
The scalar equations $\PV - \PI$ are also equivalent to polynomial Hamiltonian systems of the form \eqref{eq:fullansatz}\footnote{Except for some degenerations of the $\PIII$ equation. For their non-abelian generalizations see Appendix~\ref{P3'D7}.} with $f(z)=z$ or $f(z)=1$.
In Section \ref{sec:P5_P1systems} we find all systems of $\PV$, $\PIV$, $\PIIIpr$, and $\PII$ types that have Okamoto integrals and obtain 10, 6, 8, and 2 non-abelian systems, respectively. Such non-abelian systems of type $\PI$ do not exist. 

Although the  system $\PIIn{2}$ found in \cite{Adler_Sokolov_2020_1} as well as six systems of $\PIV$ type from \cite{Bobrova_Sokolov_2021_1} are missing in our paper, we find interesting new non-abelian \Painleve type systems.

In Section \ref{sec:deg} we extend the scheme 
\begin{figure}[H]
    \centering
    \scalebox{1.1}{\tikzset{every picture/.style={line width=0.75pt}} 

\begin{tikzpicture}[x=0.75pt,y=0.75pt,yscale=-1,xscale=1]

\draw    (32.65,36.5) -- (59.88,36.5) ;
\draw [shift={(61.88,36.5)}, rotate = 180] [color={rgb, 255:red, 0; green, 0; blue, 0 }  ][line width=0.75]    (10.93,-3.29) .. controls (6.95,-1.4) and (3.31,-0.3) .. (0,0) .. controls (3.31,0.3) and (6.95,1.4) .. (10.93,3.29)   ;
\draw    (99.3,24.46) -- (128.8,13.68) ;
\draw [shift={(130.67,12.99)}, rotate = 159.93] [color={rgb, 255:red, 0; green, 0; blue, 0 }  ][line width=0.75]    (10.93,-3.29) .. controls (6.95,-1.4) and (3.31,-0.3) .. (0,0) .. controls (3.31,0.3) and (6.95,1.4) .. (10.93,3.29)   ;
\draw    (164.3,64.46) -- (193.8,51.68) ;
\draw [shift={(195.67,50.5)}, rotate = 152.93] [color={rgb, 255:red, 0; green, 0; blue, 0 }  ][line width=0.75]    (10.93,-3.29) .. controls (6.95,-1.4) and (3.31,-0.3) .. (0,0) .. controls (3.31,0.3) and (6.95,1.4) .. (10.93,3.29)   ;
\draw    (99.3,49.46) -- (128.72,63.72) ;
\draw [shift={(130.52,64.59)}, rotate = 205.86] [color={rgb, 255:red, 0; green, 0; blue, 0 }  ][line width=0.75]    (10.93,-3.29) .. controls (6.95,-1.4) and (3.31,-0.3) .. (0,0) .. controls (3.31,0.3) and (6.95,1.4) .. (10.93,3.29)   ;
\draw    (165.3,11.46) -- (194.72,25.72) ;
\draw [shift={(196.52,26.59)}, rotate = 205.86] [color={rgb, 255:red, 0; green, 0; blue, 0 }  ][line width=0.75]    (10.93,-3.29) .. controls (6.95,-1.4) and (3.31,-0.3) .. (0,0) .. controls (3.31,0.3) and (6.95,1.4) .. (10.93,3.29)   ;
\draw    (229.65,36.5) -- (256.88,36.5) ;
\draw [shift={(258.88,36.5)}, rotate = 180] [color={rgb, 255:red, 0; green, 0; blue, 0 }  ][line width=0.75]    (10.93,-3.29) .. controls (6.95,-1.4) and (3.31,-0.3) .. (0,0) .. controls (3.31,0.3) and (6.95,1.4) .. (10.93,3.29)   ;

\draw (7.02,29.79) node [anchor=north west][inner sep=0.75pt]    
{$\PVI$};
\draw (73.02,29.8) node [anchor=north west][inner sep=0.75pt]    
{$\PV$};
\draw (138.8,4.04) node [anchor=north west][inner sep=0.75pt]   
{$\PIIIpr$};
\draw (138.8,57.4) node [anchor=north west][inner sep=0.75pt]    
{$\PIV$};
\draw (202.77,29.8) node [anchor=north west][inner sep=0.75pt]    
{$\PII$};
\draw (270.02,29.8) node [anchor=north west][inner sep=0.75pt]    
{$\PI$};

\end{tikzpicture}}
    \caption{
    The degeneration scheme of the \Painleve equations
    \cite{gambier1910equations}
    }
\end{figure}
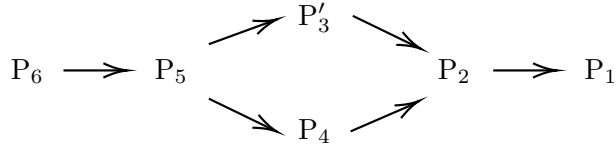
\hspace{-0.51cm}of degenerations of the scalar \Painleve equations \cite{gambier1910equations} and their isomonodromic Lax pairs to the non-abelian case and show that all 26 systems of $\PV - \PII$ type found is Section \ref{sec:P5_P1systems} and their Lax representations can be obtained by limiting transitions from the systems of $\PVI$-type described in Section \ref{section2} and Appendix \ref{sec:sysintlistP6}. In addition, we somewhat unexpectedly obtain Hamiltonian non-abelian $\PV - \PI$ systems  \cite{Kawakami_2015,boalch2012simply} (see Appendix \ref{sec:syshamlist}).
These systems do not have the Okamoto integral. The reason for the appearance of these systems is that sometimes the Okamoto integral degenerates into the integral $I = u v - v u$ that all Hamiltonian systems have (see Lemma~\ref{thm:hamint}). 

In Appendix \ref{sec:sysintlist} we present explicit list of all $\PVI - \PII$ type systems that have the Okamoto integral. In Appendix \ref{P3'D7} we consider degenerate non-abelian  $\PIIIpr$ systems,  which we call systems of $\PIIIpr(D_7)$ type. Although the corresponding scalar system is not polynomial, its non-abelian counterparts admitting Okamoto integrals were easily found.

\subsection{Non-abelian ODE systems}
 
In this paper we consider non-abelian systems of the form 
\begin{gather} \label{geneqm}
    \left\{
    \begin{array}{lcl}
       \displaystyle \frac{d u}{d t} 
         &=&F_1(u, v),
         \\[3mm]
         \displaystyle \frac{d v}{d t}
         &=&F_2(u, v)
         ,
    \end{array}
    \right.
\end{gather}
where $u$ and $v$  are generators of the free associative algebra $ {\cal A} $  over $ \C $ with the unity $\bf 1$\footnote{For any $k\in \C$ we often write $k$ instead of $k {\bf 1}$.}, and $F_{\alpha}\in  {\cal A}$. Actually, \eqref{geneqm}  is a notation for the derivation $d_t$ of  $ {\cal A} $ such that $d_t (u) = F_1, \quad d_t (v) = F_2 $. The element $ d_t (f) $ is uniquely determined for any element $f \in {\cal A} $ by the  Leibniz identity.

Usually, the first integrals of a system \eqref{geneqm} are some elements of the quotient vector space 
${\cal A} / [{\cal A}, \, {\cal A}]$. They are a formalization of integrals of the form $trace(h(u,v))$ in the matrix case $u(t),\, v(t)\in Mat_m$.  For the Hamiltonian non-abelian systems the Hamiltonians are first integrals of this kind.

The Hamiltonian systems \eqref{geneqm} have the form 
\begin{gather} \label{Ham}
    \left\{
    \begin{array}{lcl}
       \displaystyle \frac{d u}{d t} 
         &=&\displaystyle \frac{\partial H}{\partial v},
         \\[3mm]
         \displaystyle \frac{d v}{d t}
         &=&\displaystyle -\frac{\partial H}{\partial u}
         ,
    \end{array}
    \right.
\end{gather}
where $H\in {\cal A}$ and $\frac{\partial}{\partial u},\, \frac{\partial}{\partial v}$ are non-abelian derivatives \text{\rm(}see \text{\rm\cite{kontsevich1993formal})}.
For any polynomial $f(u,v)\in~{\cal A}$  these derivatives  are defined by the identity
 $$df=\frac{\partial f}{\partial u} du+\frac{\partial f}{\partial v} dv, $$
where the additional non-abelian variables $du$ and $dv$ are supposed to be moved to the right by the cyclic permutations of generators in monomials.   
Notice that  in the non-abelian case the partial derivatives are not vector fields.

\begin{remark}
It is easy to verify that $\frac{\partial}{\partial u} (a b- b a)= \frac{\partial}{\partial v} (a b- b a)=0$ for any $a,b\in  {\cal A}$ and therefore the non-abelian partial derivatives are well-defined maps $ {\cal A} / [{\cal A}, \, {\cal A}] \to  {\cal A}$. For this reason, we can assume that $H\in {\cal A} / [{\cal A}, \, {\cal A}] $ in the formula \eqref{Ham}.
\end{remark}
 
\begin{example} 
The auxiliary system for \text{\rm\ref{eq:hamP6case1}} is Hamiltonian with
\begin{gather}
    H = H_1 (u, v) + z \, H_2 (u, v),
    \\[2mm]
    \begin{aligned}
    H_1
    &=  u^2 v u v
    - u v u v
    - \kappa_1 u^2 v
    + \kappa_2 u v 
    - \kappa_3 u
    ,
    &
    H_2
    &= - u^2 v^2
    + u v^2
    + \kappa_4 u v
    + (\kappa_1-\kappa_2-\kappa_4) v
    .
    \end{aligned}
\end{gather}
Indeed, $d H = d H_1 + z \, d H_2$, where
\begin{align}
    \begin{array}{rcr}
    d H_1
    &=& du \, u v u v 
    + u \, du \, v u v 
    + u^2 \, dv \, u v 
    + u^2 v \, du \, v 
    + u^2 v u \, dv
    - du \, v u v 
    - u \, dv \, u v 
    - u v \, du \, v
    \\[1mm]
    && 
    - \, u v u \, dv
    - \kappa_1 (du \, u v + u \, du \, v + u^2 \, dv)
    + \kappa_2 (du \, v + u \, dv)
    - \kappa_3 du,
    \\[2mm]
    d H_2
    &=& - du \, u v^2 
    - u \, du \, v^2 
    - u^2 \, dv \, v
    - u^2 v \, dv
    + du \, v^2
    + u \, dv \, v
    + u v \, dv
    + \kappa_4 (du \, v + u \, dv)
    \\[1mm]
    && 
    + \, (\kappa_1 - \kappa_2 - \kappa_4) dv.
    \end{array}
\end{align}
Making cyclic permutations in all monomials to bring $du$ and $dv$ to the right, we obtain
\begin{align}
    \begin{array}{rcr}
    d H_1
    &=& \brackets{
    u v u v + v u v u + v u^2 v
    - 2 v u v
    - \kappa_1 u v - \kappa_1 v u
    + \kappa_2 v
    - \kappa_3
    } du
    \\[1mm]
    &&
    + \, \brackets{
    u v u^2 + u^2 v u 
    - 2 u v u 
    - \kappa_1 u^2 + \kappa_2 u
    } dv,
    \\[2mm]
    d H_2
    &=& \brackets{
    - u v^2 - v^2 u 
    + v^2 
    + \kappa_4 v
    } du
    + \left(
    - v u^2 - u^2 v
    + v u + u v 
    + \kappa_4 u 
    \right.
    \\[1mm]
    &&
    \left.
    \phantom{u^2}
    + \, (\kappa_1 - \kappa_2 - \kappa_4)
    \right) dv.
    \end{array}
\end{align}
Therefore,
\begin{align}
    \dfrac{\partial H}{\partial v}
    &= u^2 v u + u v u^2
    - 2 u v u 
    - \kappa_1 u^2 + \kappa_2 u
    + z \brackets{
    - u^2 v - v u^2 
    + u v + v u
    + \kappa_4 u 
    + (\kappa_1 - \kappa_2 - \kappa_4)
    }
    ,
    \\
    - \dfrac{\partial H}{\partial u}
    &= - u v u v - v u v u - v u^2 v
    + 2 v u v
    + \kappa_1 u v + \kappa_1 v u
    - \kappa_2 v
    + \kappa_3
    + z \brackets{
    u v^2 + v^2 u 
    - v^2 
    - \kappa_4 v
    }
\end{align}
and \eqref{Ham} coincides with the auxiliary system for \text{\rm\ref{eq:hamP6case1}}.
\end{example}

In this paper we are dealing with {\it non-abelian first integrals}, which are elements of $\cal A$ but not traces from $ {\cal A} / [{\cal A}, \, {\cal A}]$ . 
  
\begin{definition}\label{Def1} An element $h\in {\cal A}$ are called a {non-abelian first integral} for system \eqref{geneqm} if $d_t(h)=0.$\footnote{Notice that the Hamiltonian $H$ of a system \eqref{Ham} is not a first integral in the sense of Definition \ref{Def1}.}
\end{definition}

For non-abelian Hamiltonian systems with two variables $u$ and $v$ a special integral $I=u v-v u$ appears in the following statement:
\begin{lemma} \label{thm:hamint}
Any system of the form \eqref{Ham}
has the non-abelian first integral $I=u v - v u$.
\end{lemma}
\begin{proof}
It follows from the following well-known identity \cite{kontsevich1993formal} for partial derivatives:
$$
\left[u, \frac{\partial f}{\partial u}\right]+\left[v, \frac{\partial f}{\partial v}\right]=0, \qquad f\in{\cal A}.
$$
 \end{proof}

In our paper we assume that the auxiliary system \eqref{eq:noz} for \PPainleve \, type system \eqref{eq:fullansatz}  has a non-abelian Okamoto first integral (see Introduction) of the 
form \eqref{eq:matintansatz}.  
Both in the auxiliary system and in the Okamoto integral, the variable $z$ plays the role of an arbitrary parameter.
Using the terminology of the bi-Hamiltonian formalism, we have pencils of two non-abelian dynamical systems and two non-abelian first integrals.

Instead of algebra ${\cal{A}}$ with multiplication $x \, y$ one can consider the associative algebra with the opposite product 
$x\star y = y \, x$. The transition to the opposite multiplication is represented by the involution $\tau: {\cal A}\to {\cal A}$ defined by 
\begin{equation}\label{tau}
\tau(u)=u,\quad \tau(v)=v, \quad \tau(a x+b y) = a \tau(x)+b \tau(y), \quad \tau(x y) = \tau(y) \tau(x), 
\end{equation}
where $x,y\in {\cal A}, \,\, a,b\in \C$. This involution is called {\it transposition}. 

All key properties of integrable systems, such as the existence of first integrals, infinitesimal symmetries, Lax representations, etc., are invariant under $\tau$.

\section{\PPainleve-6 systems} \label{section2}

Consider the non-abelian systems \eqref{eq:noz}, where the non-commutative polynomials $P_i (u, v)$, $Q_i (u, v)$ are given by the formulas
\eqref{eq:P1form}, \eqref{eq:Q2form}.

\begin{proposition} \label{Prop1}
Such a system possesses a non-abelian Okamoto integral of the form \eqref{eq:matintansatz}, \eqref{eq:matintansatzI2}, where $\kappa_1$, $\kappa_2$, $\kappa_3$, $\kappa_4$ are arbitrary parameters iff the corresponding $\PVI$ system belongs to the list \text{\rm\ref{eq:P6case11} -- \ref{eq:P6case53}} from Appendix \ref{sec:sysintlistP6}. 
\end{proposition}

\begin{proof} 
Differentiating the integral \eqref{eq:matintansatz} with respect to the system defined by \eqref{eq:noz},  \eqref{eq:P1form}, \eqref{eq:Q2form}, we obtain a polynomial $Y(u,v,z)$ of degree 8 in $u, v$. 
Equating to zero the coefficients of different monomials  in $Y(u,v,z)$, we arrive at a system of nonlinear algebraic equations. The simplest equations from this system are:
$$
p_1=p_3=p_5=p_7=0, \qquad p_2=1-p_4-p_5-p_8-p_9.
$$
It turns out that all coefficients of polynomials $P_i$, $Q_i$ can be expressed in terms of the Okamoto integral in the following way:
$$
\begin{array}{l}
    a_1 =  1 - p_4 - p_5 - p_8 - p_9, \quad 
    a_2 = 1 + p_4 - p_8 - p_9, \quad 
    a_3 = p_5 + 2 p_8 + p_9, \quad 
    \\[3mm]
    b_1 = -1 + p_4 + p_5 + p_8 + p_9, \qquad
    b_2 = -2 + p_5 + 2 p_8 + 2 p_9, 
    \\[3mm]
    b_3 = 0,\qquad 
    b_4 = -p_4 - p_5 - p_8,\qquad
    b_5 = -p_5 - 2 p_8 - 2 p_9, \qquad 
    \\[3mm]
    c_1 = -d_1 = 2 q_1 + q_2, \qquad 
    c_2 = -d_2 = -2 - 2 q_1 - 2 q_2 - 2 q_3 - 2 q_4 - q_5, \qquad 
    h_1 = 2 x_1 + x_2, \qquad 
    \\[3mm]
    e_1 = 2 \kappa_1 + 2 r_1 + r_2, \qquad  
    f_2 = 1 -t_3 + x_1 + x_2 - p_5 - p_8 - 2 p_9 + 2 q_1 + 
    2 q_2 + q_3 + 2 q_4 + q_5,  \qquad 
    \\[3mm]
    f_1 = -2 - t_3 - x_1 + p_4 + p_5 + 2 p_8 + 2 p_9 - 2 q_1 - q_2 - q_3 .
\end{array}
$$

Equating to zero the coefficients of different monomials of degree 8 in $Y(u,v,z)$, we obtain a system of nonlinear algebraic equations for the variables $a_i, \, i=1,2,3$, \quad $b_i, \, i=1,...,5$ \,\, and \,\, $p_i,\, i=1,...,9.$ Using the above formulas, we can eliminate $a_i$  and $b_i$ and obtain a system for $p_4, p_5, p_8, p_9$, which is equivalent to 
$$  (p_4-1) p_4 =  (p_5-1) p_5= (p_8-1) p_8= (p_9-1) p_9=0; \quad 
p_4 p_5 =  p_4 p_8 = p_4 p_9 =  p_5 p_8 = p_5 p_9 =  p_8 p_9 = 0.
$$

This system have  5 solutions which leads to the following cases:
$$
\begin{array}{ll}
    \text{\bf Case 1}: 
    & 
    a_1 = 0,\, 
    a_2 = 0,\, 
    a_3 = 2,\, 
    b_1 = 0,\, 
    b_2 = 0, \, 
    b_3 = 0,\, 
    b_4 = -1,\, 
    b_5 = -2,
    \\[2mm] 
    &
    p_1 = 0,\, 
    p_2 = 0,\, 
    p_3 = 0,\,
    p_4 = 0,\, 
    p_5 = 0,\, 
    p_6 = 0,\, 
    p_7 = 0,\, 
    p_8 = 1,\, 
    p_9 = 0;
\end{array}
$$
$$
\begin{array}{ll}
    \text{\bf Case 2}: 
    & 
    a_1 = 0, \,
    a_2 = 0, \,
    a_3 = 1, \,
    b_1 = 0, \,
    b_2 = 0, \,
    b_3 = 0, \,
    b_4 = 0, \,
    b_5 = -2,
    \\[2mm] 
    &
    p_1 = 0, \, 
    p_2 = 0, \,
    p_3 = 0, \,
    p_4 = 0, \,
    p_5 = 0, \,
    p_6 = 0, \,
    p_7 = 0, \,
    p_8 = 0, \,
    p_9 = 1;
\end{array}
$$
$$
\begin{array}{ll}
    \text{\bf Case 3}: 
    & 
    a_1 = 0, \,
    a_2 = 1, \,
    a_3 = 1, \,
    b_1 = 0, \,
    b_2 = -1, \,
    b_3 = 0, \,
    b_4 = -1, \,
    b_5 = -1,
    \\[2mm] 
    &
    p_1 = 0, \,
    p_2 = 0, \,
    p_3 = 0, \,
    p_4 = 0, \,
    p_5 = 1, \,
    p_6 = 0, \,
    p_7 = 0, \,
    p_8 = 0, \,
    p_9 = 0;
\end{array}
$$
$$
\begin{array}{ll}
    \text{\bf Case 4}: 
    & 
    a_1 = 0, \,
    a_2 = 2, \,
    a_3 = 0, \,
    b_1 = 0, \,
    b_2 = -2, \,
    b_3 = 0, \,
    b_4 = -1, \,
    b_5 = 0,
    \\[2mm] 
    &
    p_1 = 0, \,
    p_2 = 0, \,
    p_3 = 0, \,
    p_4 = 1, \,
    p_5 = 0, \,
    p_6 = 0, \,
    p_7 = 0, \,
    p_8 = 0, \,
    p_9 = 0;
\end{array}
$$
$$
\begin{array}{ll}
    \text{\bf Case 5}: 
    & 
    a_1 = 1, \,
    a_2 = 1, \,
    a_3 = 0, \,
    b_1 = -1, \,
    b_2 = -2, \,
    b_3 = 0, \,
    b_4 = 0, \,
    b_5 = 0,
    \\[2mm] 
    &
    p_1 = 0, \,
    p_2 = 1, \,
    p_3 = 0, \,
    p_4 = 0, \,
    p_5 = 0, \,
    p_6 = 0, \,
    p_7 = 0, \,
    p_8 = 0, \,
    p_9 = 0.
\end{array}
$$
In each case, equating to zero the remaining coefficients in the polynomial $Y(u,v,z)$, we obtain a large but rather simple algebraic system for $c_i,\, d_i, \, f_i, \, g_i,\, q_i, \, r_i,\, s_i, \, t_i, \, x_i, \, y_i$. This system contains $\kappa_1, \, \kappa_2, \, \kappa_3, \, \kappa_4$ as parameters. Solving the algebraic system in  {Case \bf 1}, we obtain systems \ref{eq:P6case11}~--~\ref{eq:P6case14}, {Case \bf 2} leads to \ref{eq:P6case21}~--~\ref{eq:P6case23}, in {Case \bf 3} we have \ref{eq:P6case31}~--~\ref{eq:P6case36}. In {Cases \bf 4} and {\bf 5} the systems are given by \ref{eq:P6case41}~--~\ref{eq:P6case44} and \ref{eq:P6case51}~--~\ref{eq:P6case53} respectively. All systems contain four arbitrary parameters $\kappa_1-\kappa_4$. Notice that additional systems that correspond to particular values of the parameters do not exist.
\end{proof}

\subsection{Transformation group}
\label{subsec21}
The transformations 
\begin{gather}
    \label{eq:scalP6sym1}
    \begin{aligned}
    r_1:&
    &
    \brackets{
    z, u, v
    }
    &\mapsto \brackets{
    1 - z, \, 
    1 - u, \,
    - v
    },
    &&&&&
    r_2:&
    &
    \brackets{
    z, u, v
    }
    &\mapsto \brackets{
    z^{-1}, \,
    z^{-1} u, \,
    z v
    },
    \end{aligned}
\end{gather}
act on the set of eighteen systems from Proposition \ref{Prop1}. They change the parameters in the following way
\begin{align}
    \label{eq:scalP6sym1k}
    &\begin{aligned}
    r_1:&
    &
    \brackets{
    \kappa_1, 
    \kappa_2, 
    \kappa_3, 
    \kappa_4
    }
    &\mapsto \brackets{
    \kappa_1, \,
    2 \kappa_1 - \kappa_2 - \kappa_4, \,
    \kappa_3, \,
    \kappa_4
    },
    \end{aligned}
    \\[2mm]
    \label{eq:scalP6sym2k}
    &\begin{aligned}
    r_2:&
    &
    \brackets{
    \kappa_1, 
    \kappa_2, 
    \kappa_3, 
    \kappa_4
    }
    &\mapsto \brackets{
    \kappa_1, \,
    \kappa_4 - 1, \,
    \kappa_3, \,
    \kappa_2 + 1
    }.
    \end{aligned}
     \end{align}
These transformations generate a group isomorphic to~$S_3$. The involution $\tau$ defined by  \eqref{tau} commutes with $r_i$ and also acts  on the set of systems. We denote the action of $\tau$ by the superscript $T$. The whole transformation group is isomorphic to $S_3\times \Z_2$ and acts on the $\PVI$ type systems as follows:
\begin{align}
    \text{Cases \textbf{1} and \textbf{4}}&:
    &
    r_1 {\text{\eqref{eq:P6case11}}}
    &= \text{\ref{eq:P6case14}},
    &
    r_2 {\text{\eqref{eq:P6case11}}}
    &= \text{\ref{eq:P6case11}},
    &
    \text{\eqref{eq:P6case11}}^T
    &= \text{\ref{eq:P6case41}},
    \\[2mm]
    &&
    r_1 {\text{\eqref{eq:P6case13}}}
    &= \text{\ref{eq:P6case13}},
    &
    r_2 {\text{\eqref{eq:P6case13}}}
    &= \text{\ref{eq:P6case14}},
    &
    \text{\eqref{eq:P6case13}}^T
    &= \text{\ref{eq:P6case44}},
    \\[2mm]
    &&
    r_1 {\text{\eqref{eq:P6case14}}}
    &= \text{\ref{eq:P6case11}},
    &
    r_2 {\text{\eqref{eq:P6case14}}}
    &= \text{\ref{eq:P6case13}},
    &
    \text{\eqref{eq:P6case14}}^T
    &= \text{\ref{eq:P6case43}};
    \\[3mm]
    \text{Cases \textbf{2} and \textbf{5}}&:
    &
    r_1 {\text{\eqref{eq:P6case21}}}
    &= \text{\ref{eq:P6case23}},
    &
    r_2 {\text{\eqref{eq:P6case21}}}
    &= \text{\ref{eq:P6case22}},
    &
    \text{\eqref{eq:P6case21}}^T
    &= \text{\ref{eq:P6case51}},
    \\[2mm]
    &&
    r_1 {\text{\eqref{eq:P6case22}}}
    &= \text{\ref{eq:P6case22}},
    &
    r_2 {\text{\eqref{eq:P6case22}}}
    &= \text{\ref{eq:P6case21}},
    &
    \text{\eqref{eq:P6case22}}^T
    &= \text{\ref{eq:P6case52}},
    \\[2mm]
    &&
    r_1 {\text{\eqref{eq:P6case23}}}
    &= \text{\ref{eq:P6case21}},
    &
    r_2 {\text{\eqref{eq:P6case23}}}
    &= \text{\ref{eq:P6case23}},
    &
    \text{\eqref{eq:P6case23}}^T
    &= \text{\ref{eq:P6case53}};
\end{align}
\begin{align}
    \phantom{\text{s and \textbf{44}\,}}
    \text{Case \textbf{3}}&:
    &
    r_1 {\text{\eqref{eq:P6case31}}}
    &= \text{\ref{eq:P6case32}},
    &
    r_2 {\text{\eqref{eq:P6case31}}}
    &= \text{\eqref{eq:P6case31}}^T,
    &
    \text{\eqref{eq:P6case31}}^T
    &= \text{\ref{eq:P6case34}},
    \\[2mm]
    &&
    r_1 {\text{\eqref{eq:P6case32}}}
    &= \text{\ref{eq:P6case31}},
    &
    r_2 {\text{\eqref{eq:P6case32}}}
    &= \text{\eqref{eq:P6case33}}^T,
    &
    \text{\eqref{eq:P6case32}}^T
    &= \text{\ref{eq:P6case35}},
    \\[2mm]
    &&
    r_1 {\text{\eqref{eq:P6case33}}}
    &= \text{\eqref{eq:P6case33}}^T,
    &
    r_2 {\text{\eqref{eq:P6case33}}}
    &= \text{\eqref{eq:P6case32}}^T,
    &
    \text{\eqref{eq:P6case33}}^T
    &= \text{\ref{eq:P6case36}}.
\end{align}
It follows from these relations that there are three orbits of the action of the transformation group: 
\begin{gather*}
    \text{\bf Orbit 1} = \left\{ 
        \text{\ref{eq:P6case11}}, \,
        \text{\ref{eq:P6case13}}, \,
        \text{\ref{eq:P6case14}}, \,
        \text{\ref{eq:P6case41}}, \,
        \text{\ref{eq:P6case43}}, \,
        \text{\ref{eq:P6case44}}
    \right\},
    \\[2mm]
    \text{\bf Orbit 2} = \left\{ 
        \text{\ref{eq:P6case21}}, \,
        \text{\ref{eq:P6case22}}, \,
        \text{\ref{eq:P6case23}}, \,
        \text{\ref{eq:P6case51}}, \,
        \text{\ref{eq:P6case52}}, \,
        \text{\ref{eq:P6case53}}
    \right\},
    \\[2mm]
    \text{\bf Orbit 3} = \left\{ 
        \text{\ref{eq:P6case31}}, \,
        \text{\ref{eq:P6case32}}, \,
        \text{\ref{eq:P6case33}}, \,
        \text{\ref{eq:P6case34}}, \,
        \text{\ref{eq:P6case35}}, \,
        \text{\ref{eq:P6case36}}
    \right\}.
\end{gather*}
We choose the systems \ref{eq:P6case14}, \ref{eq:P6case23}, and \ref{eq:P6case31} as representatives of these orbits. Other systems can be obtained from them by the formulas:
\begin{align}
    &\begin{aligned}
    \text{\ref{eq:P6case11}}
    &= r_1 \text{\eqref{eq:P6case14}}
    ,
    &
    \text{\ref{eq:P6case13}}
    &= r_2 \text{\eqref{eq:P6case14}}
    ,
    &
    \text{\ref{eq:P6case41}}
    &= \brackets{r_1 \text{\eqref{eq:P6case14}}}^T
    ,
    &
    \text{\ref{eq:P6case43}}
    &= \text{\eqref{eq:P6case14}}^T
    ,
    &
    \text{\ref{eq:P6case44}}
    &= \brackets{r_2 \text{\eqref{eq:P6case14}}}^T
    ,
    \end{aligned}
    \\
    &\begin{aligned}
    \text{\ref{eq:P6case21}}
    &= r_1 \text{\eqref{eq:P6case23}}
    ,
    &
    \text{\ref{eq:P6case22}}
    &= r_1 r_2 \text{\eqref{eq:P6case23}}
    ,
    &
    \text{\ref{eq:P6case51}}
    &= \brackets{r_1 \text{\eqref{eq:P6case23}}}^T
    ,
    &
    \text{\ref{eq:P6case52}}
    &= \brackets{r_1 r_2 \text{\eqref{eq:P6case23}}}^T
    ,
    &
    \text{\ref{eq:P6case53}}
    &= \text{\eqref{eq:P6case23}}^T
    ,
    \end{aligned}
    \\
    &\begin{aligned}
    \text{\ref{eq:P6case32}}
    &= r_1 \text{\eqref{eq:P6case31}}
    ,
    &
    \text{\ref{eq:P6case33}}
    &= \brackets{r_1 r_2 \text{\eqref{eq:P6case31}}}^T
    ,
    &
    \text{\ref{eq:P6case34}}
    &= \text{\eqref{eq:P6case31}}^T
    ,
    &
    \text{\ref{eq:P6case35}}
    &= \brackets{r_1 \text{\eqref{eq:P6case31}}}^T
    ,
    &
    \text{\ref{eq:P6case36}}
    &= r_1 r_2 \text{\eqref{eq:P6case31}}
    .
    \end{aligned}
\end{align}

\subsection{Isomonodromic representations} \label{Isomonodromic}

It is well-known \cite{jimbo1981monodromy} that the scalar system \eqref{eq:scalP6sys} has the isomonodromic representation \eqref{eq:zerocurvcond}, where matrices $\mathbf{A} (z, \lambda)$ and $\mathbf{B} (z, \lambda)$ have the form 
\begin{align} \label{eq:matABform}
    \mathbf{A} (z, \lambda)
    &= \dfrac{A_0}{\lambda}
    + \dfrac{A_1}{\lambda - 1}
    + \dfrac{A_2}{\lambda - z},
    &
    \mathbf{B} (z, \lambda)
    &= B_{0}
    - \dfrac{A_2}{\lambda - z}
\end{align}
with the following matrices $A_0$, $A_1$, $A_2$, and $B_{0}$:
\begin{gather} 
    \notag
    \begin{aligned}
    A_0
    &= 
    \begin{pmatrix}
    - 1
    - \kappa_1
    + \kappa_4
    & 
    u z^{-1} - 1
    \\[0.9mm]
    0 & 0
    \end{pmatrix},
    &&&
    A_1
    &= 
    \begin{pmatrix}
    - u v + \kappa_1
    & 
    1 
    \\[0.9mm]
    - u^2 v^2 
    + \kappa_1 u v 
    + \kappa_3
    & 
    u v
    \end{pmatrix},
    \end{aligned}
    \\[2mm]
    \label{eq:scalLaxpair}
    A_2
    = 
    \begin{pmatrix}
    u v 
    + (\kappa_1 - \kappa_2 - \kappa_4) 
    & 
    - u z^{-1}
    \\[0.9mm]
    z u v^2 
    + (\kappa_1 - \kappa_2 - \kappa_4)  z v 
    & 
    - u v
    \end{pmatrix},
    \\[2mm]
    \notag
    B_{0}
    = 
    \begin{pmatrix}
    (z (z - 1))^{-1} \brackets{
    2 u^2 v - \kappa_1 u 
    - z \brackets{
    2 u v
    + (\kappa_1 - \kappa_2 - \kappa_4)
    }
    }
    & 0
    \\[0.9mm]
    - u v^2 - (\kappa_1 - \kappa_2 - \kappa_4)  v
    & 0
    \end{pmatrix}.
\end{gather}
It is clear that the shifts
\begin{equation}\label{shifts}
 \mathbf{A}\mapsto   \mathbf{A} + p(\lambda)\, \mathbf{I}, \qquad  \mathbf{B}\mapsto   \mathbf{B} + q(z, u, v)\, \mathbf{I}
\end{equation}
are allowed by the relation \eqref{eq:matABform}. Note that in the non-abelian case the second shift is admissible only if $q$ does not depend on $u$ and $v.$

We generalize this representation to the case of non-abelian systems  \ref{eq:P6case14}, \ref{eq:P6case23}, \ref{eq:P6case31} by means of the non-abelinizating procedure described in \cite{Bobrova_Sokolov_2021_2}. 

\subsubsection{System \texorpdfstring{\ref{eq:P6case14}}{3P6}:}
\begin{gather} 
    \notag
    \begin{aligned}
    A_0
    &= 
    \begin{pmatrix}
    - 1
    - \kappa_1 
    + \kappa_4
    & 
    u z^{-1} - 1
    \\[0.9mm]
    0 & 0
    \end{pmatrix},
    &&&
    A_1
    &= 
    \begin{pmatrix}
    - u v + \kappa_1
    & 
    1 
    \\[0.9mm]
    \begin{array}{l}
        - v u^2 v 
        + \tfrac12 (\kappa_1 + \kappa_3) u v
         \\[0.6mm]
        + \tfrac12 (\kappa_1 - \kappa_3) v u
        + \tfrac14 (\kappa_3^2 - \kappa_1^2)
    \end{array}
    & 
    v u
    \end{pmatrix},
    \end{aligned}
\end{gather}
\begin{gather}
    \label{eq:Laxpair_P6case14}
    A_2
    = 
    \begin{pmatrix}
    u v 
    + (\kappa_1 - \kappa_2 - \kappa_4) 
    & 
    - u z^{-1}
    \\[0.9mm]
    z v u v
    + (\kappa_1 - \kappa_2 - \kappa_4)  z v 
    & 
    - v u
    \end{pmatrix},
    \\[2mm]
    \notag
    B_{0}
    = 
    \begin{pmatrix}
    (z (z - 1))^{-1} \brackets{
    2 u v u 
    - \kappa_1 u 
    - z \brackets{
    u v + v u + (\kappa_1 - \kappa_2 - \kappa_4)
    }
    }
    & 0
    \\[0.9mm]
    - v u v 
    - (\kappa_1 - \kappa_2 - \kappa_4)  v
    & 0
    \end{pmatrix}.
\end{gather}

\subsubsection{System \texorpdfstring{\ref{eq:P6case23}}{6P6}:}
\begin{gather} 
    \notag
    \begin{aligned}
    A_0
    &= 
    \begin{pmatrix}
    - 1
    - \kappa_1 
    + \kappa_4
    & 
    u z^{-1} - 1
    \\[0.9mm]
    0 & 0
    \end{pmatrix},
    &&&
    A_1
    &= 
    \begin{pmatrix}
    - v u + \kappa_1
    & 
    1 
    \\[0.9mm]
    - v u v u
    + \kappa_1 v u
    + \kappa_3
    & 
    v u
    \end{pmatrix},
    \end{aligned}
    \\[2mm]
    \label{eq:Laxpair_P6case23}
    A_2
    = 
    \begin{pmatrix}
    v u
    + (\kappa_1 - \kappa_2 - \kappa_4) 
    & 
    - u z^{-1}
    \\[0.9mm]
    z v^2 u
    + (\kappa_1 - \kappa_2 - \kappa_4)  z v 
    & 
    - v u
    \end{pmatrix},
    \\[2mm]
    \notag
    B_{0}
    = 
    \begin{pmatrix}
    (z (z - 1))^{-1} \brackets{
    u v u + v u^2
    - \kappa_1 u
    - z \brackets{
    v u 
    + v
    + (\kappa_1 - \kappa_2 - \kappa_4)
    }
    }
    & 0
    \\[0.9mm]
    - v^2 u
    - (\kappa_1 - \kappa_2 - \kappa_4)  v
    & 
    (z - 1)^{-1} \brackets{
    v u - v
    }
    \end{pmatrix}.
\end{gather}

\subsubsection{System \texorpdfstring{\ref{eq:P6case31}}{7P6}:}
\begin{gather} 
    \notag
    \begin{aligned}
    A_0
    &= 
    \begin{pmatrix}
    - 1
    - \kappa_1 
    + \kappa_4
    & 
    u z^{-1} - 1
    \\[0.9mm]
    0 & 0
    \end{pmatrix},
    &&&
    A_1
    &= 
    \begin{pmatrix}
    - u v + \kappa_1
    & 
    1 
    \\[0.9mm]
    - u v u v
    + \kappa_1 u v
    + \kappa_3
    & 
    u v
    \end{pmatrix},
    \end{aligned}
    \\[2mm]
    \label{eq:Laxpair_P6case31}
    A_2
    = 
    \begin{pmatrix}
    u v 
    + (\kappa_1 - \kappa_2 - \kappa_4) 
    & 
    - u z^{-1}
    \\[0.9mm]
    z u v^2
    + (\kappa_1 - \kappa_2 - \kappa_4)  z v 
    & 
    - u v
    \end{pmatrix},
    \\[2mm]
    \notag
    B_{0}
    = 
    \begin{pmatrix}
    (z (z - 1))^{-1} \brackets{
    u^2 v + u v u
    - \kappa_1 u
    - z \brackets{
    2 u v + v u 
    - v
    + (\kappa_1 - \kappa_2 - \kappa_4)
    }
    }
    & 0
    \\[0.9mm]
    - u v^2
    - (\kappa_1 - \kappa_2 - \kappa_4)  v
    & 
    (z - 1)^{-1} (- v u + v)
    \end{pmatrix}.
\end{gather}

The isomonodromic representations for remaining 15 systems from Proposition \ref{Prop1} can be obtained using the transformation group from Section \ref{subsec21}. 
Since the involutions $r_1$ and $r_2$ do not preserve the structure of the Lax pair, we should supplement each of them 
 by a proper  gauge transformation 
\begin{equation}\label{AB}
    \mathbf{A} \mapsto g \, \mathbf{A} \, g^{-1}
    + g_{\lambda}' \, g^{-1}, \qquad 
    \mathbf{B} \mapsto g \, \mathbf{B} \, g^{-1}
    + g_z' \, g^{-1}
\end{equation}
followed by a change of variables $z$, $u$, $v$, $\kappa_1$, $\kappa_2$, $\kappa_3$, $\kappa_4$, $\lambda$.

The involution $r_1$ \eqref{eq:scalP6sym1} we extend to the Lax pair by the formulas
 \begin{align}
    \label{eq:r1tocan_Laxpair}
    \lambda
    &\mapsto (1 - z) \, \lambda \, (\lambda - z)^{-1},
    &
    g
    &= 
    \begin{pmatrix}
    (z - 1)^{-1} (\lambda - z) & 0 \\
    \lambda \, v & -1
    \end{pmatrix}
    .
\end{align}

For the involution $r_2$ a similar transformation is defined by  
\begin{align}
    \label{eq:r2tocan_Laxpair}
    \lambda
    &\mapsto \lambda^{-1},
    &
    g
    &= 
    \begin{pmatrix}
    \lambda & 0 \\ 0 & 1
    \end{pmatrix}
    .
\end{align}
As a result,  in both cases we arrive at a Lax pair, which coincides with \eqref{eq:scalLaxpair} the under the commutative reduction.

\begin{example}
\label{exm:r1deg}
Let us show how to get a Lax pair for the {\rm\ref{eq:P6case32}} system starting with the pair \eqref{eq:Laxpair_P6case31} for the {\rm\ref{eq:P6case31}} system.

As it was mentioned in Section {\rm\ref{subsec21}}, the {\rm\ref{eq:P6case32}} system is related to the system {\rm\ref{eq:P6case31}} via the involution $r_1$. This connection on the level of Lax
pairs is described above. Applying to the pair \eqref{eq:Laxpair_P6case31} the gauge transformation \eqref{AB} with the matrix $g$ defined in \eqref{eq:r1tocan_Laxpair}
and making the change of the variables and parameters
\begin{align}
    z
    &\mapsto 1 - z,
    &
    u
    &\mapsto 1 - u,
    &
    v
    &\mapsto - v,
    &
    \kappa_2
    &\mapsto 2 \kappa_1 - \kappa_2 - \kappa_4,
    &
    \lambda
    &\mapsto (1 - z) \, \lambda \, (\lambda - z)^{-1},
\end{align}
we obtain a pair of the form \eqref{eq:matABform}, where matrices $A_0$, $A_1$, $A_2$, and $B_0$ are
\begin{gather} 
    \notag
    \begin{aligned}
    A_0
    &= 
    \begin{pmatrix}
    - 1
    - \kappa_1 
    + \kappa_4
    & 
    u z^{-1} - 1
    \\[0.9mm]
    0 & 0
    \end{pmatrix},
    &&&
    A_1
    &= 
    \begin{pmatrix}
    - u v + \kappa_1
    & 
    1 
    \\[0.9mm]
    - u v u v
    + \kappa_1 u v
    + \kappa_3
    & 
    u v
    \end{pmatrix},
    \end{aligned}
    \\[2mm]
    \label{eq:Laxpair_P6case32}
    A_2
    = 
    \begin{pmatrix}
    u v 
    + (\kappa_1 - \kappa_2 - \kappa_4) 
    & 
    - u z^{-1}
    \\[0.9mm]
    z v u v
    + (\kappa_1 - \kappa_2 - \kappa_4)  z v 
    & 
    - v u
    \end{pmatrix},
    \\[2mm]
    \notag
    \hspace{-5mm}
    B_{0}
    = 
    \begin{pmatrix}
    (z (z - 1))^{-1} \brackets{
    u^2 v + u v u
    - \kappa_1 u
    - z \brackets{
    2 u v + v u 
    - v
    + (\kappa_1 - \kappa_2 - \kappa_4)
    }
    }
    & 0
    \\[0.9mm]
    - v u v
    - (\kappa_1 - \kappa_2 - \kappa_4)  v
    & 
    (z - 1)^{-1} (- u v + v)
    \end{pmatrix}.
\end{gather}
The zero-curvature condition for this pair leads to the {\rm\ref{eq:P6case32}} system.
\end{example}

In the case of transposition \eqref{tau} the situation is slightly different. The matrices of 
a pair transform as
\begin{align}
    \label{eq:tau_Laxpair}
    \mathbf{A}
    &\mapsto - \mathbf{A}^t,
    &
    \mathbf{B}
    &\mapsto - \mathbf{B}^t,
\end{align}
where each matrix entry should be transformed by $\tau$. Here the subscript $t$ means the matrix transpose. After that, using the map \eqref{AB} with
\begin{align}
    \label{eq:tautocan_Laxpair}
    g
    &= \lambda^{- 1 - \kappa_1 + \kappa_4} \, (\lambda - 1)^{\kappa_1} \, (\lambda - z)^{\kappa_1 - \kappa_2 - \kappa_4}
    \begin{pmatrix}
    0 & 1 \\ -1 & 0
    \end{pmatrix},
\end{align}
we get matrices $A_0$, $A_1$, $A_2$ with the same dependence on the variables and parameters. The diagonal entries of the matrix $B_0$ have a different structure. However, in the scalar case the resulting matrix $B_0$ can be brought to the form \eqref{eq:scalLaxpair} by a shift \eqref{shifts}. In the non-abelian case it is impossible.   

\begin{example}
\label{exm:taudeg}
The {\rm\ref{eq:P6case34}} system can be obtained from the system {\rm\ref{eq:P6case31}}  by the transposition \eqref{tau}. Let us extend this transformation to the Lax pairs. According to \eqref{eq:tau_Laxpair}, the pair \eqref{eq:Laxpair_P6case31} turns into 
\begin{gather} 
    \notag
    \begin{aligned}
    A_0
    &= -
    \begin{pmatrix}
    - 1
    - \kappa_1 
    + \kappa_4
    & 
    0
    \\[0.9mm]
    u z^{-1} - 1 
    & 0
    \end{pmatrix},
    &&&
    A_1
    &= -
    \begin{pmatrix}
    - v u + \kappa_1
    & 
    - v u v u
    + \kappa_1 v u
    + \kappa_3
    \\[0.9mm]
    1 
    & 
    v u
    \end{pmatrix},
    \end{aligned}
    \\[2mm]
    \label{eq:Laxpair_P6case34_}
    A_2
    = -
    \begin{pmatrix}
    v u
    + (\kappa_1 - \kappa_2 - \kappa_4) 
    & 
    z v^2 u
    + (\kappa_1 - \kappa_2 - \kappa_4)  z v 
    \\[0.9mm]
    - u z^{-1}
    & 
    - v u
    \end{pmatrix},
    \\[2mm]
    \notag
    \hspace{-8mm}
    B_{0}
    = -
    \begin{pmatrix}
    (z (z - 1))^{-1} \brackets{
    v u^2 + u v u
    - \kappa_1 u
    - z \brackets{
    2 v u + u v
    - v
    + (\kappa_1 - \kappa_2 - \kappa_4)
    }
    }
    & 
    - v^2 u
    - (\kappa_1 - \kappa_2 - \kappa_4)  v
    \\[0.9mm]
    0
    & 
    (z - 1)^{-1} (- u v + v)
    \end{pmatrix}.
\end{gather}
Now, using the gauge transformation with \eqref{eq:tautocan_Laxpair}, we get
\begin{gather} 
    \notag
    \begin{aligned}
    A_0
    &= 
    \begin{pmatrix}
    - 1
    - \kappa_1 
    + \kappa_4
    & 
    u z^{-1} - 1
    \\[0.9mm]
    0 & 0
    \end{pmatrix},
    &&&
    A_1
    &= 
    \begin{pmatrix}
    - v u + \kappa_1
    & 
    1 
    \\[0.9mm]
    - v u v u
    + \kappa_1 v u
    + \kappa_3
    & 
    v u
    \end{pmatrix},
    \end{aligned}
    \\[2mm]
    \label{eq:Laxpair_P6case34}
    A_2
    = 
    \begin{pmatrix}
    v u
    + (\kappa_1 - \kappa_2 - \kappa_4) 
    & 
    - u z^{-1}
    \\[0.9mm]
    z v^2 u
    + (\kappa_1 - \kappa_2 - \kappa_4)  z v 
    & 
    - v u
    \end{pmatrix},
    \\[2mm]
    \notag
    \hspace{-8mm}
    B_{0}
    = 
    \begin{pmatrix}
    (z - 1)^{-1} (u v - v)
    & 0
    \\[0.9mm]
    - v^2 u
    - (\kappa_1 - \kappa_2 - \kappa_4)  v
    & 
    (z (z - 1))^{-1} \brackets{
    - v u^2 - u v u
    + \kappa_1 u
    + z \brackets{
    2 v u + u v
    - v
    + (\kappa_1 - \kappa_2 - \kappa_4)
    }
    }
    \end{pmatrix}.
\end{gather}
\end{example}

\section{\texorpdfstring{Systems of $\PV - \PII$}{P5 - P2} type}
\label{sec:P5_P1systems}

\subsection{\PPainleve-5 systems}
The scalar $\PV$ system   has the following structure:
\begin{align}\label{P5}
    \left\{
    \begin{array}{lclcl}
      \ds   z \, \frac{d u}{d z}
         &=& P_1 (u, v)
         + \kappa_4 \, z \, u
         ,
         \\[3mm]
     \ds    z \, \frac{d v}{d z}
         &=& P_2 (u, v)
         - \kappa_4 \, z \, v.
    \end{array}
    \right.
\end{align}
The non-abelian ansatz for the polynomials $P_1(u, v)$ and $P_2(u, v)$ are given by 
\begin{align}
    \notag
    &\begin{aligned}
    \hspace{-8mm}
    P_1 (u, v)
    = a_1 u^3 v 
    + a_2 u^2 v u 
    + a_3 u v u^2 
    + \brackets{
    2 - \sum a_i
    } v u^3
    + c_1 u^2 v 
    + (- 4 - c_1 - c_2) u v u 
    \\[1mm]
    + \, c_2 v u^2
    - \kappa_1 u^2
    + e_1 u v + (2 - e_1) v u 
    + (\kappa_1 + \kappa_2) u
    - \kappa_2
    ,
    \end{aligned}
    \\[-1mm]
    \label{eq:ncP5sys}
    \\[-1mm]
    \notag
    &\begin{aligned}
    \hspace{-8mm}
    P_2 (u, v)
    = b_1 u^2 v^2 
    + b_2 u v u v 
    + b_3 u v^2 u
    + b_4 v u^2 v
    + b_5 v u v u
    + \brackets{
    - 3 - \sum b_i
    } v^2 u^2
    + d_1 u v^2 
    \\[1mm]
    + \, (4 - d_1 - d_2) v u v 
    - d_2 v^2 u 
    - v^2 
    + f_1 u v + (2 \kappa_1 - f_1) v u
    - (\kappa_1 + \kappa_2) v
    + \kappa_3,
    \end{aligned}
\end{align}
while the non-abelian Okamoto integral has the form 
\begin{align} \label{eq:ncJ5}
\hspace{-8mm}
\begin{aligned}
    J
    = p_1 u^3 v^2
    + p_2 u^2 v u v
    + p_3 u^2 v^2 u
    + p_4 u v u^2 v
    + p_5 u v u v u 
    + p_6 u v^2 u^2 
    + p_7 v u^3 v
    + p_8 v u^2 v u 
    \\[1mm]
    + \, p_9 v u v u^2
    + \brackets{1 - \sum p_i} v^2 u^3
    + q_1 u^2 v^2 
    + q_2 u v u v 
    + q_3 u v^2 u
    + q_4 v u^2 v
    + q_5 v u v u
    \\[1mm]
    + \brackets{- 2 - \sum q_i} v^2 u^2
    + r_1 u^2 v 
    + r_2 u v u 
    + \brackets{- \kappa_1 - \sum r_i} v u^2
    + s_1 u v^2
    + s_2 v u v
    \\[1mm]
    + \brackets{1 - \sum s_i} v^2 u
    + t_1 u v 
    + (\kappa_1 + \kappa_2 - t_1) v u
    - \kappa_3 u
    - \kappa_2 v + \kappa_3
    \\[1mm]
    + \, z \brackets{
    w_1 u v + (\kappa_4 - w_1) v u
    }
    .
\end{aligned}
\end{align}
\begin{proposition}
The auxiliary system 
\begin{align*}
    \left\{
    \begin{array}{lclcl}
      \ds   \frac{d u}{d t}
         &=& P_1 (u, v)
         + \kappa_4 z u
         ,
         \\[3mm]
     \ds     \frac{d v}{d t}
         &=& P_2 (u, v)
         - \kappa_4 z v
    \end{array}
    \right.
\end{align*}
for \eqref{P5}, \eqref{eq:ncP5sys} has a non-abelian integral of the form \eqref{eq:ncJ5} iff the corresponding $\PV$ system belongs to the list \text{\rm\ref{eq:P5case11} -- \ref{eq:P5case52}} from Appendix \ref{sec:sysintlistP5}.
\end{proposition}

There are five orbits,
\begin{gather*}
    \begin{aligned}
    \text{\bf Orbit 1} 
    &= \left\{ 
        \text{\ref{eq:P5case11}}, \,
        \text{\ref{eq:P5case43}}
    \right\},
    &&&
    \text{\bf Orbit 2} 
    &= \left\{ 
        \text{\ref{eq:P5case13}}, \,
        \text{\ref{eq:P5case44}}
    \right\},
    &&&
    \text{\bf Orbit 3} 
    &= \left\{ 
        \text{\ref{eq:P5case21}}, \,
        \text{\ref{eq:P5case51}}
    \right\},
    \end{aligned}
    \\[2mm]
    \begin{aligned}
    \text{\bf Orbit 4} 
    &= \left\{ 
        \text{\ref{eq:P5case22}}, \,
        \text{\ref{eq:P5case52}}
    \right\},
    &&&
    \text{\bf Orbit 5} 
    &= \left\{ 
        \text{\ref{eq:P5case31}}, \,
        \text{\ref{eq:P5case32}}
    \right\},
    \end{aligned}
\end{gather*}
with respect to the transposition \eqref{tau}.

\subsection{The \PPainleve-4 systems}
The anzats \cite{Bobrova_Sokolov_2021_1} for non-abelian system of $\PIV$ type can be written as
\begin{align} \label{eq:ncP4sys}
    \left\{
    \begin{array}{lcl}
       \ds  \frac{d u}{d z}
         &=& - u^2 
         + 2 u v + \alpha \, [u, v]
         - 2 z u
         + \kappa_2,
         \\[2mm]
          \ds  \frac{d v}{d z}
         &=& - v^2
         + 2 v u + \beta \, [v, u]
         + 2 z v
         + \kappa_3
    \end{array}
    \right.
\end{align}
with the parameters $\alpha,\beta\in \C$ to be defined.
The non-abelian Okamoto integral has the following structure:
\begin{align} \label{eq:ncJ4}
\begin{aligned}
    J
    = a_1 u v^2 + (1 - a_1 - a_2) v u v + a_2 v^2 u
    + b_1 u^2 v + (- 1 - b_1 - b_2) u v u + b_2 v u^2
    \\[1mm]
    - \, \kappa_3 u + \kappa_2 v
    + z \brackets{
    c_1 u v + (- 2 - c_1) v u
    }
    .
\end{aligned}
\end{align}

\begin{proposition} \label{thm:ncJ4list}
There are six systems \text{\rm\ref{eq:P4casem2m1} -- \ref{eq:P4case0m1}} of the form \eqref{eq:ncP4sys} whose auxiliary system have the Okamoto integral of the form \eqref{eq:ncJ4}.  They are listed in Appendix \ref{sec:sysintlistP4}. 
\end{proposition}
\begin{remark} 
 In the papers \text{\rm\cite{Bobrova_Sokolov_2021_1,Bobrova_Sokolov_2021_2}} all systems \eqref{eq:ncP4sys} that satisfy the matrix  \PPainleve \, test have been found. The corresponding pairs $(\alpha,\beta)$ are shown in the following figure
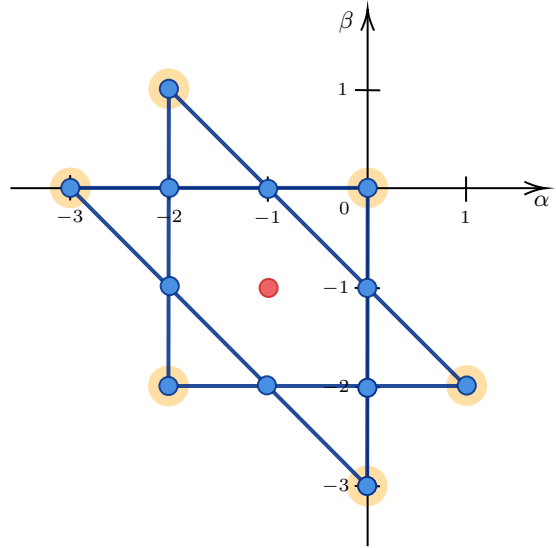
\begin{figure}[H]
    \centering
    \scalebox{1}{\input{pictures/P4/albetPlane.tex}}
    \label{pic:albetplane}
\end{figure}
\hspace{-0.5cm}The systems from Appendix \ref{sec:sysintlistP4} are marked by blue dots without orange rims and belong to the same orbit with respect to the transformation group used in these papers. The central red dot corresponds to the Hamiltonian \text{\rm\ref{eq:P4casem1m1}} system {\rm (}see Appendix \ref{sec:syshamlist}{\rm )} which will appear in Section \ref{sec:deg} devoted to limiting transitions.

The  second order  non-abelian \PPainleve\, equation corresponding to the system with $(\alpha,\beta)=(0,-1)$ was first obtained in the paper \text{\rm\cite{adler2020}}.
\end{remark}

\subsection{\PPainleve-\texorpdfstring{$\mathbf{3^\prime}$}{3'} systems}

In the scalar case, the Hamiltonian $h$ for $\PIIIpr(D_6)$-system \cite{okamoto1987studies}
\begin{align} \label{eq:scalP3D6'sys}
    \left\{
    \begin{array}{lclcl}
         z \, u'
         &=& 2 u^2 v + \kappa_1 u + \kappa_2 u^2
         + \kappa_4 z
         ,
         \\[2mm]
         z \, v'
         &=& - 2 u v^2 - \kappa_1 v
         - 2 \kappa_2 u v - \kappa_3
    \end{array}
    \right.
\end{align}
is given by
\begin{equation} \label{eq:scalJ3D6'}
    z \, h
    = u^2 v^2 + \kappa_1 u v + \kappa_2 u^2 v + \kappa_3 u + \kappa_4 z v
    .
\end{equation}

The auxiliary  system for \eqref{eq:scalP3D6'sys} has the following structure
\begin{align}\label{AuP3}
    \left\{
    \begin{array}{lclcl}
      \ds   \frac{d u}{d t }
         &=& P_1 (u, v)
         + \kappa_4 z
         ,
         \\[3mm]
       \ds   \frac{d v}{d t }
         &=& P_2 (u, v).
    \end{array}
    \right.
\end{align}

\begin{proposition} \label{thm:ncJ3D6list} Let
\begin{align} \label{eq:ncP3D6'sys}
    \begin{aligned}
    P_1 (u, v)
    &= a_1 u^2 v + (2 - a_1 - a_2) u v u + a_2 v u^2 + \kappa_1 u + \kappa_2 u^2
    ,
    \\[2mm]
    P_2 (u, v)
    &= b_1 u v^2 - (2 + b_1 + b_2) v u v + b_2 v^2 u - \kappa_1 v
    + c_1 u v + (- 2 \kappa_2 - c_1) v u
    - \kappa_3
    ,
    \end{aligned}
\end{align}
and
\begin{align} \label{eq:ncJ3D6'}
\begin{aligned}
    J
    = d_1 u^2 v^2 + d_2 u v^2 u + d_3 u v u v + d_4 v u^2 v + d_5 v u v u + \brackets{
    1 - \sum d_i
    } v^2 u^2
    + e_1 u v
    \\[1mm] 
    + \, (\kappa_1 - e_1) v u
    + h_1 u^2 v 
    + (\kappa_2 - h_1 - h_2) u v u
    + h_2 v u^2
    + \kappa_3 u 
    + \kappa_4 z v
    .
\end{aligned}
\end{align}
Then a non-abelian system \eqref{AuP3} has an Okamoto integral $J$ of the form \eqref{eq:ncJ3D6'} iff the corresponding $\PIIIpr$ system belongs to the list \text{\rm\ref{eq:P3D6case1} -- \ref{eq:P3D6case8}} from Appendix \ref{sec:sysintlistP3}.
\end{proposition}

There are four orbits,
\begin{gather*}
    \begin{aligned}
    \text{\bf Orbit 1} 
    &= \left\{ 
        \text{\ref{eq:P3D6case1}}, \,
        \text{\ref{eq:P3D6case2}}
    \right\},
    &&&
    \text{\bf Orbit 2} 
    &= \left\{ 
        \text{\ref{eq:P3D6case5}}, \,
        \text{\ref{eq:P3D6case8}}
    \right\},
    \end{aligned}
    \\[2mm]
    \begin{aligned}
    \text{\bf Orbit 3} 
    &= \left\{ 
        \text{\ref{eq:P3D6case4}}, \,
        \text{\ref{eq:P3D6case6}}
    \right\},
    &&&
    \text{\bf Orbit 4} 
    &= \left\{ 
        \text{\ref{eq:P3D6case3}}, \,
        \text{\ref{eq:P3D6case7}}
    \right\},
    \end{aligned}
\end{gather*}
with respect to the transposition action.

\subsection{\PPainleve-2 systems}
The scalar $\PII$-system
\begin{align} \label{eq:scalP2sys}
    \left\{
    \begin{array}{lcl}
         u'
         &=& - u^2 
         + v
         - \tfrac12 z
         ,
         \\[2mm]
         v'
         &=& 2 u v
         + \kappa_3
    \end{array}
    \right.
\end{align}
has the following Hamiltonian:
\begin{equation} \label{eq:scalJ2}
    h
    = \tfrac12 v^2 
    - u^2 v 
    - \kappa_3 u
    - \tfrac12 z v
    .
\end{equation}
A non-abelian generalization of \eqref{eq:scalP2sys} can be written as
\begin{align} \label{eq:ncP2sys}
    \left\{
    \begin{array}{lcl}
         u'
         &=& - u^2 
         + v
         - \tfrac12 z,
         \\[2mm]
         v'
         &=& 2 v u + \beta [v, u]
         + \kappa_3.
    \end{array}
    \right.
\end{align}
The ansatz for a non-abelian analog of \eqref{eq:scalJ2} is given by
\begin{equation} \label{eq:ncJ2}
    J
    = a_1 u^2 v + (- 1 - a_1 - a_2) u v u + a_2 v u^2
    + \tfrac12 v^2 - \kappa_3 u - \tfrac12 z v.
\end{equation}

\begin{proposition} \label{thm:ncJ2list}
There are two systems of the form \eqref{eq:ncP2sys} that have the Okamoto integral \eqref{eq:ncJ2} \text{\rm(}see Appendix \ref{sec:sysintlistP2}\text{\rm)}.
\end{proposition}

\begin{remark}
These systems are related by transposition \eqref{tau}. There exist two more non-equivalent integrable systems of $\PII$ type \text{\rm(}see  \text{\rm\cite{Adler_Sokolov_2020_1})}.  
\end{remark}

\section{Tree of degenerations}
\label{sec:deg}

Schematically, the degenerations of non-abelian  $\PVI$ systems into systems of types $\PV - \PI$ can be represented as follows. The red arrows correspond to the representatives of the $\PVI$ orbits (see Section  \ref{subsec21}) and their degenerations.
\begin{figure}[H]
    \centering
    \begin{minipage}[l]{0.49\linewidth}
    \begin{figure}[H]
        \hspace{-0.5cm}
        \scalebox{0.8}{\input{pictures/P6/deg_orbit1_P4.tex}}
    \end{figure}
    \end{minipage}
    \begin{minipage}[l]{0.49\linewidth}
    \begin{figure}[H]
        \hspace{0.3cm}
        \scalebox{0.8}{\input{pictures/P6/deg_orbit1_P3.tex}}
    \end{figure}
    \end{minipage}
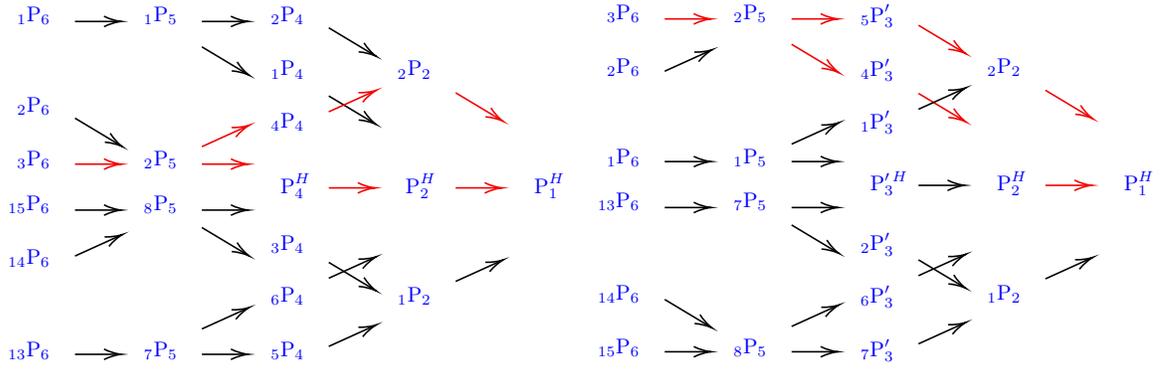
    \caption{The degeneration scheme of the systems from the Orbit \textbf{1}}
    \label{pic:deg_orbit1}
\end{figure}
\hspace{-0.5cm}
Two schemes of Figure \ref{pic:deg_orbit1} show separately chains of degenerations of the types $\PVI \to \PV \to \PIV \to \PII \to \PI$ and $\PVI \to \PV \to \PIIIpr \to \PII \to \PI$ for systems from Orbit \textbf{1}.

The degenerations of $\PVI$ systems from Orbits \textbf{2} and \textbf{3} are shown in Figure \ref{pic:deg_orbit23} below. 
\begin{figure}[H]
    \centering
    \begin{minipage}[l]{0.49\linewidth}
    \begin{figure}[H]
        \hspace{-0.5cm}
        \scalebox{0.8}{\input{pictures/P6/deg_orbit2.tex}}
    \end{figure}
    \end{minipage}
    \begin{minipage}[l]{0.49\linewidth}
    \begin{figure}[H]
        \hspace{0.3cm}
        \scalebox{0.8}{\tikzset{every picture/.style={line width=0.75pt}} 

\begin{tikzpicture}[x=0.75pt,y=0.75pt,yscale=-1,xscale=1]

\draw    (145,219) -- (172.23,219) ;
\draw [shift={(174.23,219)}, rotate = 180] [color={rgb, 255:red, 0; green, 0; blue, 0 }  ][line width=0.75]    (10.93,-3.29) .. controls (6.95,-1.4) and (3.31,-0.3) .. (0,0) .. controls (3.31,0.3) and (6.95,1.4) .. (10.93,3.29)   ;
\draw    (145,203) -- (173.18,190.23) ;
\draw [shift={(175,189.4)}, rotate = 155.61] [color={rgb, 255:red, 0; green, 0; blue, 0 }  ][line width=0.75]    (10.93,-3.29) .. controls (6.95,-1.4) and (3.31,-0.3) .. (0,0) .. controls (3.31,0.3) and (6.95,1.4) .. (10.93,3.29)   ;
\draw    (145,52) -- (172.23,52) ;
\draw [shift={(174.23,52)}, rotate = 180] [color={rgb, 255:red, 0; green, 0; blue, 0 }  ][line width=0.75]    (10.93,-3.29) .. controls (6.95,-1.4) and (3.31,-0.3) .. (0,0) .. controls (3.31,0.3) and (6.95,1.4) .. (10.93,3.29)   ;
\draw    (145,68) -- (173.18,84.93) ;
\draw [shift={(174.89,85.97)}, rotate = 211.01] [color={rgb, 255:red, 0; green, 0; blue, 0 }  ][line width=0.75]    (10.93,-3.29) .. controls (6.95,-1.4) and (3.31,-0.3) .. (0,0) .. controls (3.31,0.3) and (6.95,1.4) .. (10.93,3.29)   ;
\draw [color={rgb, 255:red, 243; green, 0; blue, 0 }  ,draw opacity=1 ]   (145,121) -- (172.23,121) ;
\draw [shift={(174.23,121)}, rotate = 180] [color={rgb, 255:red, 243; green, 0; blue, 0 }  ,draw opacity=1 ][line width=0.75]    (10.93,-3.29) .. controls (6.95,-1.4) and (3.31,-0.3) .. (0,0) .. controls (3.31,0.3) and (6.95,1.4) .. (10.93,3.29)   ;
\draw [color={rgb, 255:red, 243; green, 0; blue, 0 }  ,draw opacity=1 ]   (145,150) -- (172.23,150) ;
\draw [shift={(174.23,150)}, rotate = 180] [color={rgb, 255:red, 243; green, 0; blue, 0 }  ,draw opacity=1 ][line width=0.75]    (10.93,-3.29) .. controls (6.95,-1.4) and (3.31,-0.3) .. (0,0) .. controls (3.31,0.3) and (6.95,1.4) .. (10.93,3.29)   ;
\draw    (65,219) -- (92.23,219) ;
\draw [shift={(94.23,219)}, rotate = 180] [color={rgb, 255:red, 0; green, 0; blue, 0 }  ][line width=0.75]    (10.93,-3.29) .. controls (6.95,-1.4) and (3.31,-0.3) .. (0,0) .. controls (3.31,0.3) and (6.95,1.4) .. (10.93,3.29)   ;
\draw    (65,52) -- (92.23,52) ;
\draw [shift={(94.23,52)}, rotate = 180] [color={rgb, 255:red, 0; green, 0; blue, 0 }  ][line width=0.75]    (10.93,-3.29) .. controls (6.95,-1.4) and (3.31,-0.3) .. (0,0) .. controls (3.31,0.3) and (6.95,1.4) .. (10.93,3.29)   ;
\draw    (65,150) -- (92.23,150) ;
\draw [shift={(94.23,150)}, rotate = 180] [color={rgb, 255:red, 0; green, 0; blue, 0 }  ][line width=0.75]    (10.93,-3.29) .. controls (6.95,-1.4) and (3.31,-0.3) .. (0,0) .. controls (3.31,0.3) and (6.95,1.4) .. (10.93,3.29)   ;
\draw [color={rgb, 255:red, 243; green, 0; blue, 0 }  ,draw opacity=1 ]   (65,121) -- (92.23,121) ;
\draw [shift={(94.23,121)}, rotate = 180] [color={rgb, 255:red, 243; green, 0; blue, 0 }  ,draw opacity=1 ][line width=0.75]    (10.93,-3.29) .. controls (6.95,-1.4) and (3.31,-0.3) .. (0,0) .. controls (3.31,0.3) and (6.95,1.4) .. (10.93,3.29)   ;
\draw    (65,190) -- (93.18,206.93) ;
\draw [shift={(94.89,207.97)}, rotate = 211.01] [color={rgb, 255:red, 0; green, 0; blue, 0 }  ][line width=0.75]    (10.93,-3.29) .. controls (6.95,-1.4) and (3.31,-0.3) .. (0,0) .. controls (3.31,0.3) and (6.95,1.4) .. (10.93,3.29)   ;
\draw [color={rgb, 255:red, 243; green, 0; blue, 0 }  ,draw opacity=1 ]   (225,121) -- (252.23,121) ;
\draw [shift={(254.23,121)}, rotate = 180] [color={rgb, 255:red, 243; green, 0; blue, 0 }  ,draw opacity=1 ][line width=0.75]    (10.93,-3.29) .. controls (6.95,-1.4) and (3.31,-0.3) .. (0,0) .. controls (3.31,0.3) and (6.95,1.4) .. (10.93,3.29)   ;
\draw    (225,211) -- (253.18,198.23) ;
\draw [shift={(255,197.4)}, rotate = 155.61] [color={rgb, 255:red, 0; green, 0; blue, 0 }  ][line width=0.75]    (10.93,-3.29) .. controls (6.95,-1.4) and (3.31,-0.3) .. (0,0) .. controls (3.31,0.3) and (6.95,1.4) .. (10.93,3.29)   ;
\draw    (225,56) -- (253.18,72.93) ;
\draw [shift={(254.89,73.97)}, rotate = 211.01] [color={rgb, 255:red, 0; green, 0; blue, 0 }  ][line width=0.75]    (10.93,-3.29) .. controls (6.95,-1.4) and (3.31,-0.3) .. (0,0) .. controls (3.31,0.3) and (6.95,1.4) .. (10.93,3.29)   ;
\draw    (225,82) -- (253.18,98.93) ;
\draw [shift={(254.89,99.97)}, rotate = 211.01] [color={rgb, 255:red, 0; green, 0; blue, 0 }  ][line width=0.75]    (10.93,-3.29) .. controls (6.95,-1.4) and (3.31,-0.3) .. (0,0) .. controls (3.31,0.3) and (6.95,1.4) .. (10.93,3.29)   ;
\draw    (225,187) -- (253.18,174.23) ;
\draw [shift={(255,173.4)}, rotate = 155.61] [color={rgb, 255:red, 0; green, 0; blue, 0 }  ][line width=0.75]    (10.93,-3.29) .. controls (6.95,-1.4) and (3.31,-0.3) .. (0,0) .. controls (3.31,0.3) and (6.95,1.4) .. (10.93,3.29)   ;
\draw    (65,81) -- (93.18,68.23) ;
\draw [shift={(95,67.4)}, rotate = 155.61] [color={rgb, 255:red, 0; green, 0; blue, 0 }  ][line width=0.75]    (10.93,-3.29) .. controls (6.95,-1.4) and (3.31,-0.3) .. (0,0) .. controls (3.31,0.3) and (6.95,1.4) .. (10.93,3.29)   ;
\draw [color={rgb, 255:red, 243; green, 0; blue, 0 }  ,draw opacity=1 ]   (225,150) -- (252.23,150) ;
\draw [shift={(254.23,150)}, rotate = 180] [color={rgb, 255:red, 243; green, 0; blue, 0 }  ,draw opacity=1 ][line width=0.75]    (10.93,-3.29) .. controls (6.95,-1.4) and (3.31,-0.3) .. (0,0) .. controls (3.31,0.3) and (6.95,1.4) .. (10.93,3.29)   ;
\draw [color={rgb, 255:red, 243; green, 0; blue, 0 }  ,draw opacity=1 ]   (305,135) -- (332.23,135) ;
\draw [shift={(334.23,135)}, rotate = 180] [color={rgb, 255:red, 243; green, 0; blue, 0 }  ,draw opacity=1 ][line width=0.75]    (10.93,-3.29) .. controls (6.95,-1.4) and (3.31,-0.3) .. (0,0) .. controls (3.31,0.3) and (6.95,1.4) .. (10.93,3.29)   ;

\draw (193,111.4) node [anchor=north west][inner sep=0.75pt]    {\ref{eq:P4casem1m1}};
\draw (193,139.4) node [anchor=north west][inner sep=0.75pt]    {\ref{eq:hamP3D6case1}};
\draw (187,75.4) node [anchor=north west][inner sep=0.75pt]    {\ref{eq:P3D6case6}};
\draw (187,41.4) node [anchor=north west][inner sep=0.75pt]    {\ref{eq:P4casem2m1}};
\draw (187,175.4) node [anchor=north west][inner sep=0.75pt]    {\ref{eq:P3D6case4}};
\draw (187,209.4) node [anchor=north west][inner sep=0.75pt]    {\ref{eq:P4case0m1}};
\draw (113,125.4) node [anchor=north west][inner sep=0.75pt]    {\ref{eq:hamP5case1}};
\draw (107,41.4) node [anchor=north west][inner sep=0.75pt]    {\ref{eq:P5case32}};
\draw (107,209.4) node [anchor=north west][inner sep=0.75pt]    {\ref{eq:P5case31}};
\draw (27,41.4) node [anchor=north west][inner sep=0.75pt]    {\ref{eq:P6case32}};
\draw (21.5,75.4) node [anchor=north west][inner sep=0.75pt]    {\ref{eq:P6case36}};
\draw (27,111.4) node [anchor=north west][inner sep=0.75pt]    {\ref{eq:P6case31}};
\draw (21.5,139.4) node [anchor=north west][inner sep=0.75pt]    {\ref{eq:P6case34}};
\draw (27,175.4) node [anchor=north west][inner sep=0.75pt]    {\ref{eq:P6case33}};
\draw (21.5,209.4) node [anchor=north west][inner sep=0.75pt]    {\ref{eq:P6case35}};
\draw (270,125.4) node [anchor=north west][inner sep=0.75pt]    {\ref{eq:P2betam1}};
\draw (350,125.4) node [anchor=north west][inner sep=0.75pt]    {\ref{eq:P1ham}};

\end{tikzpicture}}
    \end{figure}
    \end{minipage}
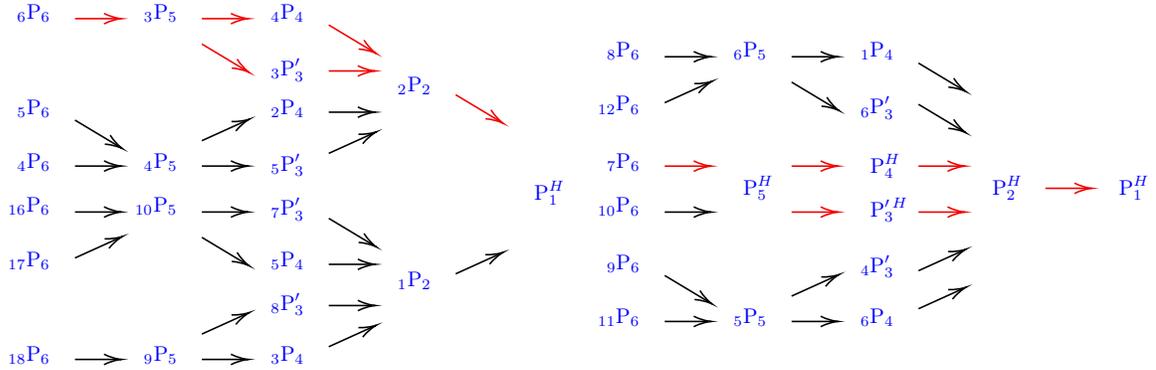
    \caption{The degeneration scheme of the systems from the Orbits \textbf{2} and \textbf{3}}
    \label{pic:deg_orbit23}
\end{figure}

\subsection{Description of degenerations}
\label{sec:degdata}
In this section, we present explicit formulas  corresponding to the arrows in Figures \ref{pic:deg_orbit1} and \ref{pic:deg_orbit23}. In addition to all systems that possess the Okamoto integrals, the non-abelian Hamiltonian  \Painleve systems of types $\PV-\PI$ appear here (see Appendix \ref{sec:syshamlist}). They are denoted by $\text{P}_i^H$.

We also describe limiting transitions for isomonodromic Lax pairs. The Lax representations for all systems shown in Figures \ref{pic:deg_orbit1} and \ref{pic:deg_orbit23} can be obtained from ones for the $\PVI$ systems (see Section \ref{Isomonodromic}). 

\begin{remark} 
\label{rem:degcorr}
In the explicit formulas for the Lax pair degenerations $\PIV \to \PII$, $\PIIIpr \to \PII$ and 
$\PII \to \PI$ the function $g$ contains a constant different for different systems.  
If we take this constant equal to zero, then in some cases we get a Lax pair that is not polynomial on $\varepsilon$ but has a pole at $\varepsilon=0$. In such a situation, we choose a constant so that the pole disappears, and after that we set $\varepsilon=0$. 

Notice that for branches starting with the three representatives and marked in red, all constants are zero. The same is true for all branches starting with 6 systems connected to the representatives via $r_1,r_2$, but not via the transposition \eqref{tau}.
 \end{remark}

\subsubsection{
\texorpdfstring{
$\PVI \to \PV$
}{P6 -> P5}
}

The systems of $\PVI$ type can be reduced to $\PV$ systems by the following map with the small parameter $\varepsilon$:
\begin{align} \label{eq:P6toP5map}
    z 
    &\mapsto \varepsilon^{-1} (z - 1),
    &
    \kappa_2 
    &\mapsto - \kappa_1 + \kappa_2 + \kappa_4,
    &
    \kappa_4
    &\mapsto - \varepsilon \kappa_1 + \varepsilon \, \kappa_4.
\end{align}
To obtain the corresponding Lax pair, we supplement the transformation \eqref{eq:P6toP5map} by 
\begin{align} \label{eq:P6toP5map_la}
    \lambda 
    &\mapsto (z - 1)^{-1} (\lambda - 1).
\end{align}
As a result, a pair of the form \eqref{eq:matABform} degenerates to a Lax pair of the following structure:
\begin{align} \label{eq:matABform_P5}
    \mathbf{A} (\lambda, z)
    &= A_0
    + \frac{A_1}{\lambda}
    + \frac{A_2}{\lambda - 1},
    &
    \mathbf{B} (\lambda, z)
    &= B_1 \lambda
    + B_{0}
\end{align}
 for a $\PV$ system.

\begin{remark} 
In the formulas \eqref{eq:matABform}, \eqref{eq:matABform_P5}, \eqref{eq:matABform_P4}, \eqref{eq:matABform_P3D6'_}, \eqref{eq:matABform_P2}, \eqref{eq:matABform_P1} we indicate only the dependence of $\bf A$ and $\bf B$~on~$\lambda$. We hope that the use of the same notation for different coefficients $A_i$ and $B_i$ in these formulas will not lead to misunderstanding.
\end{remark}

In the following example we demonstrate the mechanism for the appearance of Hamiltonian systems in the process of limiting transitions.
\begin{example}[\text{\rm\ref{eq:P6case31}} $\to$ \text{\rm\ref{eq:hamP5case1}}]
After changing the variables and parameters given by \eqref{eq:P6toP5map},  system \text{\rm\ref{eq:P6case31}} and the corresponding Okamoto integral turn into
\begin{align}
    \left\{
    \begin{array}{lcr}
         z (1 + \varepsilon \, z) \, u'
         &=& u^2 v u + u v u^2
         - u^2 v - 2 u v u - v u^2
         - \kappa_1 u^2 
         + u v + v u
         + (\kappa_1 + \kappa_2) u
         \\[1mm]
         &&
         - \, \kappa_2
         + \kappa_4 \, z u
         + \varepsilon \, z \left(
         - u v u - v u^2
         + \, u v
         + v u
         + \kappa_1 u
         - \kappa_2
         \right),
         \\[2mm]
         z (1 + \varepsilon \, z) \, v'
         &=& - u v u v - v u^2 v - v u v u
         + u v^2 + 2 v u v + v^2 u
         - v^2 
         + \kappa_1 u v
         + \kappa_1 v u
         \hspace{3mm}
         \\[1mm]
         && 
         - \, (\kappa_1 + \kappa_2) v
         + \kappa_3
         - \kappa_4 \, z v
         + \varepsilon \, z \left(
         v u v + v^2 u
         - v^2
         - \kappa_1 v
         \right),
    \end{array}
    \right.
\end{align}
\begin{align}
\begin{aligned}
    J
    = - \kappa_4 \, \varepsilon^{-1} \brackets{
    u v - v u
    }
    + u v u v u
    - u v u v - v u v u
    - \kappa_1 u v u
    + v u v
    + \kappa_1 v u + \kappa_2 u v 
    \\[1mm]
    - \, \kappa_2 v
    - \kappa_3 u
    + \kappa_4 \, z v u
    + \varepsilon \, z \brackets{
    - v u v u
    + v u v
    + \kappa_1 v u
    - \kappa_2 v
    }.
\end{aligned}
\end{align}
Under the limit $\varepsilon \to 0$, the system becomes the Hamiltonian \text{\rm\ref{eq:hamP5case1}} system
\begin{align}
    \left\{
    \begin{array}{lcr}
         z \, u'
         &=& u^2 v u + u v u^2
         - u^2 v - 2 u v u - v u^2
         - \kappa_1 u^2 
         + u v + v u
         \hspace{0.7cm}
         \\[1mm]
         &&
         + \, (\kappa_1 + \kappa_2) u
         - \kappa_2
         + \kappa_4 \, z u,
         \\[2mm]
         z \, v'
         &=& - u v u v - v u^2 v - v u v u
         + u v^2 + 2 v u v + v^2 u
         - v^2 
         + \, \kappa_1 u v
         \\[1mm]
         && 
         + \kappa_1 v u
         - (\kappa_1 + \kappa_2) v
         + \kappa_3
         - \kappa_4 \, z v
    \end{array}
    \right.
\end{align}
and the integral degenerates to {\rm (}cf. Lemma \ref{thm:hamint}{\rm)} 
\begin{align}
    I 
    = u v - v u.
\end{align}

To get a Lax pair for the resulting system, in addition to the above transformation, we should replace the spectral parameter $\lambda$ by $1 + \varepsilon \, z \, \lambda$. As a result, the Lax pair \eqref{eq:Laxpair_P6case31} takes the form \eqref{eq:matABform_P5}, where matrices $A_0$, $A_1$, $A_2$, $B_1$, and $B_0$ are given by
\begin{gather} 
    \notag
    \begin{aligned}
    A_{0}
    &= 
    \begin{pmatrix}
    \kappa_4 z & 0 \\[0.9mm] 0 & 0
    \end{pmatrix},
    &
    A_1
    &= 
    \begin{pmatrix}
    - u v + \kappa_1 & 1
    \\[0.9mm]
    - u v u v
    + \kappa_1 u v
    + \kappa_3
    & u v
    \end{pmatrix},
    &
    A_2
    &= 
    \begin{pmatrix}
    u v - \kappa_2 & - u
    \\[0.9mm]
    u v^2 - \kappa_2 v & - u v
    \end{pmatrix},
    \end{aligned}
    \\[-1mm]
    \label{eq:Laxpair_P50}
    \\[-1mm]
    \notag
    \begin{aligned}
    B_1
    &= 
    \begin{pmatrix}
    \kappa_4 & 0 \\[0.9mm] 0 & 0
    \end{pmatrix},
    &
    B_{0}
    &= z^{-1}
    \begin{pmatrix}
    u^2 v + u v u
    - 2 u v
    - v u 
    + v
    - \kappa_1 u 
    + \kappa_1
    & 
    - u + 1
    \\[0.9mm]
    - u v u v 
    + u v^2
    + \kappa_1 u v 
    - \kappa_2 v
    + \kappa_3
    &
    - v u + v
    \end{pmatrix}.
    \end{aligned}
\end{gather}
\end{example}

\subsubsection{
\texorpdfstring{
$\PV \to \PIV$
}{P5 -> P4}
}

In the case of $\PV \to \PIV$, one can consider the transformation 
\begin{align}
    \label{eq:P5toP4map}
    z 
    &\mapsto \tfrac{1}{\sqrt{2}} \, \varepsilon^{-1} (z - 1),
    &&&
    u 
    &\mapsto \sqrt{2} \, \varepsilon^{-1} u,
    &&&
    v
    &\mapsto \sqrt{2} \, \varepsilon \, v,
\end{align}
taking together with the following changes of the parameters $\kappa_i$:
\begin{align*}
    \kappa_1
    &= - \kappa_3 + 2 \varepsilon^{-2},
    &&&
    \kappa_2 
    &\mapsto - 2 \kappa_2,
    &&&
    \kappa_3
    &\mapsto \phantom{-} 2 \kappa_3 - 2 \varepsilon^{-2},
    &&&
    \kappa_4
    &= - \varepsilon^{-2};
    \\[2mm]
    \kappa_1
    &= \phantom{-} \kappa_3 + 2 \varepsilon^{-2},
    &&&
    \kappa_2 
    &\mapsto - 2 \kappa_2,
    &&&
    \kappa_3
    &\mapsto - 2 \kappa_3 - 2 \varepsilon^{-2},
    &&&
    \kappa_4
    &= - \varepsilon^{-2}
\end{align*}
for Cases \textbf{1} and \textbf{4}; and 
\begin{align}
    \label{eq:P5toP4map_k}
    \kappa_1
    &= \varepsilon^{-2},
    &&&
    \kappa_2 
    &\mapsto - 2 \kappa_2,
    &&&
    \kappa_3
    &\mapsto 2 \varepsilon^2 \kappa_3,
    &&&
    \kappa_4
    &= - \varepsilon^{-2}
\end{align}
for Cases \textbf{2}, \textbf{3}, \textbf{5}.
For the degeneration of Lax pairs, in addition to the above formulas, consider the following mapping:
\begin{align}
    \label{eq:P5toP4map_la}
    \lambda 
    &\mapsto \tfrac{1}{\sqrt{2}} \, \varepsilon^{-1} (\lambda - 1),
    &
    \mathbf{A}
    &\mapsto g \, \mathbf{A} \, g^{-1},
    &
    \mathbf{B}
    &\mapsto g \, \mathbf{B} \, g^{-1},
    &
    g
    &= 
    \begin{pmatrix}
    1 & 0 \\ 0 & \sqrt{2} \, \varepsilon
    \end{pmatrix},
\end{align}
which brings pair \eqref{eq:matABform_P5} to the form
\begin{align} \label{eq:matABform_P4}
    \mathbf{A} (\lambda, z)
    &= A_1 \lambda + A_{0} + A_{-1} \lambda^{-1},
    &
    \mathbf{B} (\lambda, z)
    &= B_1 \lambda + B_{0}.
\end{align}

\subsubsection{
\texorpdfstring{
$\PV \to \PIIIpr$
}{P5 -> P3'}
}

The mapping
\begin{gather}
\label{eq:P5toP3D6'map}
\begin{gathered}
    \begin{aligned}
    u 
    &\mapsto \varepsilon^{-1} (u - 1),
    &&&
    v
    &\mapsto \varepsilon \, v,
    \end{aligned}
\end{gathered}
\end{gather}
together with the following change of the parameters:
\begin{align*}
    \kappa_1
    &\mapsto - \kappa_1 + \kappa_2, 
    &&&
    \kappa_2 
    &\mapsto - \tfrac12 \varepsilon(\kappa_1 - \kappa_3), 
    &&&
    \kappa_3 
    &\mapsto \tfrac14 \varepsilon(\kappa_1^2 - \kappa_3^2), 
    &&&
    \kappa_4 
    &\mapsto \varepsilon^{-1} \kappa_4;
    \\[2mm]
    \kappa_1
    &\mapsto -\kappa_1 + \kappa_2, 
    &&&
    \kappa_2 
    &\mapsto - \tfrac12 \varepsilon(\kappa_1 + \kappa_3), 
    &&&
    \kappa_3 
    &\mapsto \tfrac14 \varepsilon(\kappa_1^2 - \kappa_3^2), 
    &&&
    \kappa_4 
    &\mapsto \varepsilon^{-1} \kappa_4
\end{align*}
for Cases \textbf{1} and \textbf{4}; and 
\begin{align}
\label{eq:P5toP3D6'map_k}
    \kappa_1
    &\mapsto - \kappa_1 + \kappa_2,
    &&&
    \kappa_2 
    &\mapsto - \varepsilon \, \kappa_2,
    &&&
    \kappa_3
    &\mapsto - \varepsilon \, \kappa_3,
    &&&
    \kappa_4
    &\mapsto \varepsilon^{-1} \, \kappa_4
\end{align}
for Cases \textbf{2}, \textbf{3}, \textbf{5}
describes the degeneration $\PV \to \PIIIpr$. Supplementing the latter formulas with 
\begin{align}
    \label{eq:P5toP3D6'map_la}
    \lambda 
    &\mapsto \varepsilon \, \lambda,
    &
    \mathbf{A}
    &\mapsto g \, \mathbf{A} \, g^{-1},
    &
    \mathbf{B}
    &\mapsto g \, \mathbf{B} \, g^{-1}
    ,
    &
    g
    &= 
    \begin{pmatrix}
    1 & 0 \\[0.9mm] 0 & \varepsilon
    \end{pmatrix},
\end{align}
and taking the limit $\varepsilon \to 0$, we bring matrices \eqref{eq:matABform_P5} to the form
\begin{align} \label{eq:matABform_P3D6'_}
    \mathbf{A} (\lambda, z)
    &= A_0 
    + A_{-1} \lambda^{-1} 
    + A_{-2} \lambda^{-2},
    &
    \mathbf{B} (\lambda, z)
    &= B_1 \lambda + B_0.
\end{align}

\subsubsection{
\texorpdfstring{
$\PIV \to \PII$
}{P4 -> P2}
}

Putting
\begin{gather}
\label{eq:P4toP2map}
\begin{gathered}
    \begin{aligned}
    z 
    &\mapsto \tfrac14 \varepsilon^{-4} - \varepsilon^{-1} z,
    &&&
    u 
    &\mapsto - \tfrac14 \varepsilon^{-2} - \varepsilon \, u,
    &&&
    v
    &\mapsto - \tfrac12 \varepsilon^{-1} v,
    \end{aligned}
    \\[1mm]
    \begin{aligned}
    \kappa_2
    &= - \tfrac{1}{16} \varepsilon^{-6},
    &&&
    \kappa_3
    &\mapsto \tfrac12 \kappa_3
    \end{aligned}
\end{gathered}
\end{gather}
and taking the limit $\varepsilon \to 0$, one can obtain the degeneration $\PIV \to \PII$ for $\PIV$ systems.
The corresponding degeneration for a pair \eqref{eq:matABform_P4} is given by
\begin{align}
\label{eq:P4toP2map_la}
\begin{gathered}
    \begin{aligned}
    \lambda
    &\mapsto \tfrac14 \varepsilon^{-2} + 2 \varepsilon \lambda,
    \end{aligned}
    \\
    \begin{aligned}
    \mathbf{A}
    &\mapsto g \, \mathbf{A} \, g^{-1},
    &
    \mathbf{B}
    &\mapsto g \, \mathbf{B} \, g^{-1}
    + g_z' \, g^{-1},
    &
    g
    &= e^{\const \, \varepsilon^{-3} z}
    \begin{pmatrix}
    1 & 0 \\[0.9mm]
    - \varepsilon \, v & \varepsilon
    \end{pmatrix}.
    \end{aligned}
\end{gathered}
\end{align}
As a result, we obtain  a pair
\begin{align} \label{eq:matABform_P2}
    \mathbf{A} (\lambda, z)
    &= A_2 \lambda^2 + A_1 \lambda + A_0,
    &
    \mathbf{B} (\lambda, z)
    &= B_1 \lambda + B_0
\end{align}
of the Jimbo-Miwa type for the corresponding non-abelian $\PII$ system.

\subsubsection{
\texorpdfstring{
$\PIIIpr \to \PII$
}{P3' -> P2}
}
The following map
\begin{gather}
\label{eq:P3D6'toP2map}
\begin{gathered} 
    \begin{aligned}
    z 
    &\mapsto - \tfrac12 \varepsilon^{-2} - \tfrac12 \varepsilon z
    ,
    &&&
    u 
    &\mapsto \tfrac12 \varepsilon^{-1} \brackets{
    u - 1
    }
    ,
    &&&
    v
    &\mapsto 2 \varepsilon v,
    \end{aligned}
    \\[1mm]
    \begin{aligned}
    \kappa_1
    &= \tfrac12 \varepsilon^{-3},
    &&&
    \kappa_2
    &= - \tfrac14 \varepsilon^{-3},
    &&&
    \kappa_3
    &\mapsto - 4 \varepsilon^{3} \kappa_3,
    &&&
    \kappa_4
    &= \tfrac14
    \end{aligned}
\end{gathered}
\end{gather}
reduces $\PIIIpr$ type systems to systems of $\PII$ type.
To get the corresponding Lax pair, we consider the transformation
\begin{align}
\label{eq:P3D6'toP2map_la}
\begin{gathered}
    \begin{aligned}
    \lambda
    &\mapsto - \tfrac12 \, \varepsilon^{-1} (\lambda + 1),
    \end{aligned}
    \\
    \begin{aligned}
    \mathbf{A}
    &\mapsto g \, \mathbf{A} \, g^{-1},
    &&&
    \mathbf{B}
    &\mapsto g \, \mathbf{B} \, g^{-1}
    + g_z' \, g^{-1}
    ,
    &&&
    g
    &= e^{\const \, z}
    \begin{pmatrix}
    1 & 0 \\ 0 & 2 \varepsilon^2
    \end{pmatrix},
    \end{aligned}
\end{gathered}
\end{align}
that, after taking the limit $\varepsilon \to 0$, brings \eqref{eq:matABform_P3D6'_} to the form \eqref{eq:matABform_P2}.

\subsubsection{
\texorpdfstring{
$\PII \to \PI$
}{P2 -> P1}
}
The degeneration of $\PII$ systems is given by
\begin{gather}
\label{eq:P2toP1map}
\begin{gathered}
    \begin{aligned}
    z 
    &\mapsto \varepsilon^{-2} z + 6 \varepsilon^{-12},
    &&&
    u 
    &\mapsto - \varepsilon^{-1} u + \varepsilon^4 v + 3 \varepsilon^{-6}
    ,
    &&&
    v
    &\mapsto - 2 \varepsilon^{-4} u + \varepsilon v + 4 \varepsilon^{-9},
    \end{aligned}
    \\[1mm]
    \begin{aligned}
    \kappa_3
    &= 4 \varepsilon^{-15}.
    \end{aligned}
\end{gathered}
\end{gather}
Combining this mapping with the following degeneration formulas 
\begin{gather}
\label{eq:P2toP1map_la}
\begin{gathered}
    \lambda
    \mapsto \varepsilon^{-6} + \varepsilon^{-1} \lambda,
    \\
    \begin{aligned}
    \mathbf{A}
    &\mapsto g \, \mathbf{A} \, g^{-1}
    + g_{\lambda}' \, g^{-1}
    ,
    &&&
    \mathbf{B}
    &\mapsto g \, \mathbf{B} \, g^{-1}
    + g_z' \, g^{-1}
    ,
    \end{aligned}
    \\
    \begin{aligned}
    g
    &= e^{
    4 \varepsilon^{-10} \lambda
    + \varepsilon^{-5} \lambda^2
    + \const \, \varepsilon^{-5} z
    }
    \begin{pmatrix}
    1 & 0 \\[0.9mm]
    \varepsilon^2 \, u & \varepsilon^2
    \end{pmatrix},
    \end{aligned}
\end{gathered}
\end{gather}
we obtain from \eqref{eq:matABform_P2} a Lax pair of the form 
\begin{align} \label{eq:matABform_P1}
    \mathbf{A} (\lambda, z)
    &= A_2 \lambda^2 + A_1 \lambda + A_0,
    &
    \mathbf{B} (\lambda, z)
    &= B_1 \lambda + B_0
\end{align}
for the corresponding $\PI$ type system.

\begin{example}
Below we list, as an example, the degenerations of Lax pairs for system {\rm\ref{eq:P6case32}} {\rm(}see Example \ref{exm:r1deg}{\rm)}.
The scheme of the degenerations is given in Figure \ref{pic:deg_8P6}{\rm:}
\begin{figure}[H]
    \centering
    \scalebox{1.}{\tikzset{every picture/.style={line width=0.75pt}} 

\begin{tikzpicture}[x=0.75pt,y=0.75pt,yscale=-1,xscale=1]

\draw    (32.65,36.5) -- (59.88,36.5) ;
\draw [shift={(61.88,36.5)}, rotate = 180] [color={rgb, 255:red, 0; green, 0; blue, 0 }  ][line width=0.75]    (10.93,-3.29) .. controls (6.95,-1.4) and (3.31,-0.3) .. (0,0) .. controls (3.31,0.3) and (6.95,1.4) .. (10.93,3.29)   ;
\draw    (99.3,24.46) -- (128.8,13.68) ;
\draw [shift={(130.67,12.99)}, rotate = 159.93] [color={rgb, 255:red, 0; green, 0; blue, 0 }  ][line width=0.75]    (10.93,-3.29) .. controls (6.95,-1.4) and (3.31,-0.3) .. (0,0) .. controls (3.31,0.3) and (6.95,1.4) .. (10.93,3.29)   ;
\draw    (164.3,64.46) -- (193.8,51.68) ;
\draw [shift={(195.67,50.5)}, rotate = 152.93] [color={rgb, 255:red, 0; green, 0; blue, 0 }  ][line width=0.75]    (10.93,-3.29) .. controls (6.95,-1.4) and (3.31,-0.3) .. (0,0) .. controls (3.31,0.3) and (6.95,1.4) .. (10.93,3.29)   ;
\draw    (99.3,49.46) -- (128.72,63.72) ;
\draw [shift={(130.52,64.59)}, rotate = 205.86] [color={rgb, 255:red, 0; green, 0; blue, 0 }  ][line width=0.75]    (10.93,-3.29) .. controls (6.95,-1.4) and (3.31,-0.3) .. (0,0) .. controls (3.31,0.3) and (6.95,1.4) .. (10.93,3.29)   ;
\draw    (165.3,11.46) -- (194.72,25.72) ;
\draw [shift={(196.52,26.59)}, rotate = 205.86] [color={rgb, 255:red, 0; green, 0; blue, 0 }  ][line width=0.75]    (10.93,-3.29) .. controls (6.95,-1.4) and (3.31,-0.3) .. (0,0) .. controls (3.31,0.3) and (6.95,1.4) .. (10.93,3.29)   ;
\draw    (229.65,36.5) -- (256.88,36.5) ;
\draw [shift={(258.88,36.5)}, rotate = 180] [color={rgb, 255:red, 0; green, 0; blue, 0 }  ][line width=0.75]    (10.93,-3.29) .. controls (6.95,-1.4) and (3.31,-0.3) .. (0,0) .. controls (3.31,0.3) and (6.95,1.4) .. (10.93,3.29)   ;

\draw (4.02,29.79) node [anchor=north west][inner sep=0.75pt]    
{\ref{eq:P6case32}};
\draw (70.02,29.8) node [anchor=north west][inner sep=0.75pt]    
{\ref{eq:P5case32}};
\draw (135.8,4.04) node [anchor=north west][inner sep=0.75pt]   
{\ref{eq:P3D6case6}};
\draw (135.8,57.4) node [anchor=north west][inner sep=0.75pt]    
{\ref{eq:P4casem2m1}};
\draw (202.77,26.8) node [anchor=north west][inner sep=0.75pt]    
{\ref{eq:P2betam1}};
\draw (267.02,26.8) node [anchor=north west][inner sep=0.75pt]    
{\ref{eq:P1ham}};

\end{tikzpicture}}
    \caption{Degenerations of \ref{eq:P6case32}}
    \label{pic:deg_8P6}
\end{figure}

\medskip

{\rm
\textbullet \,\,
$\text{\ref{eq:P6case32}} \to \text{\ref{eq:P5case32}}$.}
The limiting transition given by formula \eqref{eq:P6toP5map}  transforms Lax pair \eqref{eq:Laxpair_P6case32} to a pair of the form \eqref{eq:matABform_P5} with
\begin{gather} 
    \notag
    \begin{aligned}
    A_{0}
    &= 
    \begin{pmatrix}
    \kappa_4 z & 0 \\[0.9mm] 0 & 0
    \end{pmatrix},
    &
    A_1
    &= 
    \begin{pmatrix}
    - u v + \kappa_1 & 1
    \\[0.9mm]
    - u v u v
    + \kappa_1 u v 
    + \kappa_3
    & u v
    \end{pmatrix},
    &
    A_2
    &= 
    \begin{pmatrix}
    u v - \kappa_2
    & - u
    \\[0.9mm]
    v u v
    - \kappa_2 v
    & 
    - v u
    \end{pmatrix},
    \end{aligned}
    \\[-1mm]
    \label{eq:Laxpair_P5case32}
    \\[-1mm]
    \notag
    \begin{aligned}
    B_1
    &= 
    \begin{pmatrix}
    \kappa_4 & 0 \\[0.9mm] 0 & 0
    \end{pmatrix},
    &
    B_{0}
    &= z^{-1}
    \begin{pmatrix}
    u^2 v
    + u v u
    - 2 u v 
    - v u
    - \kappa_1 u
    + v
    + \kappa_1
    & 
    - u + 1
    \\[0.9mm]
    - u v u v
    + v u v
    + \kappa_1 u v
    - \kappa_2 v
    + \kappa_3 
    &
    - v u + v
    \end{pmatrix}
    \end{aligned}
\end{gather}
 for the {\rm\ref{eq:P5case32}} system.

{\rm
\textbullet \,\,
$\text{\ref{eq:P5case32}} \to \text{\ref{eq:P4casem2m1}}$.} 
Using the map \eqref{eq:P5toP4map}, \eqref{eq:P5toP4map_k}, \eqref{eq:P5toP4map_la}, we reduce Lax pair \eqref{eq:Laxpair_P5case32} to \eqref{eq:matABform_P4}, where
\begin{equation} \label{eq:Laxpair_P4casem2m1}
    \begin{gathered}
        \begin{aligned}
        A_1
        &= 
        \begin{pmatrix}
        - 2 & 0 \\[0.9mm] 0 & 0
        \end{pmatrix},
        &&&
        A_0
        &= 
        \begin{pmatrix}
        - 2 z & 1 \\[0.9mm] 
        u v + \kappa_3 & 0
        \end{pmatrix},
        &&&
        A_{-1}    
        &= \tfrac12
        \begin{pmatrix}
        u v + \kappa_2 & - u \\[0.9mm]
        v u v + \kappa_2 v & - v u
        \end{pmatrix},
        \end{aligned}
        \\[2mm]
        \begin{aligned}
        B_1
        &= 
        \begin{pmatrix}
        - 2 & 0 \\[0.9mm] 0 & 0
        \end{pmatrix},
        &&&
        B_0
        &= 
        \begin{pmatrix}
        - u + v - 2 z & 1 \\[0.9mm]
        u v + \kappa_3 & v
        \end{pmatrix}.
        \end{aligned}
    \end{gathered}
\end{equation}

{\rm
\textbullet \,\,
$\text{\ref{eq:P5case32}} \to \text{\ref{eq:P3D6case6}}$.}
The substitution of the formulas \eqref{eq:P5toP3D6'map}, \eqref{eq:P5toP3D6'map_k}, \eqref{eq:P5toP3D6'map_la} into \eqref{eq:Laxpair_P5case32} gives a pair \eqref{eq:matABform_P3D6'_} with
\begin{gather}
    \notag
    \begin{aligned}
    A_0
    &= 
    \begin{pmatrix}
    \kappa_4 z & 0 \\[0.9mm]
    0 & 0
    \end{pmatrix},
    &
    A_{-1}
    &= -
    \begin{pmatrix}
    \kappa_1
    & 
    u
    \\[0.9mm]
    u v^2
    + \kappa_1 v
    + \kappa_2 u v 
    + \kappa_3
    & 
    v u - u v
    \end{pmatrix},
    &
    A_{-2}
    &= 
    \begin{pmatrix}
    v + \kappa_2
    & 
    - 1
    \\[0.9mm]
    v^2 + \kappa_2 v
    &
    - v
    \end{pmatrix},
    \end{aligned}
    \\[-1mm]
    \label{eq:Laxpair_P3D6case6}
    \\[-1mm]
    \notag
    \begin{aligned}
    B_1
    &= 
    \begin{pmatrix}
    \kappa_4 & 0 
    \\[0.9mm]
    0 & 0
    \end{pmatrix}
    ,
    &&&
    B_0
    &= z^{-1}
    \begin{pmatrix}
    u v 
    + \kappa_2 u 
    &
    - u
    \\[0.9mm]
    - \brackets{
    u v^2 
    + \kappa_1 v
    + \kappa_2 u v 
    + \kappa_3
    }
    & 
    - v u
    \end{pmatrix}
    \end{aligned}
\end{gather}
for the {\rm\ref{eq:P3D6case6}} system.

{\rm
\textbullet \,\,
$\text{\ref{eq:P4casem2m1}}, \,\, \text{\ref{eq:P3D6case6}} \to \text{\ref{eq:P2betam1}}$.}
The degeneration data \eqref{eq:P4toP2map}, \eqref{eq:P4toP2map_la} with $\const = 0$ and \eqref{eq:P3D6'toP2map}, \eqref{eq:P3D6'toP2map_la} with $\const = 0$ for the {\rm\ref{eq:P4casem2m1}} and {\rm\ref{eq:P3D6case6}} systems, respectively, lead to the {\rm\ref{eq:P2betam1}} system with a pair \eqref{eq:matABform_P2}, where
\begin{equation} \label{eq:Laxpair_P20}
    \begin{gathered}
        \begin{aligned}
        A_2
        &= 
        \begin{pmatrix}
        2 & 0 \\[0.9mm] 0 & 0
        \end{pmatrix},
        &&&
        A_1
        &= 
        \begin{pmatrix}
        0 & -2 \\[0.9mm]
        - v & 0
        \end{pmatrix},
        &&&
        A_0
        &= 
        \begin{pmatrix}
        - v + z & - 2 u
        \\[0.9mm]
        u v + \kappa_3 & v
        \end{pmatrix},
        \end{aligned}
        \\[2mm]
        \begin{aligned}
        B_1
        &= 
        \begin{pmatrix}
        1 & 0 \\[0.9mm] 0 & 0
        \end{pmatrix},
        &&&
        B_0
        &= 
        \begin{pmatrix}
        - u & -1 \\[0.9mm]
        - \tfrac12 v & 0
        \end{pmatrix}.
        \end{aligned}
    \end{gathered}
\end{equation}

{\rm
\textbullet \,\,
$\text{\ref{eq:P2betam1}} \to \text{\ref{eq:P1ham}}$.}
The limiting transition defined by \eqref{eq:P2toP1map}, \eqref{eq:P2toP1map_la} with $\const = 0$ reduces Lax pair \eqref{eq:Laxpair_P20} to a pair \eqref{eq:matABform_P1} with
\begin{equation} \label{eq:Laxpair_P10}
    \begin{gathered}
        \begin{aligned}
        A_2
        &= 
        \begin{pmatrix}
        0 & 0 \\[0.9mm] 2 & 0
        \end{pmatrix},
        &&&
        A_1
        &= 
        \begin{pmatrix}
        0 & -2 \\[0.9mm]
        - 2 u & 0
        \end{pmatrix},
        &&&
        A_0
        &= 
        \begin{pmatrix}
        - v & - 2 u
        \\[0.9mm]
        2 u^2 + z & v
        \end{pmatrix},
        \end{aligned}
        \\[2mm]
        \begin{aligned}
        B_1
        &= 
        \begin{pmatrix}
        0 & 0 \\[0.9mm] 1 & 0
        \end{pmatrix},
        &&&
        B_0
        &= 
        \begin{pmatrix}
        0 & -1 \\[0.9mm]
        - 2 u & 0
        \end{pmatrix}
        \end{aligned}
    \end{gathered}
\end{equation}
for the {\rm\ref{eq:P1ham}} system.
\end{example}

\begin{example}
To illustrate Remark \ref{rem:degcorr}, we present the degenerations of Lax pairs for the branch starting from system {\rm\ref{eq:P6case34}} {\rm(}see Example \ref{exm:taudeg}{\rm)}.
Degenerations are shown in Figure \ref{pic:deg_10P6}{\rm:}
\begin{figure}[H]
    \centering
    \scalebox{1.}{\tikzset{every picture/.style={line width=0.75pt}} 

\begin{tikzpicture}[x=0.75pt,y=0.75pt,yscale=-1,xscale=1]

\draw    (32.65,36.5) -- (59.88,36.5) ;
\draw [shift={(61.88,36.5)}, rotate = 180] [color={rgb, 255:red, 0; green, 0; blue, 0 }  ][line width=0.75]    (10.93,-3.29) .. controls (6.95,-1.4) and (3.31,-0.3) .. (0,0) .. controls (3.31,0.3) and (6.95,1.4) .. (10.93,3.29)   ;
\draw    (99.3,24.46) -- (128.8,13.68) ;
\draw [shift={(130.67,12.99)}, rotate = 159.93] [color={rgb, 255:red, 0; green, 0; blue, 0 }  ][line width=0.75]    (10.93,-3.29) .. controls (6.95,-1.4) and (3.31,-0.3) .. (0,0) .. controls (3.31,0.3) and (6.95,1.4) .. (10.93,3.29)   ;
\draw    (164.3,64.46) -- (193.8,51.68) ;
\draw [shift={(195.67,50.5)}, rotate = 152.93] [color={rgb, 255:red, 0; green, 0; blue, 0 }  ][line width=0.75]    (10.93,-3.29) .. controls (6.95,-1.4) and (3.31,-0.3) .. (0,0) .. controls (3.31,0.3) and (6.95,1.4) .. (10.93,3.29)   ;
\draw    (99.3,49.46) -- (128.72,63.72) ;
\draw [shift={(130.52,64.59)}, rotate = 205.86] [color={rgb, 255:red, 0; green, 0; blue, 0 }  ][line width=0.75]    (10.93,-3.29) .. controls (6.95,-1.4) and (3.31,-0.3) .. (0,0) .. controls (3.31,0.3) and (6.95,1.4) .. (10.93,3.29)   ;
\draw    (165.3,11.46) -- (194.72,25.72) ;
\draw [shift={(196.52,26.59)}, rotate = 205.86] [color={rgb, 255:red, 0; green, 0; blue, 0 }  ][line width=0.75]    (10.93,-3.29) .. controls (6.95,-1.4) and (3.31,-0.3) .. (0,0) .. controls (3.31,0.3) and (6.95,1.4) .. (10.93,3.29)   ;
\draw    (229.65,36.5) -- (256.88,36.5) ;
\draw [shift={(258.88,36.5)}, rotate = 180] [color={rgb, 255:red, 0; green, 0; blue, 0 }  ][line width=0.75]    (10.93,-3.29) .. controls (6.95,-1.4) and (3.31,-0.3) .. (0,0) .. controls (3.31,0.3) and (6.95,1.4) .. (10.93,3.29)   ;

\draw (4.02,29.79) node [anchor=north west][inner sep=0.75pt]    
{\ref{eq:P6case34}};
\draw (68.02,26.8) node [anchor=north west][inner sep=0.75pt]    
{\ref{eq:hamP5case1}};
\draw (137.8,1.04) node [anchor=north west][inner sep=0.75pt]   
{\ref{eq:hamP3D6case1}};
\draw (137.8,55.4) node [anchor=north west][inner sep=0.75pt]    
{\ref{eq:P4casem1m1}};
\draw (202.77,26.8) node [anchor=north west][inner sep=0.75pt]    
{\ref{eq:P2betam1}};
\draw (267.02,26.8) node [anchor=north west][inner sep=0.75pt]    
{\ref{eq:P1ham}};

\end{tikzpicture}}
    \caption{Degenerations of \ref{eq:P6case34}}
    \label{pic:deg_10P6}
\end{figure}

\medskip

{\rm
\textbullet \,\,
$\text{\ref{eq:P6case34}} \to \text{\ref{eq:hamP5case1}}$.}
Substituting \eqref{eq:P6toP5map}  into Lax pair \eqref{eq:Laxpair_P6case34}, we obtain a pair of the form \eqref{eq:matABform_P5} for the {\rm\ref{eq:hamP5case1}} system. The matrices read as
\begin{gather} 
    \notag
    \begin{aligned}
    A_{0}
    &= 
    \begin{pmatrix}
    \kappa_4 z & 0 \\[0.9mm] 0 & 0
    \end{pmatrix},
    &
    A_1
    &= 
    \begin{pmatrix}
    - v u + \kappa_1 & 1
    \\[0.9mm]
    - v u v u
    + \kappa_1 v u
    + \kappa_3
    & v u
    \end{pmatrix},
    &
    A_2
    &= 
    \begin{pmatrix}
    v u - \kappa_2
    & - u
    \\[0.9mm]
    v^2 u
    - \kappa_2 v
    & 
    - v u
    \end{pmatrix},
    \end{aligned}
    \\[-1mm]
    \label{eq:Laxpair_hamP5case1}
    \\[-1mm]
    \notag
    \begin{aligned}
    \hspace{-8mm}
    B_1
    &= 
    \begin{pmatrix}
    \kappa_4 & 0 \\[0.9mm] 0 & 0
    \end{pmatrix},
    &
    B_{0}
    &= z^{-1}
    \begin{pmatrix}
    u v
    - v 
    + \kappa_1 - \kappa_2
    & 
    - u + 1
    \\[0.9mm]
    - v u v u
    + v^2 u
    + \kappa_1 v u
    - \kappa_2 v
    + \kappa_3 
    &
    - u v u
    - v u^2
    + u v + 2 v u
    + \kappa_1 u - v
    - \kappa_2
    \end{pmatrix}.
    \end{aligned}
\end{gather}

{\rm
\textbullet \,\,
$\text{\ref{eq:hamP5case1}} \to \text{\ref{eq:P4casem1m1}}$.} 
Using the map \eqref{eq:P5toP4map}, \eqref{eq:P5toP4map_k}, \eqref{eq:P5toP4map_la}, we reduce Lax pair \eqref{eq:Laxpair_hamP5case1} to \eqref{eq:matABform_P4}, with
\begin{equation} \label{eq:Laxpair_P4casem1m1}
    \begin{gathered}
        \begin{aligned}
        A_1
        &= 
        \begin{pmatrix}
        - 2 & 0 \\[0.9mm] 0 & 0
        \end{pmatrix},
        &&&
        A_0
        &= 
        \begin{pmatrix}
        - 2 z & 1 \\[0.9mm] 
        v u + \kappa_3 & 0
        \end{pmatrix},
        &&&
        A_{-1}    
        &= \tfrac12
        \begin{pmatrix}
        v u + \kappa_2 & - u \\[0.9mm]
        v^2 u + \kappa_2 v & - v u
        \end{pmatrix},
        \end{aligned}
        \\[2mm]
        \begin{aligned}
        B_1
        &= 
        \begin{pmatrix}
        - 2 & 0 \\[0.9mm] 0 & 0
        \end{pmatrix},
        &&&
        B_0
        &= 
        \begin{pmatrix}
        - v - 2 z & 1 \\[0.9mm]
        v u + \kappa_3 & u - v
        \end{pmatrix}.
        \end{aligned}
    \end{gathered}
\end{equation}

{\rm
\textbullet \,\,
$\text{\ref{eq:hamP5case1}} \to \text{\ref{eq:hamP3D6case1}}$.}
The limiting transition defined by the formulas \eqref{eq:P5toP3D6'map}, \eqref{eq:P5toP3D6'map_k}, \eqref{eq:P5toP3D6'map_la} transforms the pair \eqref{eq:Laxpair_hamP5case1} to \eqref{eq:matABform_P3D6'_} with
\begin{gather}
    \notag
    \begin{aligned}
    A_0
    &= 
    \begin{pmatrix}
    \kappa_4 z & 0 \\[0.9mm]
    0 & 0
    \end{pmatrix},
    &
    A_{-1}
    &= -
    \begin{pmatrix}
    \kappa_1
    & 
    u
    \\[0.9mm]
    v u v
    + \kappa_1 v
    + \kappa_2 v u
    + \kappa_3
    & 
    0
    \end{pmatrix},
    &
    A_{-2}
    &= 
    \begin{pmatrix}
    v + \kappa_2
    & 
    - 1
    \\[0.9mm]
    v^2 + \kappa_2 v
    &
    - v
    \end{pmatrix},
    \end{aligned}
    \\[-1mm]
    \label{eq:Laxpair_hamP3D6case1}
    \\[-1mm]
    \notag
    \begin{aligned}
    B_1
    &= 
    \begin{pmatrix}
    \kappa_4 & 0 
    \\[0.9mm]
    0 & 0
    \end{pmatrix}
    ,
    &&&
    B_0
    &= z^{-1}
    \begin{pmatrix}
    u v 
    - \kappa_1
    &
    - u
    \\[0.9mm]
    - \brackets{
    v u v 
    + \kappa_1 v
    + \kappa_2 v u
    + \kappa_3
    }
    & 
    - v u 
    - \kappa_2 u - \kappa_1
    \end{pmatrix}.
    \end{aligned}
\end{gather}

{\rm
\textbullet \,\,
$\text{\ref{eq:P4casem1m1}}, \,\, \text{\ref{eq:hamP3D6case1}} \to \text{\ref{eq:P2betam1}}$.}
The substitution of \eqref{eq:P4toP2map}, \eqref{eq:P4toP2map_la} with $\const = \frac14$ into \eqref{eq:Laxpair_P4casem1m1} and \eqref{eq:P3D6'toP2map},  \eqref{eq:P3D6'toP2map_la} with $\const = - \frac14$ into \eqref{eq:Laxpair_hamP3D6case1} leads to the pair of the form \eqref{eq:matABform_P2} for the {\rm\ref{eq:P2betam1}} system, where
\begin{equation} \label{eq:Laxpair_P20_}
    \begin{gathered}
        \begin{aligned}
        A_2
        &= 
        \begin{pmatrix}
        2 & 0 \\[0.9mm] 0 & 0
        \end{pmatrix},
        &&&
        A_1
        &= 
        \begin{pmatrix}
        0 & -2 \\[0.9mm]
        - v & 0
        \end{pmatrix},
        &&&
        A_0
        &= 
        \begin{pmatrix}
        - v + z & - 2 u
        \\[0.9mm]
        v u + \kappa_3 & v
        \end{pmatrix},
        \end{aligned}
        \\[2mm]
        \begin{aligned}
        B_1
        &= 
        \begin{pmatrix}
        1 & 0 \\[0.9mm] 0 & 0
        \end{pmatrix},
        &&&
        B_0
        &= 
        \begin{pmatrix}
        0 & -1 \\[0.9mm]
        - \tfrac12 v & u
        \end{pmatrix}.
        \end{aligned}
    \end{gathered}
\end{equation}

{\rm
\textbullet \,\,
$\text{\ref{eq:P2betam1}} \to \text{\ref{eq:P1ham}}$.}
The mapping \eqref{eq:P2toP1map}, \eqref{eq:P2toP1map_la} with $\const = 0$ reduces Lax pair \eqref{eq:Laxpair_P20_} to pair \eqref{eq:Laxpair_P10} for the {\rm\ref{eq:P1ham}} system.
\end{example}

\section{Conclusion}
We have constructed a special class of \PPainleve \, type systems and presented isomonodromic Lax pairs for them. Our systems from Appendix \ref{sec:sysintlist}, taking together with Hamiltonian systems from Appendix \ref{sec:syshamlist} found by H. Kawakami \cite{Kawakami_2015}, form a collection of systems that are closed with respect to degenerations.

  The following unsolved problems appear in connection with the obtained results:
\begin{itemize}
    \item Extend this class of systems to a complete list of non-abelian analogs of \Painleve equations, including all known examples; 
    \item Find the corresponding \Painleve - Calogero systems \cite{Bertola2018}; 
    \item Interpret obtainеd systems as symmetry reductions of $1+1$ or $2+1$ dimensional integrable systems;
    \item Prove that their solutions are meromorphic; 
    \item Find B\"acklund transformations for the resulting systems.
\end{itemize}

\subsubsection*{Acknowledgements}

The authors are grateful to V.~E.~Adler, M.~A.~Bershtein and B.~I.~Suleimanov for useful discussions.  The research of the second author was carried out under the State Assignment 0029-2021-0004 (Quantum field theory) of the Ministry of Science and Higher Education of the Russian Federation. The first author was partially supported by the International Laboratory of Cluster Geometry HSE, RF Government grant № 075-15-2021-608, and by Young Russian Mathematics award.

\appendix

\section*{Appendices}
\section{Lists of non-abelian systems of  \Painleve type that have Okamoto integral}
\label{sec:sysintlist}
  
\subsection{Systems of \texorpdfstring{$\PVI$}{P6} type}\label{sec:sysintlistP6}

\subsubsection{Case 1}\label{sec:P6case1}

\begin{gather}
    \label{eq:P6case11}
    \tag*{$\PPVIn{1}$}
    \left\{
    \begin{array}{lcr}
         z (z - 1) u'
         &=& 2 u v u^2
         - u v u - v u^2
         - \kappa_1 u^2 
         + \kappa_2 u
         \hspace{4.4cm}
         \\[1mm]
         && 
         + \, z \brackets{
         - u v u - v u^2
         + 2 v u
         + \kappa_4 u 
         + (\kappa_1 - \kappa_2 - \kappa_4)
         },
         \\[2mm]
         z (z - 1) v'
         &=& - v u^2 v - 2 v u v u
         + v u v + v^2 u
         + \tfrac12 (\kappa_1 + \kappa_3) u v 
         + \tfrac12 (3 \kappa_1 - \kappa_3) v u
         \\[1mm]
         && 
         - \, \kappa_2 v
         + \tfrac14 (\kappa_3^2 - \kappa_1^2)
         + z \brackets{
         v u v + v^2 u
         - v^2
         - \kappa_4 v
         },
    \end{array}
    \right.
    \\[3mm]
    \notag
    \begin{aligned}
    J
    = v u^2 v u
    - v u v u
    - \tfrac12 (\kappa_3 + \kappa_1) u v u
    + \tfrac12 (\kappa_3 - \kappa_1) v u^2
    + \kappa_2 v u
    - \tfrac14 (\kappa_3^2 - \kappa_1^2) u
    \\[1mm]
    + \, z \brackets{
    - v u v u
    + v^2 u
    + \kappa_4 v u
    + (\kappa_1 - \kappa_2 - \kappa_4) v
    }
    .
    \end{aligned}
\end{gather}

\begin{gather}
    \label{eq:P6case13}
    \tag*{$\PPVIn{2}$}
    \left\{
    \begin{array}{lcr}
         z (z - 1) u'
         &=& 2 u v u^2
         - u v u - v u^2
         - \kappa_1 u^2 
         + \kappa_2 u
         \hspace{4.4cm}
         \\[1mm]
         && 
         + \, z \brackets{
         - 2 u v u
         + u v + v u
         + \kappa_4 u 
         + (\kappa_1 - \kappa_2 - \kappa_4)
         },
         \\[2mm]
         z (z - 1) v'
         &=& - v u^2 v - 2 v u v u
         + v u v + v^2 u
         + \tfrac12 (\kappa_1 + \kappa_3) u v
         + \tfrac12 (3 \kappa_1 - \kappa_3) v u
         \\[1mm]
         && 
         - \, \kappa_2 v
         + \tfrac14 (\kappa_3^2 - \kappa_1^2)
         + z \brackets{
         2 v u v
         - v^2
         - \kappa_4 v
         },
    \end{array}
    \right.
    \\[3mm]
    \notag
    \begin{aligned}
    J
    = v u^2 v u
    - v u v u
    - \tfrac12 (\kappa_3 + \kappa_1) u v u
    + \tfrac12 (\kappa_3 - \kappa_1) v u^2
    + \kappa_2 v u
    - \tfrac14 (\kappa_3^2 - \kappa_1^2) u
    + z \left(
    - v u^2 v 
    \right.
    \\[1mm]
    \left.
    + \, v u v
    + \tfrac12 (\kappa_1 + \kappa_3) u v
    + \tfrac12 (2 \kappa_4 - \kappa_1 - \kappa_3) v u
    + (\kappa_1 - \kappa_2 - \kappa_4) v
    \right)
    .
    \end{aligned}
\end{gather}

\begin{gather}
    \label{eq:P6case14}
    \tag*{$\PPVIn{3}$}
    \left\{
    \begin{array}{lcr}
         z (z - 1) u'
         &=& 2 u v u^2
         - 2 u v u
         - \kappa_1 u^2 
         + \kappa_2 u
         \hspace{4.4cm}
         \\[1mm]
         && 
         + \, z \brackets{
         - u v u - v u^2
         + u v + v u
         + \kappa_4 u 
         + (\kappa_1 - \kappa_2 - \kappa_4)
         },
         \\[2mm]
         z (z - 1) v'
         &=& - v u^2 v - 2 v u v u
         + 2 v u v
         + \tfrac12 (\kappa_1 + \kappa_3) u v
         + \tfrac12 (3 \kappa_1 - \kappa_3) v u
         \\[1mm]
         && 
         - \, \kappa_2 v
         + \tfrac14 (\kappa_3^2 - \kappa_1^2)
         + z \brackets{
         v u v + v^2 u
         - v^2
         - \kappa_4 v
         },
    \end{array}
    \right.
    \\[3mm]
    \notag
    \begin{aligned}
    J
    = v u^2 v u
    - v u^2 v
    - \tfrac12 (\kappa_3 + \kappa_1) u v u
    + \tfrac12 (\kappa_3 - \kappa_1) v u^2
    + \tfrac12 (\kappa_1 + \kappa_3) u v
    + \tfrac12 (2 \kappa_2 - \kappa_1 - \kappa_3) v u
    \\[1mm]
    - \, \tfrac14 (\kappa_3^2 - \kappa_1^2) u
    + z \brackets{
    - v u v u
    + v u v
    + \kappa_4 v u
    + (\kappa_1 - \kappa_2 - \kappa_4) v
    }
    .
    \end{aligned}
\end{gather}

\subsubsection{Case 2}
\label{sec:P6case2}

\begin{gather}
    \label{eq:P6case21}
    \tag*{$\PPVIn{4}$}
    \left\{
    \begin{array}{lcr}
         z (z - 1) u'
         &=& u v u^2 + v u^3
         - 2 v u^2
         - \kappa_1 u^2 
         + \kappa_2 u 
         \hspace{4.9cm}
         \\[1mm]
         && 
         + \, z \brackets{
         - u v u - v u^2
         + 2 v u
         + \kappa_4 u 
         + (\kappa_1 - \kappa_2 - \kappa_4)
         },
         \\[2mm]
         z (z - 1) v'
         &=& - 2 v u v u - v^2 u^2
         + 2 v^2 u
         + 2 \kappa_1 v u
         - \kappa_2 v
         + \kappa_3
         \hspace{3.3cm}
         \\[1mm]
         && 
         + \, z \brackets{
         v u v + v^2 u
         - v^2
         - \kappa_4 v
         },
    \end{array}
    \right.
    \\[3mm]
    \notag
    \begin{aligned}
    J
    = v u v u^2
    - v^2 u^2
    - \kappa_1 v u^2
    + \kappa_2 v u
    - \kappa_3 u
    + z \brackets{
    - v u v u
    + v^2 u
    + \kappa_4 v u
    + (\kappa_1 - \kappa_2 - \kappa_4) v
    }
    .
    \end{aligned}
\end{gather}

\begin{gather}
    \label{eq:P6case22}
    \tag*{$\PPVIn{5}$}
    \left\{
    \begin{array}{lcr}
         z (z - 1) u'
         &=& u v u^2 + v u^3
         - u v u - v u^2
         - \kappa_1 u^2 
         + \kappa_2 u
         \hspace{4.cm}
         \\[1mm]
         && 
         + \, z \brackets{
         - 2 v u^2
         + 2 v u
         + \kappa_4 u 
         + (\kappa_1 - \kappa_2 - \kappa_4)
         },
         \\[2mm]
         z (z - 1) v'
         &=& - 2 v u v u - v^2 u^2
         + v u v + v^2 u
         + 2 \kappa_1 v u
         - \kappa_2 v
         + \kappa_3
         \hspace{2.4cm}
         \\[1mm]
         && 
         + \, z \brackets{
         2 v^2 u 
         - v^2
         - \kappa_4 v
         },
    \end{array}
    \right.
    \\[3mm]
    \notag
    \begin{aligned}
    J
    = v u v u^2
    - v u v u
    - \kappa_1 v u^2
    + \kappa_2 v u
    - \kappa_3 u
    + z \brackets{
    - v^2 u^2
    + v^2 u
    + \kappa_4 v u
    + (\kappa_1 - \kappa_2 - \kappa_4) v
    }
    .
    \end{aligned}
\end{gather}

\begin{gather}
    \label{eq:P6case23}
    \tag*{$\PPVIn{6}$}
    \left\{
    \begin{array}{lcr}
         z (z - 1) u'
         &=& u v u^2 + v u^3
         - u v u - v u^2
         - \kappa_1 u^2 
         + \kappa_2 u
         \hspace{4.cm}
         \\[1mm]
         && 
         + \, z \brackets{
         - u v u - v u^2
         + u v + v u
         + \kappa_4 u 
         + (\kappa_1 - \kappa_2 - \kappa_4)
         },
         \\[2mm]
         z (z - 1) v'
         &=& - 2 v u v u - v^2 u^2
         + v u v + v^2 u
         + 2 \kappa_1 v u 
         - \kappa_2 v
         + \kappa_3
         \hspace{2.4cm}
         \\[1mm]
         && 
         + \, z \brackets{
         v u v + v^2 u
         - v^2
         - \kappa_4 v
         },
    \end{array}
    \right.
    \\[3mm]
    \notag
    \begin{aligned}
    J
    = v u v u^2
    - v u v u
    - \kappa_1 v u^2
    + \kappa_2 v u
    - \kappa_3 u
    + z \brackets{
    - v u v u
    + v u v
    + \kappa_4 v u
    + (\kappa_1 - \kappa_2 - \kappa_4) v
    }
    .
    \end{aligned}
\end{gather}

\subsubsection{Case 3}
\label{sec:P6case3}

\begin{gather}
    \label{eq:P6case31}
    \tag*{$\PPVIn{7}$}
    \left\{
    \begin{array}{lcr}
         z (z - 1) u'
         &=& u^2 v u + u v u^2
         - u^2 v - u v u
         - \kappa_1 u^2 
         + \kappa_2 u 
         \hspace{3.4cm}
         \\[1mm]
         && 
         + \, z \brackets{
         - u v u - v u^2
         + u v + v u
         + \kappa_4 u 
         + (\kappa_1 - \kappa_2 - \kappa_4)
         },
         \\[2mm]
         z (z - 1) v'
         &=& - u v u v - v u^2 v - v u v u
         + u v^2 + v u v
         + \kappa_1 u v + \kappa_1 v u
         - \kappa_2 v
         + \kappa_3
         \\[1mm]
         && 
         + \, z \brackets{
         v u v + v^2 u
         - v^2
         - \kappa_4 v
         },
    \end{array}
    \right.
    \\[3mm]
    \notag
    \begin{aligned}
    J
    = u v u v u
    - u v u v
    - \kappa_1 u v u
    + \kappa_2 u v
    - \kappa_3 u
    + z \brackets{
    - v u v u
    + v u v
    + \kappa_4 v u
    + (\kappa_1 - \kappa_2 - \kappa_4) v
    }
    .
    \end{aligned}
\end{gather}

\begin{gather}
    \label{eq:P6case32}
    \tag*{$\PPVIn{8}$}
    \left\{
    \begin{array}{lcr}
         z (z - 1) u'
         &=& u^2 v u + u v u^2
         - 2 u v u
         - \kappa_1 u^2 
         + \kappa_2 u 
         \hspace{3.4cm}
         \\[1mm]
         && 
         + \, z \brackets{
         - u v u - v u^2
         + 2 v u
         + \kappa_4 u 
         + (\kappa_1 - \kappa_2 - \kappa_4)
         },
         \\[2mm]
         z (z - 1) v'
         &=& - u v u v - v u^2 v - v u v u
         + 2 v u v
         + \kappa_1 u v + \kappa_1 v u
         - \kappa_2 v
         + \kappa_3
         \\[1mm]
         && 
         + \, z \brackets{
         v u v + v^2 u
         - v^2
         - \kappa_4 v
         },
    \end{array}
    \right.
    \\[3mm]
    \notag
    \begin{aligned}
    J
    = u v u v u
    - u v^2 u
    - \kappa_1 u v u
    + (\kappa_2 - \kappa_1 + \kappa_4) u v
    + (\kappa_1 - \kappa_4) v u
    - \kappa_3 u
    \\[1mm]
    + \, z \brackets{
    - v u v u
    + v^2 u
    + \kappa_4 v u
    + (\kappa_1 - \kappa_2 - \kappa_4) v
    }
    .
    \end{aligned}
\end{gather}

\begin{gather}
    \label{eq:P6case33}
    \tag*{$\PPVIn{9}$}
    \left\{
    \begin{array}{lcr}
         z (z - 1) u'
         &=& u^2 v u + u v u^2
         - u^2 v - u v u
         - \kappa_1 u^2 
         + \kappa_2 u 
         \hspace{3.4cm}
         \\[1mm]
         && 
         + \, z \brackets{
         - 2 u v u
         + 2 u v
         + \kappa_4 u 
         + (\kappa_1 - \kappa_2 - \kappa_4)
         },
         \\[2mm]
         z (z - 1) v'
         &=& - u v u v - v u^2 v - v u v u
         + u v^2 + v u v
         + \kappa_1 u v + \kappa_1 v u
         - \kappa_2 v
         + \kappa_3
         \\[1mm]
         && 
         + \, z \brackets{
         2 v u v
         - v^2
         - \kappa_4 v
         },
    \end{array}
    \right.
    \\[3mm]
    \notag
    \begin{aligned}
    J
    = u v u v u
    - u v u v
    - \kappa_1 u v u
    + \kappa_2 u v
    - \kappa_3 u
    + z \left(
    - u v^2 u
    + u v^2
    + (\kappa_1 - \kappa_2) u v
    \right.
    \\[1mm]
    \left.
    + \, (\kappa_4 - \kappa_1 + \kappa_2) v u
    + (\kappa_1 - \kappa_2 - \kappa_4) v
    \right)
    .
    \end{aligned}
\end{gather}

\begin{gather}
    \label{eq:P6case34}
    \tag*{$\PPVIn{10}$}
    \left\{
    \begin{array}{lcr}
         z (z - 1) u'
         &=& u^2 v u + u v u^2
         - u v u - v u^2
         - \kappa_1 u^2 
         + \kappa_2 u 
         \hspace{3.4cm}
         \\[1mm]
         && 
         + \, z \brackets{
         - u^2 v - u v u
         + u v + v u
         + \kappa_4 u 
         + (\kappa_1 - \kappa_2 - \kappa_4)
         },
         \\[2mm]
         z (z - 1) v'
         &=& - u v u v - v u^2 v - v u v u
         + v u v + v^2 u
         + \kappa_1 u v + \kappa_1 v u
         - \kappa_2 v
         + \kappa_3
         \\[1mm]
         && 
         + \, z \brackets{
         u v^2 + v u v
         - v^2
         - \kappa_4 v
         },
    \end{array}
    \right.
    \\[3mm]
    \notag
    \begin{aligned}
    J
    = u v u v u
    - v u v u
    - \kappa_1 u v u
    + \kappa_2 v u 
    - \kappa_3 u
    + \, z \brackets{
    - u v u v
    + v u v
    + \kappa_4 u v
    + (\kappa_1 - \kappa_2 - \kappa_4) v
    }
    .
    \end{aligned}
\end{gather}

\begin{gather}
    \label{eq:P6case35}
    \tag*{$\PPVIn{11}$}
    \left\{
    \begin{array}{lcr}
         z (z - 1) u'
         &=& u^2 v u + u v u^2
         - 2 u v u
         - \kappa_1 u^2 
         + \kappa_2 u 
         \hspace{3.4cm}
         \\[1mm]
         && 
         + \, z \brackets{
         - u^2 v - u v u
         + 2 u v
         + \kappa_4 u 
         + (\kappa_1 - \kappa_2 - \kappa_4)
         },
         \\[2mm]
         z (z - 1) v'
         &=& - u v u v - v u^2 v - v u v u
         + 2 v u v
         + \kappa_1 u v + \kappa_1 v u
         - \kappa_2 v
         + \kappa_3
         \\[1mm]
         && 
         + \, z \brackets{
         u v^2 + v u v
         - v^2
         - \kappa_4 v
         },
    \end{array}
    \right.
    \\[3mm]
    \notag
    \begin{aligned}
    J
    = u v u v u
    - u v^2 u
    - \kappa_1 u v u
    + (\kappa_1 - \kappa_4) u v
    + (\kappa_2 - \kappa_1 + \kappa_4) v u
    - \kappa_3 u
    \\[1mm]
    + \, z \brackets{
    - u v u v
    + u v^2
    + \kappa_4 u v
    + (\kappa_1 - \kappa_2 - \kappa_4) v
    }
    .
    \end{aligned}
\end{gather}

\begin{gather}
    \label{eq:P6case36}
    \tag*{$\PPVIn{12}$}
    \left\{
    \begin{array}{lcr}
         z (z - 1) u'
         &=& u^2 v u + u v u^2
         - u v u - v u^2
         - \kappa_1 u^2 
         + \kappa_2 u 
         \hspace{3.4cm}
         \\[1mm]
         && 
         + \, z \brackets{
         - 2 u v u
         + 2 v u
         + \kappa_4 u 
         + (\kappa_1 - \kappa_2 - \kappa_4)
         },
         \\[2mm]
         z (z - 1) v'
         &=& - u v u v - v u^2 v - v u v u
         + v u v + v^2 u
         + \kappa_1 u v + \kappa_1 v u
         - \kappa_2 v
         + \kappa_3
         \\[1mm]
         && 
         + \, z \brackets{
         2 v u v
         - v^2
         - \kappa_4 v
         },
    \end{array}
    \right.
    \\[3mm]
    \notag
    \begin{aligned}
    J
    = u v u v u
    - v u v u
    - \kappa_1 u v u
    + \kappa_2 v u
    - \kappa_3 u
    + z \left(
    - u v^2 u
    + v^2 u
    + (\kappa_4 - \kappa_1 + \kappa_2) u v
    \right.
    \\[1mm]
    \left.
    + \, (\kappa_1 - \kappa_2) v u
    + (\kappa_1 - \kappa_2 - \kappa_4) v
    \right)
    .
    \end{aligned}
\end{gather}

\subsubsection{Case 4}
\label{sec:P6case4}

\begin{gather}
    \label{eq:P6case41}
    \tag*{$\PPVIn{13}$}
    \left\{
    \begin{array}{lcr}
         z (z - 1) u'
         &=& 2 u^2 v u
         - u^2 v - u v u
         - \kappa_1 u^2 
         + \kappa_2 u 
         \hspace{4.4cm}
         \\[1mm]
         && 
         + \, z \brackets{
         - u^2 v - u v u
         + 2 u v
         + \kappa_4 u 
         + (\kappa_1 - \kappa_2 - \kappa_4)
         },
         \\[2mm]
         z (z - 1) v'
         &=& - 2 u v u v - v u^2 v
         + u v^2 + v u v
         + \tfrac12 \brackets{3 \kappa_1 -  \kappa_3} u v
         + \tfrac12 \brackets{\kappa_1 + \kappa_3} v u
         \\[1mm]
         && 
         - \, \kappa_2 v
         + \tfrac14 (\kappa_3^2 - \kappa_1^2)
         + z \brackets{
         u v^2 + v u v
         - v^2
         - \kappa_4 v
         },
    \end{array}
    \right.
    \\[3mm]
    \notag
    \begin{aligned}
    J
    = u v u^2 v
    - u v u v
    + \tfrac12 (\kappa_3 - \kappa_1) u^2 v
    - \tfrac12 (\kappa_3 + \kappa_1) u v u
    + \kappa_2 u v
    - \tfrac14 (\kappa_3^2 - \kappa_1^2) u
    \\[1mm]
    + \, z \brackets{
    - u v u v
    + u v^2
    + \kappa_4 u v
    + (\kappa_1 - \kappa_2 - \kappa_4) v
    }
    .
    \end{aligned}
\end{gather}

\begin{gather}
    \label{eq:P6case43}
    \tag*{$\PPVIn{14}$}
    \left\{
    \begin{array}{lcr}
         z (z - 1) u'
         &=& 2 u^2 v u
         - 2 u v u
         - \kappa_1 u^2 
         + \kappa_2 u
         \hspace{4.4cm}
         \\[1mm]
         && 
         + \, z \brackets{
         - u^2 v - u v u
         + u v + v u
         + \kappa_4 u 
         + (\kappa_1 - \kappa_2 - \kappa_4)
         },
         \\[2mm]
         z (z - 1) v'
         &=& - 2 u v u v - v u^2 v
         + 2 v u v
         + \tfrac12 (3 \kappa_1 - \kappa_3) u v
         + \tfrac12 (\kappa_1 + \kappa_3) v u
         \\[1mm]
         && 
         - \, \kappa_2 v
         + \tfrac14 (\kappa_3^2 - \kappa_1^2)
         + z \brackets{
         u v^2 + v u v
         - v^2
         - \kappa_4 v
         },
    \end{array}
    \right.
    \\[3mm]
    \notag
    \begin{aligned}
    J
    = u v u^2 v
    - v u^2 v
    + \tfrac12 (\kappa_3 - \kappa_1) u^2 v
    - \tfrac12 (\kappa_3 + \kappa_1) u v u
    + \tfrac12 (2 \kappa_2 - \kappa_1 - \kappa_3) u v
    \hspace{1cm}
    \\[1mm]
    + \, \tfrac12 (\kappa_1 + \kappa_3) v u
    - \tfrac14 (\kappa_3^2 - \kappa_1^2) u
    + z \brackets{
    - u v u v
    + v u v
    + \kappa_4 u v
    + (\kappa_1 - \kappa_2 - \kappa_4) v
    }
    .
    \end{aligned}
\end{gather}

\begin{gather}
    \label{eq:P6case44}
    \tag*{$\PPVIn{15}$}
    \left\{
    \begin{array}{lcr}
         z (z - 1) u'
         &=& 2 u^2 v u
         - u^2 v - u v u
         - \kappa_1 u^2 
         + \kappa_2 u
         \hspace{4.4cm}
         \\[1mm]
         && 
         + \, z \brackets{
         - 2 u v u
         + u v + v u
         + \kappa_4 u 
         + (\kappa_1 - \kappa_2 - \kappa_4)
         },
         \\[2mm]
         z (z - 1) v'
         &=& - 2 u v u v - v u^2 v
         + u v^2 + v u v
         + \tfrac12 (3 \kappa_1 - \kappa_3) u v
         + \tfrac12 (\kappa_1 + \kappa_3) v u
         \\[1mm]
         && 
         - \, \kappa_2 v
         + \tfrac14 (\kappa_3^2 - \kappa_1^2)
         + z \brackets{
         2 v u v
         - v^2
         - \kappa_4 v
         },
    \end{array}
    \right.
    \\[3mm]
    \notag
    \begin{aligned}
    J
    = u v u^2 v
    - u v u v
    + \tfrac12 (\kappa_3 - \kappa_1) u^2 v
    - \tfrac12 (\kappa_3 + \kappa_1) u v u
    + \kappa_2 u v
    - \tfrac14 (\kappa_3^2 - \kappa_1^2) u
    \hspace{1cm}
    \\[1mm]
    + \, z \brackets{
    - v u^2 v
    + v u v
    + \tfrac12 (2 \kappa_4 - \kappa_1 - \kappa_3) u v
    + \tfrac12 (\kappa_1 + \kappa_3) v u
    + (\kappa_1 - \kappa_2 - \kappa_4) v
    }
    .
    \end{aligned}
\end{gather}

\subsubsection{Case 5}
\label{sec:P6case5}

\begin{gather}
    \label{eq:P6case51}
    \tag*{$\PPVIn{16}$}
    \left\{
    \begin{array}{lcr}
         z (z - 1) u'
         &=& u^3 v + u^2 v u
         - 2 u^2 v
         - \kappa_1 u^2 
         + \kappa_2 u 
         \hspace{4.4cm}
         \\[1mm]
         && 
         + \, z \brackets{
         - u^2 v - u v u
         + 2 u v
         + \kappa_4 u 
         + (\kappa_1 - \kappa_2 - \kappa_4)
         },
         \\[2mm]
         z (z - 1) v'
         &=& - u^2 v^2 - 2 u v u v
         + 2 u v^2
         + 2 \kappa_1 u v
         - \kappa_2 v
         + \kappa_3
         \hspace{2.8cm}
         \\[1mm]
         && 
         + \, z \brackets{
         u v^2 + v u v
         - v^2
         - \kappa_4 v
         },
    \end{array}
    \right.
    \\[3mm]
    \notag
    \begin{aligned}
    J
    = u^2 v u v
    - u^2 v^2
    - \kappa_1 u^2 v
    + \kappa_2 u v
    - \kappa_3 u
    + z \brackets{
    - u v u v
    + u v^2
    + \kappa_4 u v
    + (\kappa_1 - \kappa_2 - \kappa_4) v
    }
    .
    \end{aligned}
\end{gather}

\begin{gather}
    \label{eq:P6case52}
    \tag*{$\PPVIn{17}$}
    \left\{
    \begin{array}{lcr}
         z (z - 1) u'
         &=& u^3 v + u^2 v u
         - u^2 v - u v u
         - \kappa_1 u^2 
         + \kappa_2 u 
         \hspace{4.cm}
         \\[1mm]
         && 
         + \, z \brackets{
         - 2 u^2 v
         + 2 u v
         + \kappa_4 u 
         + (\kappa_1 - \kappa_2 - \kappa_4)
         },
         \\[2mm]
         z (z - 1) v'
         &=& - u^2 v^2 - 2 u v u v
         + u v^2 + v u v
         + 2 \kappa_1 u v
         - \kappa_2 v
         + \kappa_3
         \hspace{2.4cm}
         \\[1mm]
         && 
         + \, z \brackets{
         2 u v^2
         - v^2
         - \kappa_4 v
         },
    \end{array}
    \right.
    \\[3mm]
    \notag
    \begin{aligned}
    J
    = u^2 v u v
    - u v u v
    - \kappa_1 u^2 v
    + \kappa_2 u v
    - \kappa_3 u
    + z \brackets{
    - u^2 v^2
    + u v^2
    + \kappa_4 u v
    + (\kappa_1 - \kappa_2 - \kappa_4) v
    }
    .
    \end{aligned}
\end{gather}

\begin{gather}
    \label{eq:P6case53}
    \tag*{$\PPVIn{18}$}
    \left\{
    \begin{array}{lcr}
         z (z - 1) u'
         &=& u^3 v + u^2 v u
         - u^2 v - u v u
         - \kappa_1 u^2 
         + \kappa_2 u 
         \hspace{4.cm}
         \\[1mm]
         && 
         + \, z \brackets{
         - u^2 v - u v u
         + u v + v u
         + \kappa_4 u 
         + (\kappa_1 - \kappa_2 - \kappa_4)
         },
         \\[2mm]
         z (z - 1) v'
         &=& - u^2 v^2 - 2 u v u v
         + u v^2 + v u v
         + 2 \kappa_1 u v
         - \kappa_2 v
         + \kappa_3
         \hspace{2.4cm}
         \\[1mm]
         && 
         + \, z \brackets{
         u v^2 + v u v
         - v^2
         - \kappa_4 v
         },
    \end{array}
    \right.
    \\[3mm]
    \notag
    \begin{aligned}
    J
    = u^2 v u v
    - u v u v
    - \kappa_1 u^2 v
    + \kappa_2 u v
    - \kappa_3 u
    + z \brackets{
    - u v u v
    + v u v
    + \kappa_4 u v
    + (\kappa_1 - \kappa_2 - \kappa_4) v
    }
    .
    \end{aligned}
\end{gather}

\subsection{Systems of \texorpdfstring{$\PV$}{P5} type}
\label{sec:sysintlistP5}
\subsubsection{Case 1}
\label{sec:P5case1}

\begin{gather}
    \label{eq:P5case11}
    \tag*{$\PPVn{1}$}
    \left\{
    \begin{array}{lcr}
         z \, u'
         &=& 2 u v u^2
         - 2 u v u - 2 v u^2
         - \kappa_1 u^2 + 2 v u
         + (\kappa_1 + \kappa_2) u
         - \kappa_2 
         + \kappa_4 z u,
         \hspace{3mm}
         \\[2mm]
         z \, v'
         &=& - v u^2 v - 2 v u v u
         + 2 v u v + 2 v^2 u
         + \tfrac12 (\kappa_1 + \kappa_3) u v 
         + \tfrac12 (3 \kappa_1 - \kappa_3) v u
         \\[1mm]
         && 
         - \, v^2
         - (\kappa_1 + \kappa_2) v
         + \tfrac14 (\kappa_3^2 - \kappa_1^2)
         - \kappa_4 z v,
    \end{array}
    \right.
    \\[3mm]
    \notag
    \begin{aligned}
    J
    = v u^2 v u 
    - 2 v u v u 
    - \tfrac12 (\kappa_3 + \kappa_1) u v u 
    + \tfrac12 (\kappa_3 - \kappa_1) v u^2
    + v^2 u
    + (\kappa_1 + \kappa_2) v u
    \\
    - \tfrac14 (\kappa_3^2 - \kappa_1^2) u 
    - \kappa_2 v
    + \kappa_4 z v u
    .
    \end{aligned}
\end{gather}

\begin{gather}
    \label{eq:P5case13}
    \tag*{$\PPVn{2}$}
    \left\{
    \begin{array}{lcr}
         z \, u'
         &=& 2 u v u^2
         - 3 u v u - v u^2
         - \kappa_1 u^2 + u v + v u
         + (\kappa_1 + \kappa_2) u
         - \kappa_2
         + \kappa_4 z u,
         \\[2mm]
         z \, v'
         &=& - v u^2 v - 2 v u v u
         + 3 v u v + v^2 u
         + \tfrac12 (\kappa_1 + \kappa_3) u v
         + \tfrac12 (3 \kappa_1 - \kappa_3) v u
         \hspace{3mm}
         \\[1mm]
         &&  
         - \, v^2
         - (\kappa_1 + \kappa_2) v
         + \tfrac14 (\kappa_3^2 - \kappa_1^2)
         - \kappa_4 z v,
    \end{array}
    \right.
    \\[3mm]
    \notag
    \begin{aligned}
    J
    = v u^2 v u
    - v u^2 v - v u v u
    - \tfrac12 (\kappa_3 + \kappa_1) u v u
    + \tfrac12 (\kappa_3 - \kappa_1) v u^2
    + v u v 
    + \tfrac12 (\kappa_3 + \kappa_1) u v 
    \\
    + \tfrac12 (2 \kappa_2 - \kappa_3 + \kappa_1) v u
    - \tfrac14 (\kappa_3^2 - \kappa_1^2) u - \kappa_2 v 
    + \kappa_4 z v u
    .
    \end{aligned}
\end{gather}

\subsubsection{Case 2}
\label{sec:P5case2}

\begin{gather}
    \label{eq:P5case21}
    \tag*{$\PPVn{3}$}
    \hspace{-5mm}
    \left\{
    \begin{array}{lcl}
         \hspace{-1mm}
         z \, u'
         \hspace{-1mm}
         &=& 
         \hspace{-2mm}
         u v u^2 + v u^3
         - 2 u v u - 2 v u^2
         - \kappa_1 u^2 + u v + v u 
         + (\kappa_1 + \kappa_2) u
         - \kappa_2
         + \kappa_4 z u,
         \\[2mm]
         \hspace{-1mm}
         z \, v'
         \hspace{-1mm}
         &=& 
         \hspace{-2mm}
         - 2 v u v u - v^2 u^2
         + 2 v u v + 2 v^2 u
         + 2 \kappa_1 v u - v^2
         - (\kappa_1 + \kappa_2) v
         + \kappa_3
         - \kappa_4 z v,
    \end{array}
    \right.
    \hspace{-5mm}
    \\[3mm]
    \notag
    \begin{aligned}
    J
    = v u v u^2
    - 2 v u v u
    - \kappa_1 v u^2
    + v u v
    + (\kappa_1 + \kappa_2) v u
    - \kappa_3 u 
    - \kappa_2 v
    + \kappa_4 z v u
    .
    \end{aligned}
\end{gather}

\begin{gather}
    \label{eq:P5case22}
    \tag*{$\PPVn{4}$}
    \hspace{-5mm}
    \left\{
    \begin{array}{lcl}
         \hspace{-1mm}
         z \, u'
         \hspace{-1mm}
         &=& 
         \hspace{-2mm}
         u v u^2 + v u^3
         - u v u - 3 v u^2
         - \kappa_1 u^2 + 2 v u
         + (\kappa_1 + \kappa_2) u
         - \kappa_2
         + \kappa_4 z u,
         \hspace{5mm}
         \\[2mm]
         \hspace{-1mm}
         z \, v'
         \hspace{-1mm}
         &=& 
         \hspace{-2mm}
         - 2 v u v u - v^2 u^2
         + v u v + 3 v^2 u
         + 2 \kappa_1 v u - v^2 
         - (\kappa_1 + \kappa_2) v
         + \kappa_3
         - \kappa_4 z v,
    \end{array}
    \right.
    \\[3mm]
    \notag
    \begin{aligned}
    J
    = v u v u^2
    - v u v u - v^2 u^2
    - \kappa_1 v u^2 + v^2 u
    + (\kappa_1 + \kappa_2) v u
    - \kappa_3 u - \kappa_2 v
    + \kappa_4 z v u
    .
    \end{aligned}
\end{gather}

\subsubsection{Case 3}
\label{sec:P5case3}

\begin{gather}
    \label{eq:P5case31}
    \tag*{$\PPVn{5}$}
    \hspace{-5mm}
    \left\{
    \begin{array}{lcr}
         z \, u'
         &=& u^2 v u + u v u^2
         - u^2 v - 3 u v u
         - \kappa_1 u^2 + 2 u v
         + (\kappa_1 + \kappa_2) u
         - \kappa_2 
         + \kappa_4 z u,
         \\[2mm]
         z \, v'
         &=& - u v u v - v u^2 v - v u v u
         + u v^2 + 3 v u v
         + \kappa_1 u v + \kappa_1 v u 
         - v^2
         \hspace{1.8cm}
         \\[1mm]
         && 
         - \, (\kappa_1 + \kappa_2) v
         + \kappa_3
         - \kappa_4 z v,
    \end{array}
    \right.
    \\[3mm]
    \notag
    \begin{aligned}
    J
    = u v u v u 
    - u v u v - u v^2 u 
    - \kappa_1 u v u 
    + u v^2
    + \kappa_1 u v
    + \kappa_2 v u
    - \kappa_3 u
    - \kappa_2 v
    + \kappa_4 z u v
    .
    \end{aligned}
\end{gather}

\begin{gather}
    \label{eq:P5case32}
    \tag*{$\PPVn{6}$}
    \hspace{-5mm}
    \left\{
    \begin{array}{lcr}
         z \, u'
         &=& u^2 v u + u v u^2
         - 3 u v u - v u^2
         - \kappa_1 u^2 + 2 v u
         + (\kappa_1 + \kappa_2) u
         - \kappa_2
         + \kappa_4 z u,
         \\[2mm]
         z \, v'
         &=& - u v u v - v u^2 v - v u v u
         + 3 v u v + v^2 u 
         + \kappa_1 u v + \kappa_1 v u - v^2
         \hspace{1.8cm}
         \\[1mm]
         && 
         - \, (\kappa_1 + \kappa_2) v
         + \kappa_3
         - \kappa_4 z v,
    \end{array}
    \right.
    \\[3mm]
    \notag
    \begin{aligned}
    J
    = u v u v u 
    - v u v u - u v^2 u 
    - \kappa_1 u v u 
    + v^2 u
    + \kappa_2 u v
    + \kappa_1 v u
    - \kappa_3 u
    - \kappa_2 v
    + \kappa_4 z v u
    .
    \end{aligned}
\end{gather}

\subsubsection{Case 4}
\label{sec:P5case4}

\begin{gather}
    \label{eq:P5case43}
    \tag*{$\PPVn{7}$}
    \left\{
    \begin{array}{lcr}
         z \, u'
         &=& 2 u^2 v u 
         - 2 u^2 v - 2 u v u
         - \kappa_1 u^2 + 2 u v 
         + (\kappa_1 + \kappa_2) u
         - \kappa_2
         + \kappa_4 z u,
         \hspace{3mm}
         \\[2mm]
         z \, v'
         &=& - 2 u v u v - v u^2 v
         + 2 u v^2 + 2 v u v
         + \tfrac12 (3 \kappa_1 - \kappa_3) u v
         + \tfrac12 (\kappa_1 + \kappa_3) v u
         \\[1mm]
         && 
         - \, v^2
         - (\kappa_1 + \kappa_2) v
         + \tfrac14 (\kappa_3^2 - \kappa_1^2)
         - \kappa_4 z v,
    \end{array}
    \right.
    \\[3mm]
    \notag
    \begin{aligned}
    J
    = u v u^2 v
    - 2 u v u v
    + \tfrac12 (\kappa_3 - \kappa_1) u^2 v
    - \tfrac12 (\kappa_3 + \kappa_1) u v u
    + u v^2
    + (\kappa_1 + \kappa_2) u v
    \\
    - \tfrac14 (\kappa_3^2 - \kappa_1^2) u - \kappa_2 v
    + \kappa_4 z u v
    .
    \end{aligned}
\end{gather}

\begin{gather}
    \label{eq:P5case44}
    \tag*{$\PPVn{8}$}
    \left\{
    \begin{array}{lcr}
         z \, u'
         &=& 2 u^2 v u
         - u^2 v - 3 u v u
         - \kappa_1 u^2 + u v + v u
         + (\kappa_1 + \kappa_2) u
         - \kappa_2
         + \kappa_4 z u,
         \\[2mm]
         z \, v'
         &=& - 2 u v u v - v u^2 v
         + u v^2 + 3 v u v
         + \tfrac12 (3 \kappa_1 - \kappa_3) u v
         + \tfrac12 (\kappa_1 + \kappa_3) v u
         \hspace{3mm}
         \\[1mm]
         &&  
         - \, v^2 
         - (\kappa_1 + \kappa_2) v
         + \tfrac14 (\kappa_3^2 - \kappa_1^2)
         - \kappa_4 z v,
    \end{array}
    \right.
    \\[3mm]
    \notag
    \begin{aligned}
    J
    = u v u^2 v
    - u v u v - v u^2 v
    + \tfrac12 (\kappa_3 - \kappa_1) u^2 v 
    - \tfrac12 (\kappa_3 + \kappa_1) u v u 
    + v u v
    + \tfrac12 (2 \kappa_2 + \kappa_1 - \kappa_3) u v
    \\
    + \tfrac12 (\kappa_1 + \kappa_3) v u
    - \tfrac14 (\kappa_3^2 - \kappa_1^2) u 
    - \kappa_2 v
    + \kappa_4 z u v
    .
    \end{aligned}
\end{gather}

\subsubsection{Case 5}
\label{sec:P5case5}

\begin{gather}
    \label{eq:P5case51}
    \tag*{$\PPVn{9}$}
    \hspace{-5mm}
    \left\{
    \begin{array}{lcl}
         \hspace{-1mm}
         z \, u'
         \hspace{-1mm}
         &=& 
         \hspace{-2mm}
         u^3 v + u^2 v u
         - 2 u^2 v - 2 u v u
         - \kappa_1 u^2 + u v + v u 
         + (\kappa_1 + \kappa_2) u
         - \kappa_2
         + \kappa_4 z u,
         \\[2mm]
         \hspace{-1mm}
         z \, v'
         \hspace{-1mm}
         &=& 
         \hspace{-2mm}
         - u^2 v^2 - 2 u v u v
         + 2 u v^2 + 2 v u v
         + 2 \kappa_1 u v - v^2
         - (\kappa_1 + \kappa_2) v
         + \kappa_3
         - \kappa_4 z v,
    \end{array}
    \right.
    \hspace{-5mm}
    \\[3mm]
    \notag
    \begin{aligned}
    J
    = u^2 v u v
    - 2 u v u v
    - \kappa_1 u^2 v + v u v
    + (\kappa_1 + \kappa_2) u v
    - \kappa_3 u - \kappa_2 v
    + \kappa_4 z u v
    .
    \end{aligned}
\end{gather}

\begin{gather}
    \label{eq:P5case52}
    \tag*{$\PPVn{10}$}
    \hspace{-4mm}
    \left\{
    \begin{array}{lcl}
         \hspace{-1mm}
         z \, u'
         \hspace{-1mm}
         &=& 
         \hspace{-2mm}
         u^3 v + u^2 v u
         - 3 u^2 v - u v u
         - \kappa_1 u^2 + 2 u v
         + (\kappa_1 + \kappa_2) u
         - \kappa_2
         + \kappa_4 z u,
         \hspace{5mm}
         \\[2mm]
         \hspace{-1mm}
         z \, v'
         \hspace{-1mm}
         &=& 
         \hspace{-2mm}
         - u^2 v^2 - 2 u v u v
         + 3 u v^2 + v u v
         + 2 \kappa_1 u v - v^2
         - (\kappa_1 + \kappa_2) v
         + \kappa_3
         - \kappa_4 z v,
    \end{array}
    \right.
    \\[3mm]
    \notag
    \begin{aligned}
    J
    = u^2 v u v
    - u^2 v^2 - u v u v
    - \kappa_1 u^2 v + u v^2
    + (\kappa_1 + \kappa_2) u v
    - \kappa_3 u - \kappa_2 v
    + \kappa_4 z u v
    .
    \end{aligned}
\end{gather}

\subsection{Systems of \texorpdfstring{$\PIV$}{P4} type}
\label{sec:sysintlistP4}

\begin{minipage}[l]{0.49\linewidth}
\begin{gather}
    \label{eq:P4casem2m1}
    \tag*{$\PPIVn{1}$}
    \left\{
    \begin{array}{lcl}
        \hspace{-1mm}
         u' 
         \hspace{-1mm}
         &=& 
         \hspace{-1mm}
         - u^2 
         + 2 v u
         - 2 z u
         + \kappa_2,
         \\[2mm]
         \hspace{-1mm}
         v'
         \hspace{-1mm}
         &=& 
         \hspace{-1mm}
         - v^2
         + v u + u v
         + 2 z v
         + \kappa_3,
    \end{array}
    \right.
    \\[3mm]
    \notag
    J
    = v^2 u - u v u
    - \kappa_3 u + \kappa_2 v 
    - 2 z v u
    .
\end{gather}
\end{minipage}
\hspace{1.5mm}
\begin{minipage}[l]{0.49\linewidth}
\begin{gather}
    \label{eq:P4casem20}
    \tag*{$\PPIVn{2}$}
    \hspace{-5mm}
    \left\{
    \begin{array}{lcl}
        \hspace{-1mm}
         u' 
         \hspace{-1mm}
         &=& 
         \hspace{-1mm}
         - u^2 
         + 2 v u
         - 2 z u
         + \kappa_2,
         \\[2mm]
         \hspace{-1mm}
         v'
         \hspace{-1mm}
         &=& 
         \hspace{-1mm}
         - v^2
         + 2 v u
         + 2 z v
         + \kappa_3,
    \end{array}
    \right.
    \hspace{-5mm}
    \\[3mm]
    \notag
    J
    = v^2 u - v u^2
    - \kappa_3 u + \kappa_2 v 
    - 2 z v u
    .
\end{gather}
\end{minipage}

\vspace{3mm}
\begin{minipage}[l]{0.46\linewidth}
\begin{gather}
    \label{eq:P4casem1m2}
    \tag*{$\PPIVn{3}$}
    \hspace{-5mm}
    \left\{
    \begin{array}{lcl}
        \hspace{-1mm}
         u' 
         \hspace{-1mm}
         &=& 
         \hspace{-1mm}
         - u^2 
         + u v + v u
         - 2 z u
         + \kappa_2,
         \\[2mm]
         \hspace{-1mm}
         v'
         \hspace{-1mm}
         &=& 
         \hspace{-1mm}
         - v^2
         + 2 u v
         + 2 z v
         + \kappa_3,
    \end{array}
    \right.
    \hspace{-5mm}
    \\[3mm]
    \notag
    J
    = v u v - u^2 v
    - \kappa_3 u + \kappa_2 v 
    - 2 z u v
    .
\end{gather}
\end{minipage}
\hspace{1.5mm}
\begin{minipage}[l]{0.48\linewidth}
\begin{gather}
    \label{eq:P4casem10}
    \tag*{$\PPIVn{4}$}
    \hspace{-5mm}
    \left\{
    \begin{array}{lcl}
        \hspace{-1mm}
         u' 
         \hspace{-1mm}
         &=& 
         \hspace{-1mm}
         - u^2 
         + u v + v u
         - 2 z u
         + \kappa_2,
         \\[2mm]
         \hspace{-1mm}
         v'
         \hspace{-1mm}
         &=& 
         \hspace{-1mm}
         - v^2
         + 2 v u
         + 2 z v
         + \kappa_3,
    \end{array}
    \right.
    \hspace{-5mm}
    \\[3mm]
    \notag
    J
    = v u v - v u^2
    - \kappa_3 u + \kappa_2 v 
    - 2 z v u
    .
\end{gather}
\end{minipage}

\vspace{3mm}
\begin{minipage}[l]{0.46\linewidth}
\begin{gather}
    \label{eq:P4case0m2}
    \tag*{$\PPIVn{5}$}
    \hspace{-5mm}
    \left\{
    \begin{array}{lcl}
        \hspace{-1mm}
         u' 
         \hspace{-1mm}
         &=& 
         \hspace{-1mm}
         - u^2 
         + 2 u v
         - 2 z u
         + \kappa_2,
         \\[2mm]
         \hspace{-1mm}
         v'
         \hspace{-1mm}
         &=& 
         \hspace{-1mm}
         - v^2
         + 2 u v
         + 2 z v
         + \kappa_3,
    \end{array}
    \right.
    \hspace{-5mm}
    \\[3mm]
    \notag
    J
    = u v^2 - u^2 v
    - \kappa_3 u + \kappa_2 v 
    - 2 z u v
    .
\end{gather}
\end{minipage}
\hspace{1.5mm}
\begin{minipage}[l]{0.48\linewidth}
\begin{gather}
    \label{eq:P4case0m1}
    \tag*{$\PPIVn{6}$}
    \hspace{-5mm}
    \left\{
    \begin{array}{lcl}
        \hspace{-1mm}
         u' 
         \hspace{-1mm}
         &=& 
         \hspace{-1mm}
         - u^2 
         + 2 u v
         - 2 z u
         + \kappa_2,
         \\[2mm]
         \hspace{-1mm}
         v'
         \hspace{-1mm}
         &=& 
         \hspace{-1mm}
         - v^2
         + v u + u v 
         + 2 z v
         + \kappa_3,
    \end{array}
    \right.
    \hspace{-5mm}
    \\[3mm]
    \notag
    J
    = u v^2 - u v u
    - \kappa_3 u + \kappa_2 v 
    - 2 z u v
    .
\end{gather}
\end{minipage}

\subsection{Systems of \texorpdfstring{$\PIIIpr$}{P3'} type}
\label{sec:sysintlistP3}
\subsubsection{Case 1}
\label{sec:P3D6case1}
\begin{gather}
    \label{eq:P3D6case1}
    \tag*{$\PPIIIprn{1}$}
    \left\{
    \begin{array}{lcl}
         z \, u'
         &=& 2 u v u + \kappa_1 u
         + \kappa_2 u^2 + \kappa_4 z,
         \\[2mm]
         z \, v'
         &=& - 2 v u v - \kappa_1 v
         - 2 \kappa_2 v u - \kappa_3,
    \end{array}
    \right.
    \\[3mm]
    \notag
    \begin{aligned}
    J
    = v u^2 v 
    + \kappa_1 v u 
    + \kappa_3 \, \kappa_2^{-1} [u, v]
    + \kappa_2 v u^2 + \kappa_3 u + \kappa_4 z v
    .
    \end{aligned}
\end{gather}

\begin{gather}
    \label{eq:P3D6case2}
    \tag*{$\PPIIIprn{2}$}
    \left\{
    \begin{array}{lcl}
         z \, u'
         &=& 2 u v u + \kappa_1 u
         + \kappa_2 u^2
         + \kappa_4 z,
         \\[2mm]
         z \, v'
         &=& - 2 v u v - \kappa_1 v
         - 2 \kappa_2 u v - \kappa_3,
    \end{array}
    \right.
    \\[3mm]
    \notag
    \begin{aligned}
    J
    = v u^2 v + \kappa_1 u v + \kappa_3 \, \kappa_2^{-1} [v, u]
    + \kappa_2 u^2 v + \kappa_3 u
    + \kappa_4 z v
    .
    \end{aligned}
\end{gather}

\subsubsection{Case 2}
\label{sec:P3D6case2}

\begin{gather}
    \label{eq:P3D6case5}
    \tag*{$\PPIIIprn{3}$}
    \left\{
    \begin{array}{lcl}
         z \, u'
         &=& 2 v u^2 + \kappa_1 u
         + \kappa_2 u^2 + \kappa_4 z,
         \\[2mm]
         z \, v'
         &=& - 2 v^2 u - \kappa_1 v
         - 2 \kappa_2 v u - \kappa_3,
    \end{array}
    \right.
    \\[3mm]
    \notag
    \begin{aligned}
    J
    = v^2 u^2 + \kappa_1 v u
    + \kappa_2 v u^2 + \kappa_3 u + \kappa_4 z v
    .
    \end{aligned}
\end{gather}

\subsubsection{Case 3}
\label{sec:P3D6case3}

\begin{gather}
    \label{eq:P3D6case4}
    \tag*{$\PPIIIprn{4}$}
    \left\{
    \begin{array}{lcl}
         z \, u'
         &=& u v u + v u^2 + \kappa_1 u
         + \kappa_2 u^2 + \kappa_4 z,
         \\[2mm]
         z \, v'
         &=& - v u v - v^2 u - \kappa_1 v
         - \kappa_2 u v - \kappa_2 v u - \kappa_3,
    \end{array}
    \right.
    \\[3mm]
    \notag
    \begin{aligned}
    J
    = v u v u + \kappa_1 v u 
    + \kappa_2 u v u + \kappa_3 u + \kappa_4 z v
    .
    \end{aligned}
\end{gather}

\begin{gather}
    \label{eq:P3D6case3}
    \tag*{$\PPIIIprn{5}$}
    \left\{
    \begin{array}{lcl}
         z \, u'
         &=& u v u + v u^2 + \kappa_1 u
         + \kappa_2 u^2 + \kappa_4 z,
         \\[2mm]
         z \, v'
         &=& - v u v - v^2 u - \kappa_1 v
         - 2 \kappa_2 v u - \kappa_3,
    \end{array}
    \right.
    \\[3mm]
    \notag
    \begin{aligned}
    J
    = v u v u + \kappa_1 v u
    + \kappa_2 v u^2 + \kappa_3 u + \kappa_4 z v
    .
    \end{aligned}
\end{gather}

\subsubsection{Case 4}
\label{sec:P3D6case4}

\begin{gather}
    \label{eq:P3D6case6}
    \tag*{$\PPIIIprn{6}$}
    \left\{
    \begin{array}{lcl}
         z \, u'
         &=& u^2 v + u v u + \kappa_1 u
         + \kappa_2 u^2 + \kappa_4 z,
         \\[2mm]
         z \, v'
         &=& - u v^2 - v u v - \kappa_1 v
         - \kappa_2 u v - \kappa_2 v u - \kappa_3,
    \end{array}
    \right.
    \\[3mm]
    \notag
    \begin{aligned}
    J
    = u v u v + \kappa_1 u v 
    + \kappa_2 u v u + \kappa_3 u + \kappa_4 z v
    .
    \end{aligned}
\end{gather}

\begin{gather}
    \label{eq:P3D6case7}
    \tag*{$\PPIIIprn{7}$}
    \left\{
    \begin{array}{lcl}
         z \, u'
         &=& u^2 v + u v u + \kappa_1 u
         + \kappa_2 u^2 + \kappa_4 z,
         \\[2mm]
         z \, v'
         &=& - u v^2 - v u v - \kappa_1 v
         - 2 \kappa_2 u v - \kappa_3,
    \end{array}
    \right.
    \\[3mm]
    \notag
    \begin{aligned}
    J
    = u v u v + \kappa_1 u v 
    + \kappa_2 u^2 v + \kappa_3 u + \kappa_4 z v
    .
    \end{aligned}
\end{gather}

\subsubsection{Case 5}
\label{sec:P3D6case5}

\begin{gather}
    \label{eq:P3D6case8}
    \tag*{$\PPIIIprn{8}$}
    \left\{
    \begin{array}{lcl}
         z \, u'
         &=& 2 u^2 v + \kappa_1 u
         + \kappa_2 u^2 + \kappa_4 z,
         \\[2mm]
         z \, v'
         &=& - 2 u v^2 - \kappa_1 v
         - 2 \kappa_2 u v - \kappa_3,
    \end{array}
    \right.
    \\[3mm]
    \notag
    \begin{aligned}
    J
    = u^2 v^2 + \kappa_1 u v 
    + \kappa_2 u^2 v + \kappa_3 u + \kappa_4 z v
    .
    \end{aligned}
\end{gather}

\subsection{Systems of \texorpdfstring{$\PII$}{P2} type}
\label{sec:sysintlistP2}
\begin{minipage}[l]{0.49\linewidth}
\begin{gather} 
    \label{eq:P2betam2}
    \tag*{$\PPIIn{1}$}
    \left\{
    \begin{array}{lcl}
         u'
         &=& - u^2 
         + v
         - \tfrac12 z,
         \\[2mm]
         v'
         &=& 2 u v
         + \kappa_3,
    \end{array}
    \right.
    \\[2mm]
    \notag
    J
    = - u^2 v
    + \tfrac12 v^2 
    - \kappa_3 u 
    - \tfrac12 z v.
\end{gather}
\end{minipage}
\hspace{1.5mm}
\begin{minipage}[l]{0.49\linewidth}
\begin{gather}
    \label{eq:P2beta0}
    \tag*{$\PPIIn{2}$}
    \left\{
    \begin{array}{lcl}
         u'
         &=& - u^2 
         + v
         - \tfrac12 z,
         \\[2mm]
         v'
         &=& 2 v u
         + \kappa_3,
    \end{array}
    \right.
    \\[2mm]
    \notag
    J
    = - v u^2 
    + \tfrac12 v^2 
    - \kappa_3 u 
    - \tfrac12 z v.
\end{gather}
\end{minipage}

\section{List of Hamiltonian non-abelian systems of \Painleve type}
\label{sec:syshamlist}

\begin{gather}
    \label{eq:hamP6case1_}
    \tag*{$\PVIn{H}$}
    \left\{
    \begin{array}{lcr}
         z (z - 1) u'
         &=& u^2 v u 
        + u v u^2 
        - 2 u v u
        - \kappa_1 u^2 
        + \kappa_2 u
        \hspace{3.4cm}
        \\[1mm]
        && 
        + \, z \brackets{
        - u^2 v 
        - v u^2
        + u v + v u
        + \kappa_4 u 
         + (\kappa_1 - \kappa_2 - \kappa_4)
         },
         \\[2mm]
         z (z - 1) v'
         &=& - u v u v
        - v u v u
        - v u^2 v
        + 2 v u v
        + \kappa_1 u v
        + \kappa_1 v u
        - \kappa_2 v
        + \kappa_3
        \\[1mm]
        && 
        + \, z \brackets{
        u v^2 
        + v^2 u
        - v^2
        - \kappa_4 v
         }.
    \end{array}
    \right.
\end{gather}

\begin{gather}
    \label{eq:hamP5case1}
    \tag*{$\PVn{H}$}
    \left\{
    \begin{array}{lcr}
         z \, u'
         &=& u^2 v u + u v u^2
         - u^2 v - 2 u v u - v u^2
         - \kappa_1 u^2 + u v + v u 
         + (\kappa_1 + \kappa_2) u
         \\[1mm]
         && 
         - \, \kappa_2
         + \kappa_4 z u,
         \\[2mm]
         z \, v'
         &=& - u v u v - v u^2 v - v u v u
         + u v^2 + 2 v u v + v^2 u 
         + \kappa_1 u v + \kappa_1 v u - v^2
         \hspace{0.4cm}
         \\[1mm]
         && 
         - \, (\kappa_1 + \kappa_2) v
         + \kappa_3
         - \kappa_4 z v.
    \end{array}
    \right.
\end{gather}

\begin{gather}
    \label{eq:P4casem1m1}
    \tag*{$\PIVn{H}$}
    \left\{
    \begin{array}{lcl}
        \hspace{-1mm}
         u' 
         \hspace{-1mm}
         &=& 
         \hspace{-1mm}
         - u^2 
         + u v + v u
         - 2 z u
         + \kappa_2,
         \\[2mm]
         \hspace{-1mm}
         v'
         \hspace{-1mm}
         &=& 
         \hspace{-1mm}
         - v^2
         + v u + u v
         + 2 z v
         + \kappa_3.
    \end{array}
    \right.
\end{gather}

\begin{gather}
    \label{eq:hamP3D6case1}
    \tag*{$\PIIIprn{H}$}
    \left\{
    \begin{array}{lcl}
         z \, u'
         &=& 2 u v u + \kappa_1 u
         + \kappa_2 u^2 + \kappa_4 z,
         \\[2mm]
         z \, v'
         &=& - 2 v u v - \kappa_1 v
         - \kappa_2 u v - \kappa_2 v u - \kappa_3.
    \end{array}
    \right.
\end{gather}

\begin{gather} 
    \label{eq:P2betam1}
    \tag*{$\PIIn{H}$}
    \left\{
    \begin{array}{lcl}
         u'
         &=& - u^2 
         + v
         - \tfrac12 z,
         \\[2mm]
         v'
         &=& v u + u v
         + \kappa_3.
    \end{array}
    \right.
\end{gather}

\begin{gather} 
    \label{eq:P1ham}
    \tag*{$\PIn{H}$}
    \left\{
    \begin{array}{lcl}
         u'
         &=& v,
         \\[2mm]
         v'
         &=& 6 u^2
         + z.
    \end{array}
    \right.
\end{gather}

\section{Special cases of \texorpdfstring{$\PIIIpr$}{P3'} type systems}\label{P3'D7}

Let us consider the scalar $\PIIIpr$ system \eqref{eq:scalP3D6'sys}.
Eliminating $v$ from the system, one arrives at the \PPainleve-$3^\prime(D_6)$ equation for $y(z) = u(z)$ of the form
\begin{gather}
    \label{eq:scalP3D6'eq}
    y''
    = \frac{1}{y} {(y')}^2
    - \frac1z y'
    + \frac{1}{z^2} y^2 \brackets{
    \gamma y + \alpha
    }
    + \frac{\beta}{z}
    + \frac{\delta }{y},
    \\[2mm]
    \notag
    \begin{aligned}
    \alpha
    &= \kappa_1 \kappa_2 - 2 \kappa_3,
    &&&
    \beta
    &= \kappa_4 (1 - \kappa_1),
    &&&
    \gamma
    &= \kappa_2^2,
    &&&
    \delta
    &= - \kappa_4^2.
    \end{aligned}
\end{gather}
It follows from these relations that the equation 
\begin{align}
        \label{eq:scalP3D72'eq}
        \tag*{$\PIIIpr(D_7)$}
        y''
        = \frac{1}{y} {(y')}^2
        - \frac1z y'
        + \frac{1}{z^2} y^2 \brackets{
        \gamma y + \alpha
        }
        + \frac{\beta}{z}
    \end{align}
of $\PIIIpr(D_7)$ type with $\delta=0, \, \beta\ne 0$ cannot be obtained from \eqref{eq:scalP3D6'sys}. However, this equation can be represented \cite[p. 12142, \Painleve III $(D_7^{(1)})$-2]{ohyama2006coalescent} as the Hamiltonian system
\begin{gather}
     \label{eq:scalP3D72'sys}
    \left\{
    \begin{array}{lclcl}
         z \, u'
         &=& 2 u^2 v + \kappa_1 u
         + \kappa_2 u^2
         ,
         \\[2mm]
         z \, v'
         &=& - 2 u v^2 - \kappa_1 v
         - 2 \kappa_2 u v
         - \kappa_3
         + \kappa_4 z u^{-2},
    \end{array}
    \right.
\end{gather}
where 
\begin{align}
    &&
    \alpha
    &= \kappa_1 \kappa_2 - 2 \kappa_3,
    &
    \beta
    &= 2 \kappa_4,
    &
    \gamma
    &= \kappa_2^2.
    &&
\end{align}
The Hamiltonian for this system is given by
\begin{equation} 
    \label{eq:scalH3D72'}
    z \, h 
    = u^2 v^2 + \kappa_1 u v 
    + \kappa_2 u^2 v
    + \kappa_3 u
    + \kappa_4 z u^{-1}.
\end{equation}
Notice that system \eqref{eq:scalP3D72'sys} has the structure
\begin{align}\label{P3D7}
    \left\{
    \begin{array}{lclcl}
         z \, u'
         &=& P_1 (u, v)
         ,
         \\[2mm]
         z \, v'
         &=& P_2 (u, v)
         + \kappa_4 z u^{-2}.
    \end{array}
    \right.
\end{align}

\subsection{Non-abelian systems of \texorpdfstring{$\PIIIpr(D_7)$}{P3'(D7)} type}
\label{sec:sysintP3D72}

\begin{proposition} \label{thm:ncJ3D72list}
There are six non-abelian systems \eqref{P3D7} with components $P_1(u, v)$ and $P_2(u, v)$ given by the formulas
\begin{align} \label{eq:ncP3D72sys}
    \begin{aligned}
    P_1 (u, v)
    &= a_1 u^2 v + (2 - a_1 - a_2) u v u + a_2 v u^2 + \kappa_1 u + \kappa_2 u^2
    ,
    \\[2mm]
    P_2 (u, v)
    &= b_1 u v^2 - (2 + b_1 + b_2) v u v + b_2 v^2 u - \kappa_1 v
    c_1 u v + (- 2 \kappa_2 - c_1) v u
    - \kappa_3
    ,
    \end{aligned}
\end{align}
whose auxiliary system
\begin{align}
    \left\{
    \begin{array}{lclcl}
         \dfrac{du}{dt}
         &=& P_1 (u, v)
         ,
         \\[2mm]
         \dfrac{dv}{dt}
         &=& P_2 (u, v)
         + \kappa_4 z u^{-2}
    \end{array}
    \right.
\end{align}
have an Okamoto integral of the form
\begin{align} \label{eq:ncJ3D72}
\begin{aligned}
    J
    = d_1 u^2 v^2 + d_2 u v^2 u + d_3 u v u v + d_4 v u^2 v + d_5 v u v u + \brackets{
    1 - \sum d_i
    } v^2 u^2
    + e_1 u v
    \\[1mm] 
    + \, (\kappa_1 - e_1) v u
    + h_1 u^2 v 
    + (\kappa_2 - h_1 - h_2) u v u
    + h_2 v u^2
    + \kappa_4 z u^{-1}
    .
\end{aligned}
\end{align}
These systems are given in the list below:
\begin{itemize}
\item Case \textbf{1}
\label{sec:P3D72case1}
\begin{gather}
    \label{eq:P3D72case1}
    \tag*{$\PPIIIprn{1}(D_7)$}
    \left\{
    \begin{array}{lcl}
         z \, u'
         &=& 2 u v u
         + \kappa_1 u
         + \kappa_2 u^2,
         \\[2mm]
         z \, v'
         &=& - 2 v u v
         - \kappa_1 v
         - 2 \kappa_2 v u - \kappa_3 
         + \kappa_4 z u^{-2}
         ,
    \end{array}
    \right.
    \\[3mm]
    \notag
    \begin{aligned}
    J
    = v u^2 v
    + \kappa_1 v u 
    + \kappa_3 \kappa_2^{-1} [u, v]
    + \kappa_2 v u^2
    + \kappa_3 u 
    + \kappa_4 z u^{-1}
    ;
    \end{aligned}
\end{gather}

\begin{gather}
    \label{eq:P3D72case2}
    \tag*{$\PPIIIprn{2}(D_7)$}
    \left\{
    \begin{array}{lcl}
         z \, u'
         &=& 2 u v u
         + \kappa_1 u
         + \kappa_2 u^2,
         \\[2mm]
         z \, v'
         &=& - 2 v u v
         - \kappa_1 v
         - 2 \kappa_2 u v - \kappa_3 
         + \kappa_4 z u^{-2},
    \end{array}
    \right.
    \\[3mm]
    \notag
    \begin{aligned}
    J
    = v u^2 v
    + \kappa_1 u v
    + \kappa_3 \kappa_2^{-1} [v, u]
    + \kappa_2 u^2 v
    + \kappa_3 u 
    + \kappa_4 z u^{-1}
    .
    \end{aligned}
\end{gather}

\item Case \textbf{2}
\label{sec:P3D72case2}
\begin{gather}
    \label{eq:P3D72case3}
    \tag*{$\PPIIIprn{3}(D_7)$}
    \left\{
    \begin{array}{lcl}
         z \, u'
         &=& u v u + v u^2
         + \kappa_1 u
         + \kappa_2 u^2
         ,
         \\[2mm]
         z \, v'
         &=& - v u v - v^2 u
         - \kappa_1 v
         - \kappa_2 u v - \kappa_2 v u
         - \kappa_3 + \kappa_4 z u^{-2},
    \end{array}
    \right.
    \\[3mm]
    \notag
    \begin{aligned}
    J
    = v u v u
    + \kappa_1 v u
    + \kappa_2 u v u
    + \kappa_3 u 
    + \kappa_4 z u^{-1}
    ;
    \end{aligned}
\end{gather}

\begin{gather}
    \label{eq:P3D72case4}
    \tag*{$\PPIIIprn{4}(D_7)$}
    \left\{
    \begin{array}{lcl}
         z \, u'
         &=& u v u + v u^2 
         + \kappa_1 u
         + \kappa_2 u^2
         ,
         \\[2mm]
         z \, v'
         &=& - v u v - v^2 u
         - \kappa_1 v
         - 2 \kappa_2 v u - \kappa_3 + \kappa_4 z u^{-2}
         ,
    \end{array}
    \right.
    \\[3mm]
    \notag
    \begin{aligned}
    J
    = v u v u
    + \kappa_1 v u 
    + \kappa_2 v u^2
    + \kappa_3 u 
    + \kappa_4 z u^{-1}
    .
    \end{aligned}
\end{gather}

\item Case \textbf{3}
\label{sec:P3D72case3}
\begin{gather}
    \label{eq:P3D72case5}
    \tag*{$\PPIIIprn{5}(D_7)$}
    \left\{
    \begin{array}{lcl}
         z \, u'
         &=& u^2 v + u v u
         + \kappa_1 u
         + \kappa_2 u^2
         ,
         \\[2mm]
         z \, v'
         &=& - u v^2 - v u v
         - \kappa_1 v
         - \kappa_2 u v - \kappa_2 v u
         - \kappa_3 + \kappa_4 z u^{-2},
    \end{array}
    \right.
    \\[3mm]
    \notag
    \begin{aligned}
    J
    = u v u v
    + \kappa_1 u v
    + \kappa_2 u v u
    + \kappa_3 u
    + \kappa_4 z u^{-1}
    ;
    \end{aligned}
\end{gather}

\begin{gather}
    \label{eq:P3D72case6}
    \tag*{$\PPIIIprn{6}(D_7)$}
    \left\{
    \begin{array}{lcl}
         z \, u'
         &=& u^2 v + u v u
         + \kappa_1 u
         + \kappa_2 u^2
         ,
         \\[2mm]
         z \, v'
         &=& - u v^2 - v u v
         - \kappa_1 v
         - 2 \kappa_2 u v - \kappa_3 + \kappa_4 z u^{-2},
    \end{array}
    \right.
    \\[3mm]
    \notag
    \begin{aligned}
    J
    = u v u v
    + \kappa_1 u v
    + \kappa_2 u^2 v
    + \kappa_3 u 
    + \kappa_4 z u^{-1}
    .
    \end{aligned}
\end{gather}
\end{itemize}
\end{proposition}

The action of the transposition \eqref{tau} on this list of systems defines three non-equivalent orbits:
\begin{gather*}
    \begin{aligned}
    \text{\bf Orbit 1} 
    &= \left\{ 
        \text{\ref{eq:P3D72case1}}, \,
        \text{\ref{eq:P3D72case2}}
    \right\},
    &&&
    \text{\bf Orbit 2} 
    &= \left\{ 
        \text{\ref{eq:P3D72case3}}, \,
        \text{\ref{eq:P3D72case5}}
    \right\},
    \end{aligned}
    \\[2mm]
    \begin{aligned}
    \text{\bf Orbit 3} 
    &= \left\{ 
        \text{\ref{eq:P3D72case4}}, \,
        \text{\ref{eq:P3D72case6}}
    \right\}.
    \end{aligned}
\end{gather*}

\subsection{Limiting transitions}
\label{sec:deg_P3}

The degenerations of non-abelian special cases of $\PIIIpr$ type systems are given in Figure \ref{pic:deg_P3}. As in Section \ref{sec:degdata}, the red arrows correspond to the representatives of the $\PVI$ orbits (see Section~\ref{subsec21}) and their degenerations.
\begin{figure}[H]
    \centering
        \scalebox{1.}{\tikzset{every picture/.style={line width=0.75pt}} 

\begin{tikzpicture}[x=0.75pt,y=0.75pt,yscale=-1,xscale=1]

\draw    (52,28) -- (79.23,28) ;
\draw [shift={(81.23,28)}, rotate = 180] [color={rgb, 255:red, 0; green, 0; blue, 0 }  ][line width=0.75]    (10.93,-3.29) .. controls (6.95,-1.4) and (3.31,-0.3) .. (0,0) .. controls (3.31,0.3) and (6.95,1.4) .. (10.93,3.29)   ;
\draw    (52,63) -- (79.23,63) ;
\draw [shift={(81.23,63)}, rotate = 180] [color={rgb, 255:red, 0; green, 0; blue, 0 }  ][line width=0.75]    (10.93,-3.29) .. controls (6.95,-1.4) and (3.31,-0.3) .. (0,0) .. controls (3.31,0.3) and (6.95,1.4) .. (10.93,3.29)   ;
\draw [color={rgb, 255:red, 243; green, 0; blue, 0 }  ,draw opacity=1 ]   (52,99) -- (79.23,99) ;
\draw [shift={(81.23,99)}, rotate = 180] [color={rgb, 255:red, 243; green, 0; blue, 0 }  ,draw opacity=1 ][line width=0.75]    (10.93,-3.29) .. controls (6.95,-1.4) and (3.31,-0.3) .. (0,0) .. controls (3.31,0.3) and (6.95,1.4) .. (10.93,3.29)   ;
\draw [color={rgb, 255:red, 243; green, 0; blue, 0 }  ,draw opacity=1 ]   (51,130) -- (79.73,149.1) ;
\draw [shift={(81.4,150.21)}, rotate = 213.61] [color={rgb, 255:red, 243; green, 0; blue, 0 }  ,draw opacity=1 ][line width=0.75]    (10.93,-3.29) .. controls (6.95,-1.4) and (3.31,-0.3) .. (0,0) .. controls (3.31,0.3) and (6.95,1.4) .. (10.93,3.29)   ;
\draw [color={rgb, 255:red, 243; green, 0; blue, 0 }  ,draw opacity=1 ]   (51,165) -- (78.23,165) ;
\draw [shift={(80.23,165)}, rotate = 180] [color={rgb, 255:red, 243; green, 0; blue, 0 }  ,draw opacity=1 ][line width=0.75]    (10.93,-3.29) .. controls (6.95,-1.4) and (3.31,-0.3) .. (0,0) .. controls (3.31,0.3) and (6.95,1.4) .. (10.93,3.29)   ;
\draw    (51,201) -- (80.12,182.1) ;
\draw [shift={(81.8,181.01)}, rotate = 147.01] [color={rgb, 255:red, 0; green, 0; blue, 0 }  ][line width=0.75]    (10.93,-3.29) .. controls (6.95,-1.4) and (3.31,-0.3) .. (0,0) .. controls (3.31,0.3) and (6.95,1.4) .. (10.93,3.29)   ;
\draw [color={rgb, 255:red, 243; green, 0; blue, 0 }  ,draw opacity=1 ]   (51,234) -- (78.23,234) ;
\draw [shift={(80.23,234)}, rotate = 180] [color={rgb, 255:red, 243; green, 0; blue, 0 }  ,draw opacity=1 ][line width=0.75]    (10.93,-3.29) .. controls (6.95,-1.4) and (3.31,-0.3) .. (0,0) .. controls (3.31,0.3) and (6.95,1.4) .. (10.93,3.29)   ;
\draw    (51,269) -- (78.23,269) ;
\draw [shift={(80.23,269)}, rotate = 180] [color={rgb, 255:red, 0; green, 0; blue, 0 }  ][line width=0.75]    (10.93,-3.29) .. controls (6.95,-1.4) and (3.31,-0.3) .. (0,0) .. controls (3.31,0.3) and (6.95,1.4) .. (10.93,3.29)   ;
\draw    (51,305) -- (78.23,305) ;
\draw [shift={(80.23,305)}, rotate = 180] [color={rgb, 255:red, 0; green, 0; blue, 0 }  ][line width=0.75]    (10.93,-3.29) .. controls (6.95,-1.4) and (3.31,-0.3) .. (0,0) .. controls (3.31,0.3) and (6.95,1.4) .. (10.93,3.29)   ;
\draw [color={rgb, 255:red, 243; green, 0; blue, 0 }  ,draw opacity=1 ]   (166,165) -- (193.23,165) ;
\draw [shift={(195.23,165)}, rotate = 180] [color={rgb, 255:red, 243; green, 0; blue, 0 }  ,draw opacity=1 ][line width=0.75]    (10.93,-3.29) .. controls (6.95,-1.4) and (3.31,-0.3) .. (0,0) .. controls (3.31,0.3) and (6.95,1.4) .. (10.93,3.29)   ;
\draw [color={rgb, 255:red, 243; green, 0; blue, 0 }  ,draw opacity=1 ]   (168,104) -- (196.73,123.1) ;
\draw [shift={(198.4,124.21)}, rotate = 213.61] [color={rgb, 255:red, 243; green, 0; blue, 0 }  ,draw opacity=1 ][line width=0.75]    (10.93,-3.29) .. controls (6.95,-1.4) and (3.31,-0.3) .. (0,0) .. controls (3.31,0.3) and (6.95,1.4) .. (10.93,3.29)   ;
\draw    (166,68) -- (194.73,87.1) ;
\draw [shift={(196.4,88.21)}, rotate = 213.61] [color={rgb, 255:red, 0; green, 0; blue, 0 }  ][line width=0.75]    (10.93,-3.29) .. controls (6.95,-1.4) and (3.31,-0.3) .. (0,0) .. controls (3.31,0.3) and (6.95,1.4) .. (10.93,3.29)   ;
\draw    (166,33) -- (194.73,52.1) ;
\draw [shift={(196.4,53.21)}, rotate = 213.61] [color={rgb, 255:red, 0; green, 0; blue, 0 }  ][line width=0.75]    (10.93,-3.29) .. controls (6.95,-1.4) and (3.31,-0.3) .. (0,0) .. controls (3.31,0.3) and (6.95,1.4) .. (10.93,3.29)   ;
\draw [color={rgb, 255:red, 243; green, 0; blue, 0 }  ,draw opacity=1 ]   (166,229) -- (195.12,210.1) ;
\draw [shift={(196.8,209.01)}, rotate = 147.01] [color={rgb, 255:red, 243; green, 0; blue, 0 }  ,draw opacity=1 ][line width=0.75]    (10.93,-3.29) .. controls (6.95,-1.4) and (3.31,-0.3) .. (0,0) .. controls (3.31,0.3) and (6.95,1.4) .. (10.93,3.29)   ;
\draw    (166,300) -- (195.12,281.1) ;
\draw [shift={(196.8,280.01)}, rotate = 147.01] [color={rgb, 255:red, 0; green, 0; blue, 0 }  ][line width=0.75]    (10.93,-3.29) .. controls (6.95,-1.4) and (3.31,-0.3) .. (0,0) .. controls (3.31,0.3) and (6.95,1.4) .. (10.93,3.29)   ;
\draw    (166,264) -- (195.12,245.1) ;
\draw [shift={(196.8,244.01)}, rotate = 147.01] [color={rgb, 255:red, 0; green, 0; blue, 0 }  ][line width=0.75]    (10.93,-3.29) .. controls (6.95,-1.4) and (3.31,-0.3) .. (0,0) .. controls (3.31,0.3) and (6.95,1.4) .. (10.93,3.29)   ;

\draw (95,18.4) node [anchor=north west][inner sep=0.75pt]    {\ref{eq:P3D72case1}};
\draw (14,18.4) node [anchor=north west][inner sep=0.75pt]    {\ref{eq:P3D6case1}};
\draw (95,54.4) node [anchor=north west][inner sep=0.75pt]    {\ref{eq:P3D72case2}};
\draw (14,54.4) node [anchor=north west][inner sep=0.75pt]    {\ref{eq:P3D6case2}};
\draw (95,90.4) node [anchor=north west][inner sep=0.75pt]    {\ref{eq:P3D72case3}};
\draw (14,90.4) node [anchor=north west][inner sep=0.75pt]    {\ref{eq:P3D6case4}};
\draw (14,120.4) node [anchor=north west][inner sep=0.75pt]    {\ref{eq:P3D6case5}};
\draw (101,156.4) node [anchor=north west][inner sep=0.75pt]    {\ref{eq:hamP3D72case1}};
\draw (20,156.4) node [anchor=north west][inner sep=0.75pt]    {\ref{eq:hamP3D6case1}};
\draw (14,192.4) node [anchor=north west][inner sep=0.75pt]    {\ref{eq:P3D6case8}};
\draw (95,224.4) node [anchor=north west][inner sep=0.75pt]    {\ref{eq:P3D72case4}};
\draw (14,224.4) node [anchor=north west][inner sep=0.75pt]    {\ref{eq:P3D6case3}};
\draw (95,260.4) node [anchor=north west][inner sep=0.75pt]    {\ref{eq:P3D72case5}};
\draw (14,260.4) node [anchor=north west][inner sep=0.75pt]    {\ref{eq:P3D6case6}};
\draw (95,296.4) node [anchor=north west][inner sep=0.75pt]    {\ref{eq:P3D72case6}};
\draw (14,296.4) node [anchor=north west][inner sep=0.75pt]    {\ref{eq:P3D6case7}};
\draw (214,156.4) node [anchor=north west][inner sep=0.75pt]    {\ref{eq:P1ham}};

\end{tikzpicture}}
    \caption{Degeneration scheme for special cases of $\PIIIpr$ type systems}
    \label{pic:deg_P3}
\end{figure}
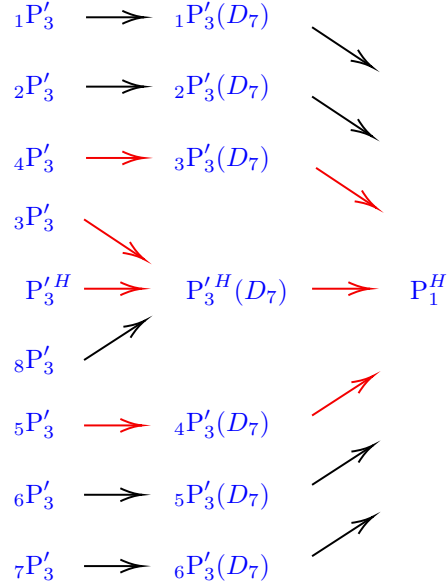

\subsubsection{
\texorpdfstring{
$\PIIIpr \to \PIIIpr(D_7)$
}{P3' -> P3'D7}
}
Systems of $\PIIIpr$ type from Appendix \ref{sec:sysintlistP3} can be reduced to $\PIIIpr(D_7)$ systems by the following limiting transition with the small parameter $\varepsilon$:
\begin{align}
    \label{eq:mapP3toP3D72}
    v
    &\mapsto v - \varepsilon^{-1} u^{-1},
    &
    \kappa_1
    &\mapsto \kappa_1 + 2 \varepsilon^{-1},
    &
    \kappa_3
    &\mapsto \kappa_3 + \varepsilon^{-1} \kappa_2,
    &
    \kappa_4
    &\mapsto \varepsilon^{-1} \kappa_4.
\end{align}
Using these formulas, we obtain the following degenerations:
\begin{gather*}
    \begin{aligned}
    \text{\rm\ref{eq:P3D6case1}}
    &\to \text{\rm\ref{eq:P3D72case1}},
    &
    \text{\rm\ref{eq:P3D6case2}}
    &\to \text{\rm\ref{eq:P3D72case2}},
    &
    \text{\rm\ref{eq:P3D6case5}}
    &\to \text{\rm\eqref{eq:P3D72case7}},
    &
    \text{\rm\ref{eq:P3D6case4}}
    &\to \text{\rm\ref{eq:P3D72case3}},
    \\[2mm]
    \text{\rm\ref{eq:P3D6case3}}
    &\to \text{\rm\ref{eq:P3D72case4}},
    &
    \text{\rm\ref{eq:P3D6case6}}
    &\to \text{\rm\ref{eq:P3D72case5}},
    &
    \text{\rm\ref{eq:P3D6case7}}
    &\to \text{\rm\ref{eq:P3D72case6}},
    &
    \text{\rm\ref{eq:P3D6case8}}
    &\to \text{\rm\eqref{eq:P3D72case8}},
    \end{aligned}
    \\[2mm]
    \text{\rm\ref{eq:hamP3D6case1}}
    \to \text{\rm\ref{eq:hamP3D72case1}}.
\end{gather*}
Неre $\PIIIprn{H}(D_7)$ is the Hamiltonian system
\begin{gather}
    \label{eq:hamP3D72case1}
    \tag*{$\PIIIprn{H}(D_7)$}
    \left\{
    \begin{array}{lcl}
         z \, u'
         &=& 2 u v u
         + \kappa_1 u
         + \kappa_2 u^2,
         \\[2mm]
         z \, v'
         &=& - 2 v u v
         - \kappa_1 v
         - \kappa_2 u v - \kappa_2 v u 
         - \kappa_3 
         + \kappa_4 z u^{-2}
         .
    \end{array}
    \right.
\end{gather}

\begin{remark}
Systems  

\begin{minipage}[l]{0.49\linewidth}
\begin{gather}
    \label{eq:P3D72case7}
    \left\{
    \begin{array}{lcr}
         z \, u'
         &=& 2 v u^2
         + \kappa_1 u
         + \kappa_2 u^2,
         \hspace{0.3cm}
         \\[2mm]
         z \, v'
         &=& - 2 v^2 u
         - \kappa_1 v
         - 2 \kappa_2 v u
         \\[1mm]
         &&
         - \, \kappa_3
         + \kappa_4 z u^{-2}
         ,
    \end{array}
    \right.
    \\[3mm]
    \notag
    \begin{aligned}
    J
    = u^{-1} (u v - v u) u
    ,
    \end{aligned}
\end{gather}
\end{minipage}
\begin{minipage}[l]{0.49\linewidth}
\begin{gather}
    \label{eq:P3D72case8}
    \left\{
    \begin{array}{lcr}
         z \, u'
         &=& 2 u^2 v
         + \kappa_1 u
         + \kappa_2 u^2,
         \hspace{0.3cm}
         \\[2mm]
         z \, v'
         &=& - 2 u v^2
         - \kappa_1 v
         - 2 \kappa_2 u v
         \\[1mm]
         &&
         - \, \kappa_3
         + \kappa_4 z u^{-2},
    \end{array}
    \right.
    \\[3mm]
    \notag
    \begin{aligned}
    J
    = u (u v - v u) u^{-1}
    ,
    \end{aligned}
\end{gather}
\end{minipage}
are equivalent to the \text{\rm\ref{eq:hamP3D72case1}} system. This equivalence is defined by the Laurent mappings $(u, v)~\mapsto~(u, u^{-1} v u)$ and $(u, v) \mapsto (u, u v u^{-1})$, respectively.
\end{remark}

Supplementing formulas \eqref{eq:mapP3toP3D72} by the following mapping
\begin{gather}
    \lambda
    \mapsto - \lambda,
    \\
    \begin{aligned}
    \mathbf{A}
    &\mapsto g \, \mathbf{A} \, g^{-1} 
    + g_{\lambda}' \, g^{-1},
    &
    \mathbf{B}
    &\mapsto g \, \mathbf{B} \, g^{-1}
    + g_z' \, g^{-1},
    &
    g
    &= \lambda^{- \varepsilon^{-1}} \,
    z^{\const \, \varepsilon^{-1}}
    \begin{pmatrix}
    1 & 0 \\ - v & 1
    \end{pmatrix}
    ,
    \end{aligned}
\end{gather}
we degenerate a pair of the form \eqref{eq:matABform_P3D6'_} for any system of $\PIIIprn{H}(D_6)$ type to a pair of the same structure for the corresponding $\PIIIpr(D_7)$ system.  

\subsubsection{
\texorpdfstring{
$\PIIIpr(D_7) \to \PI$
}{P3'D7 -> P1}
}

The $\PIIIpr(D_7)$-systems listed in Appendix \ref{sec:sysintP3D72} can be degenerated to the \PPainleve-1 system \ref{eq:P1ham} by the map
\begin{gather}
\label{eq:P3D72'toP1map}
\begin{gathered}
    \begin{aligned}
    z 
    &\mapsto \varepsilon^6 z - 2 \varepsilon^{-4}
    ,
    &&&
    u 
    &\mapsto \varepsilon^{-2} \brackets{
    u - 1
    }
    ,
    &&&
    v
    &\mapsto \varepsilon^2 v - 2 \varepsilon^{-3},
    \end{aligned}
    \\[1mm]
    \begin{aligned}
    \kappa_1
    &= \varepsilon,
    &&&
    \kappa_2
    &= - 4 \varepsilon^{-5},
    &&&
    \kappa_3
    &= 12 \varepsilon^{-10},
    &&&
    \kappa_4
    &= 2.
    \end{aligned}
\end{gathered}
\end{gather}
To get a Lax pair of the form \eqref{eq:matABform_P1} for \ref{eq:P1ham}, one may consider the degeneration data
\begin{gather}
    \lambda
    \mapsto \varepsilon^{-2} (\lambda - 1),
    \\
\begin{aligned}
    \mathbf{A}
    &\mapsto g \, \mathbf{A} \, g^{-1}
    + g_{\lambda}' \, g^{-1},
    &\mathbf{B}
    &\mapsto g \, \mathbf{B} \, g^{-1}
    + g_z' \, g^{-1}
    ,
    &
    g
    &= e^{2 \varepsilon^{-5} \lambda^{-1} 
    + \, \const \, \varepsilon^{5} z}
    \begin{pmatrix}
    2 & 0 \\ \varepsilon^{4} v & \varepsilon^4
    \end{pmatrix}.
\end{aligned}
\end{gather}

\urlstyle{same}
\bibliographystyle{plain}
\bibliography{bib}


\end{document}